\documentclass[10pt, aps, prx,twocolumn,preprintnumbers,
superscriptaddress,amsmath,amssymb,nofootinbib,nobalancelastpage,
longbibliography]{revtex4-2}
\pdfoutput=1

\usepackage[english]{babel}
\usepackage{graphicx, xcolor}
\graphicspath{{figures/}}
\usepackage{amsfonts, amsthm, microtype, mathrsfs, bbm}
\usepackage[colorlinks,allcolors=blue]{hyperref}
\usepackage{braket, mathtools, physics, enumerate}
\usepackage[capitalize]{cleveref}
\usepackage{cases}
\usepackage{thm-restate}
\usepackage[export]{adjustbox}
\usepackage{comment}
\usepackage{lipsum}
\usepackage[normalem]{ulem}

\usepackage[section]{placeins}

\usepackage{tikz}
\usetikzlibrary{arrows.meta}
\usetikzlibrary{angles, quotes}

\usetikzlibrary{positioning}

\newcommand{\Id}{\mathbbm{1}}
\newcommand{\be}{\begin{equation}}
\newcommand{\ee}{\end{equation}}
\newcommand{\bea}{\begin{eqnarray}}
\newcommand{\eea}{\end{eqnarray}}

\newtheorem{thm}{Theorem}%[section]
\newtheorem*{thm*}{Theorem}
\newtheorem{prop}{Proposition}
\newtheorem*{prop*}{Proposition}

\newtheorem*{cor*}{Corollary}
\newtheorem{lem}{Lemma}
\newtheorem{defn}{Definition}

\newtheorem{ex}{Example}

\renewcommand{\ket}[1]{|#1\rangle}
\renewcommand{\bra}[1]{\langle#1|}
\newcommand{\1}{\mathbbm{1}}

\newcommand{\gs}[1]{\textcolor{blue}{#1}}

\usepackage{tikz}
\usepackage{comment}
\usetikzlibrary{positioning}
\usetikzlibrary{decorations.pathreplacing,calligraphy,decorations.markings}

\definecolor{tensorcolor}{rgb}{0.65,0.77,0.95}
\definecolor{btensorcolor}{rgb}{0.65,0.50,0.69}
\definecolor{whitetensorcolor}{rgb}{0.93,0.93,0.93}
\definecolor{yorgosorange}{RGB}{252, 198, 3}
\definecolor{purification}{RGB}{124, 138, 128}

\newcommand\doubledx{1.6}
\newcommand\singledx{1.8}

\newcommand\stradius{0.3}

\newcommand{\HorizontalLine}[2]{
    \begin{scope}[shift={(#1)}]
        \draw[very thick] (0,0) -- ++(#2,0); 
    \end{scope}
}

\newcommand{\HorseshoeTensor}[8]{%
  \begin{scope}[shift={(#1)}]

    \def\barw{#3}
    \def\outer{#6}
    \def\height{#7}
    \def\radius{0.09cm}
    \def\legx{\outer-0.5} % New x-position for vertical legs

    \begin{scope}[shift={(-0.5*\barw,-0.5*\height)}]
      \path[fill=tensorcolor, draw=black, line width=0.7pt, rounded corners=\radius]
        (0,0) -- (\outer,0)
        -- (\outer,\barw)
        -- (\barw,\barw)
        -- (\barw,\height - \barw)
        -- (\outer,\height - \barw)
        -- (\outer,\height)
        -- (0,\height)
        -- cycle;
    \end{scope}

        \draw[very thick] (\legx,0.5*\height) -- ++(0,#2);
        \draw[very thick] (\legx,-0.5*\height) -- ++(0,-#2);

        \draw[very thick] (\legx,0.5*\height-\barw) -- ++(0,-#2);
        \draw[very thick] (\legx,-0.5*\height+\barw) -- ++(0,#2);

    \ifnum#5=0 %
        \draw[very thick] (\barw/2,0) -- ++(#2,0);
        \draw[very thick] (-\barw/2,0) -- ++(-#2,0);
    \fi

    \ifnum#5=1 %
        \draw[very thick] (\barw/2,0) -- ++(#2,0);
    \fi

    \ifnum#5=-1 %
        \draw[very thick] (-\barw/2,0) -- ++(-#2,0);
    \fi

    \ifnum#5=2 %
        \draw[very thick] (\barw/2,0.1) -- ++(#2,0);
        \draw[very thick] (-\barw/2,-0.1) -- ++(-#2,0);
        \draw[very thick] (-\barw/2,0.1) -- ++(-#2,0);
        \draw[very thick] (\barw/2,-0.1) -- ++(#2,0);
    \fi

        \ifnum#5=3 %
    \fi

    \ifnum#8=1
        \pgfmathsetmacro{\centerY}{0}
        \pgfmathsetmacro{\halfHeight}{0.5*\height}
        \pgfmathsetmacro{\offsetY}{0.5*\barw}
        
        \node at (-\barw-0.2, 0.2) {$m$};                  % left leg
        \node at (\barw+0.1, 0.2) {$n$};                  % right leg
        \node at (\legx+0.2, \halfHeight - \offsetY - 0.4) {$i$};  % top inner
        \node at (\legx+0.2, -\halfHeight + \offsetY + 0.5) {$j$}; % bottom inner
        \node at (\legx+0.2, \halfHeight + \offsetY + 0.1) {$o$};      % top outer
        \node at (\legx+0.2, -\halfHeight - \offsetY - 0.1) {$p$};     % bottom outer

        \fi

    \node at (-\barw,-0.5*\height-\barw) {#4};
    
  \end{scope}
}

\newcommand{\HorseshoeTensorForChannelAction}[8]{%
  \begin{scope}[shift={(#1)}]

    \def\barw{#3}
    \def\outer{#6}
    \def\height{#7}
    \def\radius{0.09cm}
    \def\legx{\outer-0.5} % New x-position for vertical legs

    \begin{scope}[shift={(-0.5*\barw,-0.5*\height)}]
      \path[fill=tensorcolor, draw=black, line width=0.7pt, rounded corners=\radius]
        (0,0) -- (\outer,0)
        -- (\outer,\barw)
        -- (\barw,\barw)
        -- (\barw,\height - \barw)
        -- (\outer,\height - \barw)
        -- (\outer,\height)
        -- (0,\height)
        -- cycle;
    \end{scope}

        \draw[very thick] (\legx,0.5*\height) -- ++(0,#2);
        \draw[very thick] (\legx,-0.5*\height) -- ++(0,-#2);

        \draw[very thick] (\legx,0.5*\height-\barw) -- ++(0,-#2);
        \draw[very thick] (\legx,-0.5*\height+\barw) -- ++(0,#2);

        \draw[very thick] (\legx,0.5*\height-\barw) -- (\legx,-0.5*\height+\barw);
        \draw[very thick] (\legx - #3 - #2,0) -- (\legx + #3 + #2,0);

    \ifnum#5=0 %
        \draw[very thick] (\barw/2,\height/4) -- ++(#2,0);
        \draw[very thick] (-\barw/2,\height/4) -- ++(-#2,0);
    \fi

    \ifnum#5=1 %
        \draw[very thick] (\barw/2,0) -- ++(#2,0);
    \fi

    \ifnum#5=-1 %
        \draw[very thick] (-\barw/2,0) -- ++(-#2,0);
    \fi

    \ifnum#8=1
        \pgfmathsetmacro{\centerY}{0}
        \pgfmathsetmacro{\halfHeight}{0.5*\height}
        \pgfmathsetmacro{\offsetY}{0.5*\barw}
        
        \node at (-\barw-0.2, 0.2) {$m$};                  % left leg
        \node at (\barw+0.1, 0.2) {$n$};                  % right leg
        \node at (\legx+0.2, \halfHeight - \offsetY - 0.4) {$i$};  % top inner
        \node at (\legx+0.2, -\halfHeight + \offsetY + 0.5) {$j$}; % bottom inner
        \node at (\legx+0.2, \halfHeight + \offsetY + 0.1) {$o$};      % top outer
        \node at (\legx+0.2, -\halfHeight - \offsetY - 0.1) {$p$};     % bottom outer

        \fi

    \node at (-\barw,-0.5*\height-\barw) {\scriptsize #4};
    \pgfmathsetmacro{\squarehalf}{0.35}
    \draw[ thick, fill=yorgosorange, rounded corners=2pt] (-\squarehalf + \legx,-\squarehalf) rectangle (\squarehalf  + \legx ,\squarehalf);
    \draw (\legx,0) node {\scriptsize $G_k$};
    
  \end{scope}
}

\newcommand{\IsomTensor}[5]{
    \begin{scope}[shift={(#1)}]
        \pgfmathsetmacro{\topY}{#3}     % top of the box
        \pgfmathsetmacro{\bottomY}{-#3} % bottom of the box
        
        \draw[very thick] (#3/4,\topY) -- ++(0,#2);
        \draw[very thick, color=purification] (-#3/4,\topY) -- ++(0,#2);
        
        \draw[very thick] (0,\bottomY) -- ++(0,-#2);
    
        \ifnum#5=0
            \draw[very thick] (-#3,0) -- +(-#2,0);  % Left leg from box edge
            \draw[very thick] (#3,0) -- +(#2,0);    % Right leg from box edge
        \fi
        
        \ifnum#5=-1
            \draw[very thick] (#3,0) -- +(#2,0);    % Right leg from box edge
        \fi
        
        \ifnum#5=1
            \draw[very thick] (-#3,0) -- +(-#2,0);  % Left leg from box edge
        \fi
    
        \draw[ thick, fill=tensorcolor, rounded corners=2pt] (-#3,-#3) rectangle (#3,#3);
        \draw (0,0) node {\scriptsize #4};
    \end{scope}
}

\newcommand{\IsomConjTensor}[5]{
    \begin{scope}[shift={(#1)}]

        \pgfmathsetmacro{\topY}{#3}     % top of the box
        \pgfmathsetmacro{\bottomY}{-#3} % bottom of the box
        
        \draw[very thick] (#3/4,\bottomY) -- ++(0,-#2);
        \draw[very thick, color=purification] (-#3/4,\bottomY) -- ++(0,-#2);
        
        \draw[very thick] (0,\topY) -- ++(0,#2);
    
        \ifnum#5=0
            \draw[very thick] (-#3,0) -- +(-#2,0);  % Left leg
            \draw[very thick] (#3,0) -- +(#2,0);    % Right leg
        \fi
    
        \ifnum#5=-1
            \draw[very thick] (#3,0) -- +(#2,0);    % Right leg only
        \fi
    
        \ifnum#5=1
            \draw[very thick] (-#3,0) -- +(-#2,0);  % Left leg only
        \fi

        \draw[ thick, fill=tensorcolor, rounded corners=2pt] (-#3,-#3) rectangle (#3,#3);
        \draw (0,0) node {\scriptsize \ensuremath{\overline{\expandafter\stripdollar#4}}};

    \end{scope}
}

\newcommand{\IsomTensorPermuted}[5]{
    \begin{scope}[shift={(#1)}]
        \pgfmathsetmacro{\lateralpush}{#3 + #2/2}
        \pgfmathsetmacro{\topY}{#3}     % top of the box
        \pgfmathsetmacro{\bottomY}{-#3} % bottom of the box
        
        \draw[very thick] (#3/4,\topY) -- ++(0,#2);
        \draw[very thick, color=purification] (-#3/4,\topY) -- ++(0,#2);
        \draw [very thick, color=purification] (-#3/4,\topY + #2) to  [bend right=90] (-#3/4-\lateralpush, \topY + #2);
        \draw [very thick, color=purification] (-#3/4-\lateralpush,\topY + #2) to (-#3/4-\lateralpush, \bottomY - #2);
        
        \draw[very thick] (0,\bottomY) -- ++(0,-#2);
        \draw [very thick] (0,\bottomY - #2) to  [bend right=90] (\lateralpush, \bottomY - #2);
        \draw [very thick] (\lateralpush,\bottomY - #2) to (\lateralpush, \topY + #2);
    
        \ifnum#5=0
            \draw[very thick] (-#3,0) -- +(-#2,0);  % Left leg from box edge
            \draw[very thick] (#3,0) -- +(#2,0);    % Right leg from box edge
        \fi
        
        \ifnum#5=-1
            \draw[very thick] (#3,0) -- +(#2,0);    % Right leg from box edge
        \fi
        
        \ifnum#5=1
            \draw[very thick] (-#3,0) -- +(-#2,0);  % Left leg from box edge
        \fi
    
        \draw[ thick, fill=tensorcolor, rounded corners=2pt] (-#3,-#3) rectangle (#3,#3);
        \draw (0,0) node {\scriptsize #4};
    \end{scope}
}

\def\stripdollar$#1${\def\temp{#1}\temp}  % Helper to remove $ signs

\newcommand{\LPTensor}[6]{
	\begin{scope}[shift={(#1)}]

    \def\distancefromcenter{#6}
    \IsomTensor{(0,\distancefromcenter)}{#2}{#3}{#4}{#5};
    \IsomConjTensor{(0,-\distancefromcenter)}{#2}{#3}{#4}{#5};

    \pgfmathsetmacro{\dx}{#3 + #2 / 2};
	\draw [very thick, color=purification] (-#3/4,-\distancefromcenter -#2-#3) to  [bend left=90] (-\dx,-\distancefromcenter -#2-#3);
	\draw [very thick, color=purification] (-#3/4,\distancefromcenter +#2+#3) to  [bend right=90] (-\dx,\distancefromcenter +#2+#3);
	\draw [very thick, color=purification] (-\dx,\distancefromcenter +#2+#3) to  (-\dx,-\distancefromcenter -#2 -#3);
        
	\end{scope}
}

\newcommand{\SingleSiteDouble}[5]{
	\begin{scope}[shift={(#1)}]
        \IsomTensor{(0,-#2-#3)}{#2}{#3}{#4}{#5};
        \IsomConjTensor{(0,#2+#3)}{#2}{#3}{#4}{#5};    
	\end{scope}
}

\newcommand{\TransferMatrixTensor}[5]{
	\begin{scope}[shift={(#1)}]
        \SingleSiteDouble{(0,0)}{#2}{#3}{#4}{#5}
        
        \pgfmathsetmacro{\topY}{2*#3+#2}     % top of the box
        \pgfmathsetmacro{\bottomY}{-2*#3-#2}
        \def\dx{.75};
        \draw [very thick] (0,\topY + #2) to  [bend right=90] (-\dx,\topY + #2);
        \draw [very thick] (0,\bottomY - #2) to  [bend left=90] (-\dx, \bottomY - #2);
        \draw [very thick] (-\dx,\topY + #2) to  (-\dx, \bottomY -#2);   
    
	\end{scope}
}

\newcommand{\SMatrixTensor}[6]{
	\begin{scope}[shift={(#1)}]
        \SingleSiteDouble{(0,0)}{#2}{#3}{#4}{#5}
        \filldraw[color=black, fill=whitetensorcolor, thick] (0,2*#2+2*#3 + \stradius) circle (\stradius);

        \node at (0,2*#2+2*#3 + \stradius) {\tiny  \ensuremath{\expandafter\stripdollar#6}};

        \pgfmathsetmacro{\topY}{2*#3+#2}     % top of the box
        \pgfmathsetmacro{\bottomY}{-2*#3-#2}
        \def\dx{.75};
        \draw [very thick] (0,\topY + #2 + 2*\stradius) to  [bend right=90] (-\dx,\topY + 2*\stradius+ #2);
        \draw [very thick] (0,\bottomY - #2) to  [bend left=90] (-\dx, \bottomY - #2);
        \draw [very thick] (-\dx,\topY + #2+ 2*\stradius) to  (-\dx, \bottomY -#2);   
    
	\end{scope}
}

\newcommand{\SMatrixTensorMiddleSigma}[8]{
	\begin{scope}[shift={(#1)}]
        \SingleSiteDouble{(0,0)}{#2}{#3}{#4}{#5}

        \pgfmathsetmacro{\topY}{2*#3+#2}     % top of the box
        \pgfmathsetmacro{\bottomY}{-2*#3-#2}
        \pgfmathsetmacro{\dx}{#8}
        \draw [very thick] (0,\topY + #2) to  [bend right=90] (-\dx,\topY + #2);
        \draw [very thick] (0,\bottomY - #2) to  [bend left=90] (-\dx, \bottomY - #2);
        \draw [very thick] (-\dx,\topY + #2) to  (-\dx, \bottomY -#2);    

        \filldraw[color=black, fill=whitetensorcolor, thick] (-\dx,0) circle (#7);
        \node at (-\dx,0) {\scriptsize \ensuremath{\expandafter\stripdollar#6}};
    
	\end{scope}
}

\newcommand{\TracedTensor}[6]{
	\begin{scope}[shift={(#1)}]

        \pgfmathsetmacro{\topY}{2*#3+#2}     % top of the box
        \pgfmathsetmacro{\bottomY}{-2*#3-#2}
        \def\dx{.75};
        \draw [very thick] (0,\topY + #2) to  [bend right=90] (-\dx,\topY + #2);
        \draw [very thick] (0,\bottomY - #2) to  [bend left=90] (-\dx, \bottomY - #2);
        \draw [very thick] (-\dx,\topY + #2) to  (-\dx, \bottomY -#2);
        \draw [very thick] (-0,\topY + #2) to  (-0, \bottomY -#2);

        \filldraw[color=black, fill=whitetensorcolor, thick] (-\dx,0) circle (\stradius);
        \node at (-\dx,0) {\scriptsize \ensuremath{\expandafter\stripdollar#6}};
    
	\end{scope}
}

\newcommand{\RightEigen}[4]{
    \def\legs{#2}
    \def\circleradius{#3}
    \def\borderradius{0.09cm}

	\begin{scope}[shift={(#1)}]
    \begin{scope}[shift={(0, -\legs-\circleradius)}]
        \path[draw=black, very thick, rounded corners=\borderradius]
            (0,0) -- (\legs,0)
            -- (\legs,2*\legs+2*\circleradius)
            -- (0,2*\legs+2*\circleradius);
    \end{scope}

    \filldraw[color=black, fill=whitetensorcolor, thick] (\legs,0) circle (\circleradius);
    \node at (\legs,0) {\scriptsize #4};
    
	\end{scope}
}

\newcommand{\LeftEigen}[4]{
    \def\legs{#2}
    \def\circleradius{#3}
    \def\borderradius{0.09cm}

	\begin{scope}[shift={(#1)}]
    \begin{scope}[shift={(0, -\legs-\circleradius)}]
        \path[draw=black, very thick, rounded corners=\borderradius]
            (0,0) -- (-\legs,0)
            -- (-\legs,2*\legs+2*\circleradius)
            -- (0,2*\legs+2*\circleradius);
    \end{scope}

    \filldraw[color=black, fill=whitetensorcolor, thick] (-\legs,0) circle (\circleradius);
    \node at (-\legs,0) {\scriptsize #4};
    
	\end{scope}
}

\newcommand{\LeftEigenStripped}[3]{
    \def\legs{#2}
    \def\circleradius{#3}
    \def\borderradius{0.09cm}

	\begin{scope}[shift={(#1)}]
    \begin{scope}[shift={(0, -\legs-\circleradius)}]
        \path[draw=black, very thick, rounded corners=\borderradius]
            (0,0) -- (-\legs,0)
            -- (-\legs,2*\legs+2*\circleradius)
            -- (0,2*\legs+2*\circleradius);
    \end{scope}
    
	\end{scope}
}

\newcommand{\IdentityLine}[2]{
    \pgfmathsetmacro{\linelength}{#2}

	\begin{scope}[shift={(#1)}]
    \draw[color=black,very thick] (0,-\linelength / 2) -- (0,\linelength / 2);
	\end{scope}
}

\newcommand{\SingleTrLeft}[1]{
	\begin{scope}[shift={(#1)}]
      \draw [very thick] (0,0) to (\doubledx-0.8,0);
	   \draw [very thick] (0,0) to  [bend left=90] (0,0.8);
	   \draw [very thick, dotted] (0,0.8) to  (1,0.8);
	\end{scope}
}

\newcommand{\SingleDots}[2]{
	\begin{scope}[shift={(#1)}]
        \pgfmathsetmacro{\halflength}{#2/2}
        \draw [very thick, dotted] (-\halflength,0) -- (\halflength,0);
	\end{scope}
}

\newcommand{\DoubleDots}[3]{
	\begin{scope}[shift={(#1)}]
      \SingleDots{0,#3}{#2};
      \SingleDots{0,-#3}{#2};
	\end{scope}
}

\newcommand{\SingleTrRight}[1]{
	\begin{scope}[shift={(#1)}, xscale=-1]
	   \SingleTrLeft{(0,0)};
	\end{scope}
}

\newcommand{\BipartiteHorseshoeTensor}[8]{%
  \begin{scope}[shift={(#1)}]

    \def\barw{#3}
    \def\outer{#6}
    \def\height{#7}
    \def\radius{0.09cm}
    \def\legx{\outer-0.5} % New x-position for vertical legs
    \def\legxsec{\outer-1.0} % New x-position for vertical legs

    \begin{scope}[shift={(-0.5*\barw,-0.5*\height)}]
      \path[fill=tensorcolor, draw=black, line width=0.7pt, rounded corners=\radius]
        (0,0) -- (\outer,0)
        -- (\outer,\barw)
        -- (\barw,\barw)
        -- (\barw,\height - \barw)
        -- (\outer,\height - \barw)
        -- (\outer,\height)
        -- (0,\height)
        -- cycle;
    \end{scope}

        \draw[very thick] (\legx,0.5*\height) -- ++(0,#2);
        \draw[very thick] (\legx,-0.5*\height) -- ++(0,-#2);

        \draw[very thick] (\legxsec,0.5*\height) -- ++(0,#2);
        \draw[very thick] (\legxsec,-0.5*\height) -- ++(0,-#2);

        \draw[very thick] (\legx,0.5*\height-\barw) -- ++(0,-#2);
        \draw[very thick] (\legx,-0.5*\height+\barw) -- ++(0,#2);

        \draw[very thick] (\legxsec,0.5*\height-\barw) -- ++(0,-#2);
        \draw[very thick] (\legxsec,-0.5*\height+\barw) -- ++(0,#2);

    \ifnum#8=1
        \pgfmathsetmacro{\centerY}{0}
        \pgfmathsetmacro{\halfHeight}{0.5*\height}
        \pgfmathsetmacro{\offsetY}{0.5*\barw}
        
        \node at (\legx+0.3, \halfHeight - \offsetY - 0.4) {\scriptsize $i_2$};  % top inner
        \node at (\legx+0.3, -\halfHeight + \offsetY + 0.5) {\scriptsize $i_2$}; % bottom inner
        \node at (\legx+0.3, \halfHeight + \offsetY + 0.1) {\scriptsize $o_2$};      % top outer
        \node at (\legx+0.3, -\halfHeight - \offsetY - 0.1) {\scriptsize $o_2$};     % bottom outer

        \node at (\legxsec+0.3, \halfHeight - \offsetY - 0.4) {\scriptsize $i_1$};  % top inner
        \node at (\legxsec+0.3, -\halfHeight + \offsetY + 0.5) {\scriptsize $i_1$}; % bottom inner
        \node at (\legxsec+0.3, \halfHeight + \offsetY + 0.1) {\scriptsize $o_1$};      % top outer
        \node at (\legxsec+0.3, -\halfHeight - \offsetY - 0.1) {\scriptsize $o_1$};     % bottom outer

        \fi

    \node at (-\barw,-0.5*\height-\barw) {\scriptsize #4};
    
  \end{scope}
}

\newcommand{\BipartiteVectorized}[6]{
    \begin{scope}[shift={(#1)}]
        
        \pgfmathsetmacro{\topY}{#3}     % top of the box
        \pgfmathsetmacro{\bottomY}{-#3} % bottom of the box
        \pgfmathsetmacro{\legx}{#5 - 0.2} % bottom of the box
        \pgfmathsetmacro{\legxsec}{#5 - 0.6} % bottom of the box

        \draw[very thick] (\legx,\topY) -- ++(0,#2);
        \draw[very thick] (\legxsec,\topY) -- ++(0,#2);
        \draw[very thick] (-\legx,\topY) -- ++(0,#2);
        \draw[very thick] (-\legxsec,\topY) -- ++(0,#2);
        
        \draw[very thick] (\legx,\bottomY) -- ++(0,-#2);
        \draw[very thick] (\legxsec,\bottomY) -- ++(0,-#2);
        \draw[very thick] (-\legx,\bottomY) -- ++(0,-#2);
        \draw[very thick] (-\legxsec,\bottomY) -- ++(0,-#2);

        \ifnum#6=1
        
            \pgfmathsetmacro{\Yheight}{\topY + #2 + 0.3}

            \node at (\legx, -\Yheight) {\scriptsize $i_2$};  % top inner
            \node at (\legxsec, -\Yheight ) {\scriptsize $i_2$}; % bottom inner
            \node at (\legx, \Yheight) {\scriptsize $o_2$};      % top outer
            \node at (\legxsec, \Yheight) {\scriptsize $o_2$};     % bottom outer
    
            \node at (-\legxsec, -\Yheight) {\scriptsize $i_1$}; 
            \node at (-\legx, -\Yheight) {\scriptsize $i_1$};
            \node at (-\legxsec, \Yheight) {\scriptsize $o_1$};   
            \node at (-\legx,\Yheight) {\scriptsize $o_1$};
            
            \fi
        
        \draw[ thick, fill=tensorcolor, rounded corners=2pt] (-#5,-#3) rectangle (#5,#3);
        \draw (0,0) node {\scriptsize #4};
    \end{scope}
}

\newcommand{\SVDBipartite}[6]{
    \begin{scope}[shift={(#1)}]
        
        \pgfmathsetmacro{\topY}{#3}     % top of the box
        \pgfmathsetmacro{\bottomY}{-#3} % bottom of the box
        \pgfmathsetmacro{\legx}{#5 - 0.2} % bottom of the box
        \pgfmathsetmacro{\legxsec}{#5 - 0.6} % bottom of the box

        \draw[very thick] (\legx,\topY) -- ++(0,#2);
        \draw[very thick] (\legxsec,\topY) -- ++(0,#2);
        \draw[very thick] (-\legx,\topY) -- ++(0,#2);
        \draw[very thick] (-\legxsec,\topY) -- ++(0,#2);
        
        \draw[very thick] (\legx,\bottomY) -- ++(0,-#2);
        \draw[very thick] (\legxsec,\bottomY) -- ++(0,-#2);
        \draw[very thick] (-\legx,\bottomY) -- ++(0,-#2);
        \draw[very thick] (-\legxsec,\bottomY) -- ++(0,-#2);

        \ifnum#6=1
        
            \pgfmathsetmacro{\Yheight}{\topY + #2 + 0.3}

            \node at (\legx, -\Yheight) {\scriptsize $i_2$};  % top inner
            \node at (\legxsec, -\Yheight ) {\scriptsize $i_2$}; % bottom inner
            \node at (\legx, \Yheight) {\scriptsize $o_2$};      % top outer
            \node at (\legxsec, \Yheight) {\scriptsize $o_2$};     % bottom outer
    
            \node at (-\legxsec, -\Yheight) {\scriptsize $i_1$}; 
            \node at (-\legx, -\Yheight) {\scriptsize $i_1$};
            \node at (-\legxsec, \Yheight) {\scriptsize $o_1$};   
            \node at (-\legx,\Yheight) {\scriptsize $o_1$};
            
            \fi
            
        \draw[very thick] (-#5,0) to (#5,0);
        \draw[ thick, fill=tensorcolor, rounded corners=2pt] (-#5,-#3) rectangle (-#5 + 0.8,#3);
        \draw[ thick, fill=tensorcolor, rounded corners=2pt] (#5 - 0.8,-#3) rectangle (#5,#3);
        \draw (0,0.4) node {\small $D$};
    \end{scope}
}

\newcommand\subsetsim{\mathrel{%
  \ooalign{\raise0.2ex\hbox{$\subset$}\cr\hidewidth\raise-0.8ex\hbox{\scalebox{0.9}{$\sim$}}\hidewidth\cr}}}

  \newcommand{\GTensor}[5]{
	\begin{scope}[shift={(#1)}]
    \ifnum#5=0
		\draw[very thick] (-#2,0) -- (#2,0);
		\draw[very thick] (0,#2) -- (0,-#2);
    \fi
    \ifnum#5=-1
		\draw[very thick] (0,0) -- (#2,0);
		\draw[very thick] (0,#2) -- (0,-#2);
    \fi
    \ifnum#5=1
		\draw[very thick] (-#2,0) -- (0,0);
		\draw[very thick] (0,#2) -- (0,-#2);
    \fi

    \ifnum#5=2
    \fi

    \ifnum#5=3 % MPS
        \draw[very thick] (-#2,0) -- (#2,0);
        \draw[very thick] (0,#2) -- (0,0);
    \fi
        \draw[ thick, fill=tensorcolor, rounded corners=2pt] (-#3,-#3) rectangle (#3,#3);
		\draw (0,0) node {\scriptsize #4};
	\end{scope}
}

  \newcommand{\GTensorOrange}[5]{
	\begin{scope}[shift={(#1)}]
    \ifnum#5=0
		\draw[very thick] (-#2,0) -- (#2,0);
		\draw[very thick] (0,#2) -- (0,-#2);
    \fi
    \ifnum#5=-1
		\draw[very thick] (0,0) -- (#2,0);
		\draw[very thick] (0,#2) -- (0,-#2);
    \fi
    \ifnum#5=1
		\draw[very thick] (-#2,0) -- (0,0);
		\draw[very thick] (0,#2) -- (0,-#2);
    \fi
        \draw[ thick, fill=yorgosorange, rounded corners=2pt] (-#3,-#3) rectangle (#3,#3);
		\draw (0,0) node {\scriptsize #4};
	\end{scope}
}

\newcommand{\GDTensor}[5]{
	\begin{scope}[shift={(#1)}]
    \ifnum#5=0
		\draw[very thick] (-#2,0) -- (#2,0);
		\draw[very thick] (0,#2) -- (0,-#2);
    \fi
    \ifnum#5=-1
		\draw[very thick] (0,0) -- (#2,0);
		\draw[very thick] (0,#2) -- (0,-#2);
    \fi
    \ifnum#5=1
		\draw[very thick] (-#2,0) -- (0,0);
		\draw[very thick] (0,#2) -- (0,-#2);
    \fi
        \draw[ thick, fill=tensorcolor, rounded corners=2pt] (-#3,-#3) rectangle (#3,#3);
        \draw (0,0) node {\scriptsize \ensuremath{\overline{\expandafter\stripdollar#4}}};
	\end{scope}
}

\newcommand{\MPSTensor}[5]{
	\begin{scope}[shift={(#1)}]
    \ifnum#5=0
		\draw[very thick] (-#2 - #3,0) -- (#2 + #3,0);
		\draw[very thick] (0,#2 + #3) -- (0,0);
    \fi
    \ifnum#5=-1
		\draw[very thick] (0,0) -- (#2 + #3,0);
		\draw[very thick] (0,#2 + #3) -- (0,0);
    \fi
    \ifnum#5=1
		\draw[very thick] (-#2 - #3,0) -- (0,0);
		\draw[very thick] (0,#2 + #3) -- (0,0);
    \fi

    \ifnum#5=2
    \fi
        \draw[ thick, fill=tensorcolor, rounded corners=2pt] (-#3,-#3) rectangle (#3,#3);
		\draw (0,0) node {\scriptsize #4};
	\end{scope}
}

\newcommand{\GHZTensor}[4]{
	\begin{scope}[shift={(#1)}]
    \ifnum#4=0
		\draw[very thick] (-#2 - #3,0) -- (#2 + #3,0);
		\draw[very thick] (0,#2 + #3) -- (0,0);
    \fi
    \ifnum#4=-1
		\draw[very thick] (0,0) -- (#2 + #3,0);
		\draw[very thick] (0,#2 + #3) -- (0,0);
    \fi
    \ifnum#4=1
		\draw[very thick] (-#2 - #3,0) -- (0,0);
		\draw[very thick] (0,#2 + #3) -- (0,0);
    \fi

    \ifnum#4=2
    \fi

    \filldraw[black] (0,0) circle (4pt) node{};

	\end{scope}
}

\newcommand{\MPOTensor}[5]{
	\begin{scope}[shift={(#1)}]
    \ifnum#5=0
		\draw[very thick] (-#2 - #3,0) -- (#2 + #3,0);
		\draw[very thick] (0,#2 + #3) -- (0, - #2 - #3);
    \fi
    \ifnum#5=-1
		\draw[very thick] (0,0) -- (#2 + #3,0);
		\draw[very thick] (0,#2 + #3) -- (0, - #2 - #3);
    \fi
    \ifnum#5=1
		\draw[very thick] (-#2 - #3,0) -- (0,0);
		\draw[very thick] (0,#2 + #3) -- (0, - #2 - #3);
    \fi

        \draw[ thick, fill=tensorcolor, rounded corners=2pt] (-#3,-#3) rectangle (#3,#3);
		\draw (0,0) node {\scriptsize #4};
	\end{scope}
}

\newcommand{\MPOTensorOrange}[5]{
	\begin{scope}[shift={(#1)}]
    \ifnum#5=0
		\draw[very thick] (-#2 - #3,0) -- (#2 + #3,0);
		\draw[very thick] (0,#2 + #3) -- (0, - #2 - #3);
    \fi
    \ifnum#5=-1
		\draw[very thick] (0,0) -- (#2 + #3,0);
		\draw[very thick] (0,#2 + #3) -- (0, - #2 - #3);
    \fi
    \ifnum#5=1
		\draw[very thick] (-#2 - #3,0) -- (0,0);
		\draw[very thick] (0,#2 + #3) -- (0, - #2 - #3);
    \fi

        \draw[ thick, fill=yorgosorange, rounded corners=2pt] (-#3,-#3) rectangle (#3,#3);
		\draw (0,0) node {\scriptsize #4};
	\end{scope}
}

\newcommand{\MPOConjTensor}[5]{
	\begin{scope}[shift={(#1)}]
    \ifnum#5=0
		\draw[very thick] (-#2 - #3,0) -- (#2 + #3,0);
		\draw[very thick] (0,#2 + #3) -- (0, - #2 - #3);
    \fi
    \ifnum#5=-1
		\draw[very thick] (0,0) -- (#2 + #3,0);
		\draw[very thick] (0,#2 + #3) -- (0, - #2 - #3);
    \fi
    \ifnum#5=1
		\draw[very thick] (-#2 - #3,0) -- (0,0);
		\draw[very thick] (0,#2 + #3) -- (0, - #2 - #3);
    \fi

        \draw[ thick, fill=tensorcolor, rounded corners=2pt] (-#3,-#3) rectangle (#3,#3);
		\draw (0,0) node {\scriptsize \ensuremath{\overline{\expandafter\stripdollar#4}}};
	\end{scope}
}

\newcommand{\DoubleMPOTensor}[5]{
	\begin{scope}[shift={(#1)}]
        \MPOConjTensor{(0,#2+#3)}{#2}{#3}{#4}{#5};
        \MPOTensor{(0,-#2-#3)}{#2}{#3}{#4}{#5};
	\end{scope}
}

\newcommand{\MPOTensorDoubleLegs}[5]{
    \begin{scope}[shift={(#1)}]
        \pgfmathsetmacro{\topY}{#3 + #2}     % top of the box
        \pgfmathsetmacro{\bottomY}{-\topY} % bottom of the box

        \draw[very thick] (#3/2,\topY) -- (#3/2,\bottomY);
        \draw[very thick] (-#3/2,\topY) -- (- #3/2,\bottomY);
    
        \ifnum#5=0
            \draw[very thick] (-#3,0) -- +(-#2,0);  % Left leg from box edge
            \draw[very thick] (#3,0) -- +(#2,0);    % Right leg from box edge
        \fi
        
        \ifnum#5=-1
            \draw[very thick] (#3,0) -- +(#2,0);    % Right leg from box edge
        \fi
        
        \ifnum#5=1
            \draw[very thick] (-#3,0) -- +(-#2,0);  % Left leg from box edge
        \fi
    
        \draw[ thick, fill=tensorcolor, rounded corners=2pt] (-#3,-#3) rectangle (#3,#3);
        \draw (0,0) node {\scriptsize #4};
    \end{scope}
}

\newcommand{\MPOTensorDoubleLegsWithInputGHZ}[5]{
    \begin{scope}[shift={(#1)}]
        \pgfmathsetmacro{\topY}{#3 + #2}     % top of the box
        \pgfmathsetmacro{\bottomY}{-\topY} % bottom of the box

        \GHZTensor{#3/2,\bottomY}{#2}{#3}{0} 
        
        \draw[very thick] (#3/2,\topY) -- (#3/2,\bottomY);
        \draw[very thick] (-#3/2,0) -- (- #3/2,2*\bottomY);
        \draw[very thick, color=purification] (-#3/2,\topY) -- (- #3/2,0);
    
        \ifnum#5=0
            \draw[very thick] (-#3,0) -- +(-#2,0);  % Left leg from box edge
            \draw[very thick] (#3,0) -- +(#2,0);    % Right leg from box edge
        \fi
        
        \ifnum#5=-1
            \draw[very thick] (#3,0) -- +(#2,0);    % Right leg from box edge
        \fi
        
        \ifnum#5=1
            \draw[very thick] (-#3,0) -- +(-#2,0);  % Left leg from box edge
        \fi
    
        \draw[ thick, fill=tensorcolor, rounded corners=2pt] (-#3,-#3) rectangle (#3,#3);
        \draw (0,0) node {\scriptsize #4};
    \end{scope}
}

\newcommand{\myarrow}[2]{
	\begin{scope}[shift={(#1)}]
    \ifnum#2=1
\draw[-{Stealth[length=1mm, width=2.3mm]}] (0,0.0) -- (0,0.03);
    \fi
    \ifnum#2=2
\draw[-{Stealth[length=1mm, width=2.3mm]}] (0,0) -- (0,-0.03);
    \fi
    \ifnum#2=3
\draw[-{Stealth[length=1mm, width=2.3mm]}] (0,0) -- (0.03,0);
    \fi
    \ifnum#2=4
\draw[-{Stealth[length=1mm, width=2.3mm]}] (0,0) -- (-0.03,0);
    \fi
\end{scope}
}

\newcommand{\IsomTensorMap}[4]{
    \begin{scope}[shift={(#1)}]
        \pgfmathsetmacro{\topY}{#3}     % top of the box
        \pgfmathsetmacro{\bottomY}{-#3} % bottom of the box
        
        \draw[very thick] (#3/2,\topY) -- ++(0,#2);
        \draw[very thick, color=purification] (-#3/2,\topY) -- ++(0,#2);
        
        \draw[very thick] (0,\bottomY) -- ++(0,-#2);
    
        \draw[very thick, rounded corners=3pt] (-#3,0) -- (-#3-#2,0) -- (-#3-#2,-#3-#2);
        \draw[very thick, rounded corners=3pt] (#3,0) -- (#3+#2,0) -- (#3+#2,-#3-#2);
    
        \draw[ thick, fill=tensorcolor, rounded corners=2pt] (-#3,-#3) rectangle (#3,#3);
        \draw (0,0) node {\scriptsize #4};

        \draw[-{Stealth[length=1mm, width=2.3mm]}] (0,\bottomY - #2/3) -- (0,\bottomY+0.03 - #2/3 );
        \draw[-{Stealth[length=1mm, width=2.3mm]}] (-#3-#2,\bottomY - #2/3) -- (-#3-#2,\bottomY+0.03 - #2/3 );
        \draw[-{Stealth[length=1mm, width=2.3mm]}] (#3+#2,\bottomY - #2/3) -- (#3+#2,\bottomY+0.03 - #2/3 );
        \draw[-{Stealth[length=1mm, width=2.3mm]}] (#3/2,\topY + #2/2) -- (#3/2,\topY +0.03 + #2/2 );
        \draw[-{Stealth[length=1mm, width=2.3mm]}, color=purification] (-#3/2,\topY + #2/2) -- (-#3/2,\topY +0.03 + #2/2 );
        
    \end{scope}
}

\usepackage{tikz}
\usepackage{comment}
\usetikzlibrary{positioning}
\usetikzlibrary{decorations.pathreplacing,calligraphy,decorations.markings}

\definecolor{tensorcolor}{rgb}{0.65,0.77,0.95}
\definecolor{btensorcolor}{rgb}{0.65,0.50,0.69}
\definecolor{whitetensorcolor}{rgb}{0.93,0.93,0.93}
\definecolor{purification}{RGB}{124, 138, 128}
\newcommand{\XoneDeco}[1]{
	\begin{scope}[shift={(#1)}]
        \draw[dashed, ultra thick] (0,0) -- (0,1);
          \draw[very thick] (0.25,3) -- (0.25,1.5);  % Left curve
          \draw[very thick, color=purification] (0,3) -- (0,1.5);       % Right curve
          \draw[very thick] (0,1.5) -- (2,1.5);  % Horizontal right
          \draw[very thick] (0.5,1.5) -- (0,1);    % Diagonal down-left
          \draw[very thick] (0,1.5) -- (0,1);  % Vertical down
	\end{scope}
}

\newcommand{\YoneDeco}[1]{
	\begin{scope}[shift={(#1)}]
          \draw[dashed, ultra thick] (0,-1.5) -- (0,0);
          \draw[very thick] (-1.5,-1.5) -- (0,-1.5); % Horizontal right
          \draw[very thick] (-0.5,-1.5) -- (0,-2);   % Diagonal down-right
          \draw[very thick] (0,-1.5) -- (0,-3);  
	\end{scope}
}

\newcommand{\LeftDeco}[1]{
	\begin{scope}[shift={(#1)}]
        \XoneDeco{0,0}
        \YoneDeco{0,1.}
	\end{scope}
}

\newcommand{\XtwoDeco}[1]{
	\begin{scope}[shift={(#1)}]
          \draw[line width=2pt] (0,0) -- (0,1);
          \draw[very thick, color=purification] (-0.25,3) -- (-0.25,1.5);
          \draw[very thick] (0,3) -- (0,1.5);
          \draw[very thick] (0,1.5) -- (-2,1.5);
          \draw[very thick] (-0.5,1.5) -- (0,1);
          \draw[very thick] (0,1.5) -- (0,1);
        \filldraw[very thick, fill=gray!90] (-0.5,1.5) -- (0,1) -- (0,1.5) -- cycle;
	\end{scope}
}

\newcommand{\YtwoDeco}[1]{
	\begin{scope}[shift={(#1)}]
          \draw[line width=2pt] (0,-1.5) -- (0,0);
          \draw[very thick] (1.5,-1.5) -- (0,-1.5);
          \draw[very thick] (0.5,-1.5) -- (0,-2);
          \draw[very thick] (0,-1.5) -- (0,-3);
        \filldraw[very thick, fill=gray!90] (0.5,-1.5) -- (0,-2) -- (0,-1.5) -- cycle;
	\end{scope}
}

\newcommand{\RightDeco}[1]{
	\begin{scope}[shift={(#1)}]
        \XtwoDeco{0,0}
        \YtwoDeco{0,1.}
	\end{scope}
}

\newcommand{\IsomU}[1]{
	\begin{scope}[shift={(#1)}]
        \YtwoDeco{0,0}
        \YoneDeco{2.,0}
	\end{scope}
}

\newcommand{\IsomV}[1]{
	\begin{scope}[shift={(#1)}]
        \XtwoDeco{2,0}
        \XoneDeco{0,0}
	\end{scope}
}

\begin{document}
\def\defaultscaling{0.6}

\title{Structure and Classification of Matrix Product Quantum Channels}

\author{Giorgio \surname{Stucchi}}
\email{giorgio.stucchi@mpq.mpg.de}
\affiliation{Max-Planck-Institut f{\"{u}}r Quantenoptik,
Hans-Kopfermann-Str.\ 1, 85748 Garching, Germany}
\affiliation{Munich Center for Quantum Science and Technology, Schellingstra\ss e 4, 80799 M\"unchen, Germany}
\author{J.~Ignacio \surname{Cirac}}
\affiliation{Max-Planck-Institut f{\"{u}}r Quantenoptik,
Hans-Kopfermann-Str.\ 1, 85748 Garching, Germany}
\affiliation{Munich Center for Quantum Science and Technology, Schellingstra\ss e 4, 80799 M\"unchen, Germany}
\author{Rahul \surname{Trivedi}}
\affiliation{Max-Planck-Institut f{\"{u}}r Quantenoptik,
Hans-Kopfermann-Str.\ 1, 85748 Garching, Germany}
\affiliation{Munich Center for Quantum Science and Technology, Schellingstra\ss e 4, 80799 M\"unchen, Germany}	
\author{Georgios \surname{Styliaris}}
\email{georgios.styliaris@mpq.mpg.de}
\affiliation{Max-Planck-Institut f{\"{u}}r Quantenoptik,
Hans-Kopfermann-Str.\ 1, 85748 Garching, Germany}
\affiliation{Munich Center for Quantum Science and Technology, Schellingstra\ss e 4, 80799 M\"unchen, Germany}

\begin{abstract}

We develop a framework for \emph{Matrix Product Quantum Channels} (MPQCs), a one-dimensional tensor-network description of completely positive, trace-preserving maps. We focus on translation-invariant channels, generated by a single repeated tensor, that admit a local purification. We show that their purifying isometry can always be implemented by a constant-depth brickwork quantum circuit, implying that such channels generate only short-range correlations when applied on product states. In contrast to the unitary setting, where one-dimensional quantum cellular automata (in one-to-one correspondence with matrix product unitaries) carry a nontrivial index, we prove that all locally purified channels belong to a single phase, that is, they can be continuously deformed into one another. We then extend the framework to a broader class of translation-invariant channels capable of generating long-range entanglement and show that these remain deterministically implementable in constant depth using two rounds of measurements and feedforward.
\end{abstract}

\maketitle

{\em Introduction.---}
Tensor-network methods, particularly matrix product states (MPS) and matrix product density operators (MPDOs), have become indispensable tools in the study of one-dimensional quantum many-body systems. MPS efficiently approximate ground states of 1D local and gapped Hamiltonians \cite{hastings:arealaw,arad2013area} and provide a framework to classify 1D gapped phases of matter~\cite{chen2011complete,schuch2011classifying}. MPDOs can accurately represent thermal states of local 1D Hamiltonians~\cite{hastings2006solving,molnar2014approximating}, but also appear as boundary states of 2D gapped systems~\cite{cirac2011entanglement,cirac_matrix_2017}.
%and have been used to classify 1D mixed-state phases, an ongoing effort.
Yet, a complete description of physical processes requires going beyond states and operators to encompass \emph{quantum channels}, i.e., completely positive trace-preserving maps governing open-system dynamics. Representing these channels in matrix-product form is important for several reasons. First, they preserve the area-law and thus naturally arise in classifying mixed-state phases~\cite{Hastings_2011,de2022symmetry,ma2023average,ruiz2024matrix,sun2025anomalous,lessa2025mixed,liu2025trading,liu2026parent}. Second, they can encode symmetry actions on matrix product density operators, similar to how MPO symmetries function for MPS~\cite{wang2015bosonic,roose2019anomalous,Garre_Rubio_2023}. Finally, they offer a compact tensor-network description for correlated noise models beyond low-weight errors on quantum devices, which can be used as a variational ansatz for noise tomography~\cite{torlai2020quantumprocesstomographyunsupervised,filippov2022matrixproductchannelvariationally,liu2026efficientlylearningglobalquantum}.

Despite their relevance, a systematic analysis of the correlation structure and a classification of matrix product quantum channels (MPQCs) remain lacking. What has been thoroughly analyzed is the unitary case, namely matrix product unitaries (MPUs)~\cite{cirac_matrix_2017-1,sahinoglu2018matrix,piroli2021fermionic,shukla2025simple,styliaris2025matrix, francorubio2025symmetrydefectsgaugingquantum, Lootens_2025}.
In the homogeneous setting, where the same tensor is used at every site, MPUs are equivalent to 1D quantum cellular automata (QCA)~\cite{cirac_matrix_2017-1} -- unitaries with a strict, finite light cone~\cite{arrighi_overview_2019}. This equivalence has several strong consequences. First, it reveals that the correlations homogeneous MPUs can generate are heavily constrained. Second, it implies a classification of homogeneous MPUs, as a direct consequence of the classification of QCA via index theory~\cite{gross_index_2012}. Third, it solves the problem of the physical implementation of MPUs -- a priori nontrivial since individual tensors do not directly correspond to physical processes -- implying a low-depth quantum circuit for their physical implementation~\cite{farrelly_causal_2014,cirac_matrix_2017-1}, a fact that can be extended for more general MPUs~\cite{styliaris2025quantumcircuitcomplexitymatrixproduct}. These properties, derived from homogeneity, raise the question of whether similar structural constraints and classification principles apply to quantum channels.

\begin{figure}
    \centering
    \includegraphics[width=\linewidth]{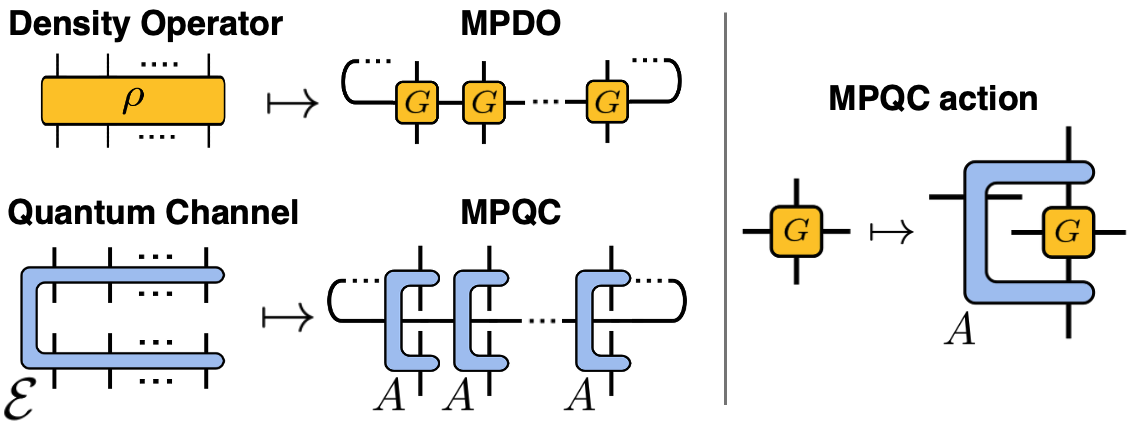}
    \caption{We analyze the structure, classification, and physical implementation of matrix product quantum channels (MPQCs) composed of a repeated tensor $A$. These act on matrix product density operators (MPDOs), preserving their matrix-product form.}
    \label{fig:channels}
\end{figure}

Here, we develop a systematic framework for MPQCs, focusing on those modeling processes that originate from an environment with an underlying matrix-product structure, which we refer to as locally purifiable (LP) channels. We show that imposing homogeneity restricts LP channels to generating only short-range correlations, a consequence of the fact that they all admit a depth-2 circuit representation as isometric gates. Remarkably, we find that all homogeneous LP channels belong to a single equivalence class, connected through continuous deformations. We further examine a class of translation-invariant channels capable of producing long-range correlations and show how these can still be implemented efficiently as quantum circuits. Finally, we present a construction that shows how such channels can be represented in terms of simpler building blocks, namely MPUs and MPS.
\\

{\em Matrix Product Quantum Channels.---}
To define an MPQC, we start from a rank-6 tensor $A$, graphically represented as
\def\leglength{0.5}
\def\barthickness{0.3}
\def\horizontallength{1.8}
\def\verticalheight{2}
\def\squarehalfside{0.3}
\begin{align}
    A^{ij,op}_{mn} =
    \begin{array}{c}
    \begin{tikzpicture}[scale=\defaultscaling,baseline={([yshift=-0.65ex] current bounding box.center)}]
        \HorseshoeTensor{0,0}{\leglength}{\barthickness}{$A$}{0}{\horizontallength}{\verticalheight}{1}
    \end{tikzpicture}
    \end{array}
    \in \mathbb C \,,
\end{align}
    see \cref{fig:channels}. The indices $i,j = 1,\dots,d_{\mathrm{in}}$ label the input physical space, $o,p = 1,\dots,d_{\mathrm{out}}$ the output physical space, and $m,n = 1,\dots,D$ the auxiliary (bond) space. The integer $D$ is the \emph{bond dimension} of the MPQC. Hence, each $A^{ij,op}$ can be represented as a $D \times D$ matrix. The tensor $A$ generates a \emph{homogeneous matrix product quantum channel} (hMPQC) if the superoperator
\begin{align}
    {\rm MPQC}_N^{A}:\;\mathcal{L}(\mathbb{C}^{d_{\mathrm{in}}})^{\otimes N} \to \mathcal{L}(\mathbb{C}^{d_{\mathrm{out}}})^{\otimes N},
\end{align}
where $\mathcal{L}(X)$ is the linear space of operators acting on $X$,
transforming density matrices as
\begin{align} \label{eq:mpqc_def}
\rho'_{\,o_1\dots o_N,\; p_1\dots p_N} =\!\!
\sum_{\substack{i_1,\dots,i_N \\ j_1,\dots,j_N}}
\!\!\Tr\!\!\left[
\prod_{k=1}^{N}
A^{i_k j_k,\, o_k p_k}
\right]
\rho_{\,i_1\dots i_N,\; j_1\dots j_N},
\end{align}
is completely positive (CP) and trace-preserving (TP) for all $N \ge 1$. 
The above construction repeats a fixed tensor $A$ along the chain with periodic boundary conditions --- hence \textit{homogeneous} --- and defines the {\sffamily hMPQC} class\footnote{Channels in {\sffamily hMPQC} are automatically \emph{translation-invariant}, i.e., they commute with the translation operator; note that homogeneity and translation-invariance are two distinct notions.}.
%The above construction repeats a fixed tensor $A$ along the chain with periodic boundary conditions, a setting we call \emph{homogeneous}. We refer to the resulting class of channels as {\sffamily hMPQC} 
We use a stylized font throughout to denote different channel classes.

Homogeneous MPDOs, and thus also the corresponding MPS, can be expressed in the hMPQC form by setting $d_{\rm in} = 1$, i.e., by trivializing the input of the channel. From the matrix-product form of the coefficients in \cref{eq:mpqc_def}, it follows that an hMPQC acting on an MPDO can increase the bond dimension by at most a multiplicative factor $D$ (but decrease it arbitrarily).

As with MPS, one can more generally allow arbitrary site-dependent tensors $A_k$ with varying bond dimensions $D_k$. In fact, any quantum channel can be represented as an MPQC in this general form by successive singular value decompositions after vectorizing, following the standard MPS procedure~\cite{vidal2003efficient,Sch_n_2005}. However, for a generic channel, the bond dimension might grow exponentially with system size $N$, rendering the matrix-product representation useless.
We therefore define the class {\sffamily MPQC}, specified over increasing system sizes $N$, consisting of only channels admitting a representation with a bond dimension uniformly bounded by a \emph{constant} independent of $N$. We henceforth focus on the homogeneous setting, where this property holds by definition, and discuss the general case in the Supplemental Material~\cite{sm}.

CP and TP, needed for the superoperator $\mathrm{MPQC}_N^{A}$ to define a proper quantum channel, can be equivalently formulated in terms of the corresponding Choi-Jamiołkowski state 
${C_N^{A} \equiv (\mathrm{MPQC}_N^{A} \otimes \mathbb{I_{\rm in}}) (\ket{\Phi^+}\bra{\Phi^+}})$:
\begin{align}
    \label{Choi}
    (i) \; C_N^{A} \succeq 0,  \quad
    (ii) \;\Tr_\mathrm{out} (C_N^{A}) = \mathbb{I}_\mathrm{in}  \qquad \forall N,
\end{align}
where $\ket{\Phi^+} \equiv (\sum_{i=1}^{d_{\rm in}} \ket{i} \ket{i})^{\otimes N}$ is an (unnormalized) Bell state.
However, verifying the first condition, i.e.,  whether the matrix product operator $C_N^A$ is positive semidefinite for all system sizes, is undecidable in the general case~\cite{cuevas_fundamental_2016}. Moreover, we want to focus on physically relevant channels, i.e., noise models arising from an environment with a matrix-product structure. Consequently, we focus on a subclass of MPQCs in which CP and TP can be explicitly imposed and verified.

{\em Locally Purified Quantum Channels.---}
We now introduce the \emph{homogeneous locally purified} ({\sffamily{hLP}}) subclass of {\sffamily{hMPQC}}. 
Within this class, CP is automatic
%, circumventing any undecidability issues,
and TP can be characterized on the level of the single tensor.

\begin{defn}[Local purification]
An MPQC tensor $A$ admits a {\em local purification} if there exists another tensor ${\cal V}$ such that\footnote{We are assuming that the virtual legs of the MPQC tensor factorize; this can be imposed by a gauge transformation.
}
\begin{align}
    \begin{array}{c}
        \begin{tikzpicture}[scale=\defaultscaling,baseline={([yshift=-0.65ex] current bounding box.center)}]
            \HorseshoeTensor{0,0}{\leglength}{\barthickness}{$A$}{0}{\horizontallength}{\verticalheight}{0}
        \end{tikzpicture}
    \end{array}
        \;=\;
    \begin{array}{c}
        \begin{tikzpicture}[scale=\defaultscaling,baseline={([yshift=-0.65ex] current bounding box.center)}]
            \pgfmathsetmacro{\calculatedDistance}{\verticalheight/2-\barthickness/2}
            \pgfmathsetmacro{\squarehalf}{0.35}
            \LPTensor{0,0}{\leglength*0.75}{\squarehalf}{${\cal V}$}{0}{\calculatedDistance}
        \end{tikzpicture}
    \end{array},
    \quad
    \begin{array}{c}
        \begin{tikzpicture}[scale=\defaultscaling,baseline={([yshift=-0.65ex] current bounding box.center)}]
            \IsomTensor{0,0}{\leglength*0.75}{0.35}{$\cal{V}$}{0}
        \end{tikzpicture}
    \end{array}
        =
    \begin{array}{c}
    \left(
        \begin{tikzpicture}[scale=\defaultscaling, baseline={([yshift=-0.65ex] current bounding box.center)}]
           \IsomConjTensor{0,0}{\leglength*0.75}{0.35}{$\cal{V}$}{0}
        \end{tikzpicture}
    \right)^*
    \end{array}.
\end{align}
An hMPQC belongs to {\sffamily{hLP}} if its tensor admits a local purification.
\end{defn}

Informally, this condition guarantees the existence of a \emph{local} Kraus decomposition at each physical site. More precisely, admitting a local purification generates, by construction, a CP map; this can be directly verified via the Choi-Jamiołkowski state, which takes the form $C_A^N = W^\dagger W$, and is thus positive-semidefinite~\cite{sm}. The Stinespring dilation of the resulting CP map is the matrix product operator (MPO)
$V_N: (\mathbb{C}^{d_{\mathrm{in}}})^{\otimes N} \to  (\mathbb{C}^{d_{\mathrm{out}} \chi})^{\otimes N}$, defined by
\begin{align}
V_N \;\equiv\;
    \begin{tikzpicture}[scale=\defaultscaling,baseline={([yshift=-0.65ex] current bounding box.center)}]
        \pgfmathsetmacro{\squarehalf}{0.4}
            \SingleTrLeft{(-4*\leglength,0)}
            \draw [very thick] (11*\leglength,0) to  (15*\leglength,0);
            \SingleTrRight{(15*\leglength,0)}
            \draw [very thick] (-4*\leglength,0) to  (0,0);
            \IsomTensor{0,0}{\leglength}{\squarehalf}{${\cal V}$}{0}
            \IsomTensor{2*\leglength + 2*\squarehalf,0}{\leglength}{\squarehalf}{${\cal V}$}{0}
            \SingleDots{4*\leglength + 4*\squarehalf, 0}{2*\leglength}
            \IsomTensor{6*\leglength + 6*\squarehalf,0}{\leglength}{\squarehalf}{${\cal V}$}{0}
    \end{tikzpicture}
    \;.
\end{align}
Here $\chi$ is the dimension of the local purification space, denoted in gray. TP, which needs to be imposed separately, amounts to the condition $V_N^\dagger V^{}_N = \mathbb{I}_{d_{\rm in}}^{\otimes N}$, i.e., that $V_N$ is a \emph{homogeneous matrix product isometry} (hMPI).  The resulting hMPQC is then
\begin{align}
    {\rm MPQC}_{N}^{\mathcal V} = \Tr_{\rm pur}[V^{}_N (\cdot) V^\dagger_N].
\end{align}

Physically, an hMPI acting on an MPS can increase the entanglement entropy across any bipartition by at most a constant amount. Consequently, an hLP channel maps any MPDO whose purifying MPS obeys an area law to one that also obeys it. In fact, locally purified channels are the most general quantum channels with this property, once area-law preservation is required to be stable under ancillary extensions~\cite{sm}.

%We note that, although every channel admits a Stinespring dilation, in the general (inhomogeneous) setting, ${\text{\sffamily LP}  \subset \text{\sffamily MPQC}}$. This is a direct consequence of the analogous separation for mixed states, shown in Ref.~\cite{cuevas_purifications_2013} (see~\cite{sm}).

{\em Structure of hMPIs.---} We will now show that imposing homogeneity and the existence of a local purification has very strong implications for the correlations that the resulting channel can generate. To this end, we will fully characterize their purifying hMPI. This question parallels the characterization of homogeneous matrix product unitaries (MPUs) in Refs.~\cite{cirac_matrix_2017-1,sahinoglu2018matrix}, 
where they were shown to coincide with quantum cellular automata (QCA) in 1D, that is, unitaries possessing a strict light cone~\cite{arrighi_overview_2019, vanrietvelde2025causal}. 
Here, we extend this framework to hMPIs.

For the correlation structure to emerge, short-range entanglement must be coarse-grained away. On the tensor level, this is done by \emph{blocking}, i.e., grouping $q$ consecutive tensors ${\cal V}$ into a single effective tensor ${\cal V}_q$ (see, e.g.,~\cite{cirac_matrix_2021}). Throughout, when we block, $q$ will be independent of the system size $N$ so that blocking eliminates only short-range, but not long-range, correlations.

\begin{thm}
\label{depth2}
        Any hMPI can be written as a depth-two brick wall quantum circuit of isometric gates $u,v$, satisfying $u^\dagger u =\Id_{d_{\mathrm{in}}^2}$ and $v^\dagger v =\Id_{\ell r}$, after blocking at most $D^4$ times:
\pgfmathsetmacro{\spacing}{1.5cm} 
\begin{align*}
\begin{tikzpicture}[scale=\defaultscaling,x=1.3*\spacing]
    \definecolor{cream}{RGB}{255,255,221}
    \pgfmathsetmacro{\height}{1.5}
    \pgfmathsetmacro{\shift}{0.7}
    \pgfmathsetmacro{\labelpositions}{\height / 2 + \shift / 2 + 0.2}
    \pgfmathsetmacro{\labelcloseness}{0.15}
    \pgfmathsetmacro{\extra}{.1}
    \foreach \i in {0,...,6} {
        \draw[color=black,very thick] (\i,\shift) -- (\i,\height);
        \draw[color=black,very thick] (\i,-\shift) -- (\i,-\height);
    }
    \foreach \i in {0,2,4,6} {
        \draw[color=black,dashed, very thick] (\i,-\shift) -- (\i,\shift);
        \draw[color=purification,very thick] (\i,\shift) -- (\i,\height);
        \draw[color=black,very thick] (\i + \extra,\shift) -- (\i + \extra,\height);
        \draw (\i - \labelcloseness, 0) node {\scriptsize $r$};
        \draw (\i + \labelcloseness + 0.05, -\labelpositions) node {\scriptsize $d_{\mathrm{in}}$};
        \draw (\i - \labelcloseness, \labelpositions) node {\scriptsize $\chi$};
        \draw (\i + \labelcloseness + \extra + 0.1, +\labelpositions) node {\scriptsize $d_{\mathrm{out}}$};
    }
    \foreach \i in {1,3,5} {
        \draw[color=black,ultra thick] (\i,-\shift) -- (\i,\shift);
        \draw[color=purification,very thick] (\i - \extra,\shift) -- (\i - \extra,\height);
        \draw (\i + \labelcloseness, 0) node {\scriptsize $\ell$};
        \draw (\i + \labelcloseness + 0.05, -\labelpositions) node {\scriptsize $d_{\mathrm{in}}$};
        \draw (\i + \labelcloseness + 0.1, +\labelpositions) node {\scriptsize $d_{\mathrm{out}}$};
        \draw (\i - \labelcloseness - \extra, +\labelpositions) node {\scriptsize $\chi$};
    }
    \foreach \i in {0,2,4} {
        \draw[thick, fill=cream, rounded corners=2pt]
            (\i - \extra, -0.3 + \shift) rectangle (\i+1+\extra, 0.3 + \shift);
        \draw (\i+0.5,\shift) node {\scriptsize $v$};
    }
    \foreach \i in {1,3,5} {
        \draw[thick, fill=cream, rounded corners=2pt]
            (\i - \extra, -0.3 - \shift) rectangle (\i+1 + \extra, 0.3 - \shift);
            \draw (\i+0.5,-\shift) node {\scriptsize $u$};
    }
    \begin{scope}
        \clip (-0.5, -\height) rectangle (0.1 + \extra, \height); % Clip to right half
        \draw[thick, fill=cream, rounded corners=2pt]
            (-1 - \extra, -0.3 - \shift) rectangle (0 + \extra, 0.3 - \shift);
    \end{scope}
    \begin{scope}
        \clip (6 - \extra - 0.1, -\height) rectangle (6 + 0.5, \height); % Clip to left half
        \draw[thick, fill=cream, rounded corners=2pt]
            (6 - \extra, -0.3 + \shift) rectangle (6 + 1 + \extra, 0.3 + \shift);
    \end{scope}
\end{tikzpicture}
\end{align*}
The isometries $u:\mathbb{C}^{d_{\mathrm{in}}^2}\mapsto \mathbb{C}^{\ell r}$ and $v:\mathbb{C}^{\ell r}\mapsto \mathbb{C}^{\chi^2d_{\mathrm{out}}^2}$ satisfy $d_{\mathrm{in}}^2 \leq r\ell \leq d_{\mathrm{out}}^2 \chi^2$.
\end{thm}

%After blocking a finite number of sites, the tensor $\mathcal{V}_q$ can be decomposed into two tensors corresponding to different bipartitions of its indices. These decompositions can be rearranged into a two-layer brick-wall circuit, with the global isometry condition ensuring that the resulting gates are themselves isometric~\cite{sm}, analogously to the hMPU case~\cite{cirac_matrix_2017-1}.
As a result, the correlations that a hMPI -- and thus also a hLP channel -- can generate are bounded by a strict light cone because of the depth-two brick-wall circuit structure\footnote{This property was assumed, rather than shown to emerge, in some prior work on channels with tensor-network structure (Ref.~\cite{de_Groot_2022}).}. Concretely, it means that, after blocking, the reduced state $\rho_{AB}$ of the global pure state $V_N \ket{0}^{\otimes N}$ over any regions $A$, $B$, separated by more than 4 sites, factorizes: $\rho_{AB} = \rho_A \otimes \rho_B$ if ${\rm dist}(A,B) > 4$. This clearly generalizes to other input states with a finite correlation range\footnote{If long-range correlations are already present in the input, they are not necessarily erased.}.

Physically, \cref{depth2} also directly implies that the channels in {\sffamily{hLP}} have low circuit complexity. More precisely, the isometries in \cref{depth2} can be physically implemented via local unitary circuits and ancillas with only a constant depth, depending only on the bond dimension but not the system size. This also includes, e.g., left and right lattice translations by one site, known as \emph{shift} unitaries~\cite{schumacher_reversible_2004,gross_index_2012}.
%In particular, the shift is included in \cref{depth2} by choosing $d \coloneqq d_\mathrm{in} = d_\mathrm{out}$ and $r$ or $l$ equal to $d^2$, respectively, as in the MPU case~\cite{cirac_matrix_2017}.
Another direct consequence of \cref{depth2} is that, if there exists two hLP channels mapping the MPDOs $\rho_N \mathrel{\raisebox{-0.2ex}{$\rightleftarrows$}} \sigma_N$, then the two states are necessarily in the same mixed-state phase~\cite{Hastings_2011,de2022symmetry,ma2023average,ruiz2024matrix,sun2025anomalous,lessa2025mixed,liu2025trading,liu2026parent}. In the homogeneous setting, this amounts to compact transformations over a single tensor, despite the channel having a brickwork structure.

\textit{Classification of hLP Channels.---}
A landmark result for one-dimensional QCAs is their classification by a positive-rational-valued \emph{index}~\cite{gross_index_2012}, which quantifies the net transport of quantum information along the chain. The index labels disconnected sectors of the QCA space: two QCAs share the same index iff they can be continuously deformed into one another while preserving the QCA property. Since homogeneous MPUs coincide with translation-invariant 1D QCAs~\cite{cirac_matrix_2017-1}, this classification can be formulated directly at the level of a single MPU tensor: two QCAs have the same index iff their MPU tensors are continuously deformable within the space of MPU-generating tensors~\cite{cirac_matrix_2017-1}.

In the same spirit, we now define an equivalence relation for hLP channels. As in the case of MPUs, this is defined by continuous deformations at the level of the single tensor generating the MPI, i.e., the channel dilation.

\begin{defn}[Equivalence]
Let $\mathcal E_0$ and $\mathcal E_1$ be channels in {\sffamily{hLP}} 
with identical input and output dimensions $d_{\mathrm{in}}$ and $d_{\mathrm{out}}$. 
They are \emph{equivalent} if there exists a continuous map of hMPI-generating tensors
\begin{align}
    [0,1]\ni s \mapsto \mathcal{V}(s) =
    \begin{array}{c}
        \begin{tikzpicture}[scale=0.9,baseline={([yshift=-0.65ex] current bounding box.center)}]
            \IsomTensor{0,0}{\leglength}{0.4}{$\mathcal{V}(s)$}{0}
            \draw (-0.55*\leglength,0.65) node {\scriptsize $\chi$};
            \draw (0.9*\leglength,0.65) node {\scriptsize $d_{\mathrm{out}}$};
            \draw (1.3*\leglength,0.15) node {\scriptsize $D$};
            \draw (-1.3*\leglength,0.15) node {\scriptsize $D$};
            \draw (0.25,-0.65) node {\scriptsize $d_{\mathrm{in}}$};
        \end{tikzpicture}
    \end{array}
\quad D,\chi \text{ $\rm const.$} \;,
\end{align}
with associated MPI $V_N(s)$, such that
\begin{align}
    \mathcal{E}_i = \Tr_{\mathrm{pur}}[V_N(i)^{}(\cdot)V_N(i)^\dagger], \quad i \in \{ 0,1 \} \quad \forall N\;.
\end{align}
\end{defn}
\noindent In essence, two channels are \emph{equivalent} if their purifying isometries can be continuously deformed into one another, allowing enlargement of the (possibly redundant) purification space, such that the endpoints reproduce the respective channels after tracing out the environment. Our result is:

\begin{thm}
    Any two channels in {\sffamily{hLP}} with matching input and output dimensions are equivalent.
\end{thm}

For the proof, we embed the two MPI tensors into a common enlarged purification space, arranged so that at the initial point of the path, one tensor acts over the input while the other acts on a fixed ancilla qubit, and is effectively discarded. The interpolation is then implemented by a continuous local unitary rotation that exchanges their roles, smoothly transferring the action from one tensor to the other. The result of the theorem further extends to arbitrary input and output dimensions once equivalence is relaxed to allow for ancillary inputs that are traced out and ancillary outputs prepared in fixed product states~\cite{sm}. 

At first sight, this seems in tension with the QCA classification, since the shift and the identity, viewed as unitaries, have distinct indices and are therefore inequivalent. The resolution is that we now regard them as channels: only their action after tracing out a purification space is fixed, not their specific unitary realization. This additional freedom permits continuous deformations that are forbidden in the strictly unitary setting. In this sense, the purification space plays the role of ancillas in QCA theory, whose inclusion is known to trivialize any QCA~\cite{arrighi_unitarity_2009,farrelly_causal_2014}.

%\gs{Finally, {\sffamily{hLP}} channels preserve the distinction between area-law and volume-law entanglement. 
%If one adopts as definition for locally purified MPDO phases that two states are considered equivalent if they can be connected by a continuous path consisting entirely of area-law states, then, by \cref{depth2}, the resulting notion of phase does not differ from the conventional one based on local finite-depth  (see \cite{sm} for details).}

{\em Appearance of Long-range Entanglement.---}
As shown before, hMPIs do not produce long-range entanglement.
\cref{depth2} raises the question of what physically forbids long-range correlations. We will show that the key ingredient is the exact isometry condition imposed for all $N$. The situation parallels the MPS setting, where normalized hMPS are necessarily short-range entangled, while relaxing the normalization condition by a constant restores long-range correlated states such as the GHZ state~\cite{schuch2011classifying}. We now show that the analogous relaxation applies to MPIs.
\begin{ex}
\label{longrangeisometry}
    Consider the quantum channel $\mathcal{E}(\cdot) = \Tr(\cdot) \, \ket{\mathrm{GHZ}}\bra{\mathrm{GHZ}}$ acting on $N$ qubits. For all $N$, it admits a purifying MPI with $D_{\rm MPI} = \chi = 2$
\begin{align*} 
V_N = \frac{1}{\sqrt{2}} \;
\begin{tikzpicture}[scale=0.45,baseline={([yshift=-0.65ex] current bounding box.center)}] 
        \pgfmathsetmacro{\squarehalf}{0.4}
        \pgfmathsetmacro{\halfunit}{\squarehalf+\leglength}
        \SingleTrLeft{(-4*\leglength,0)} 
        \draw [very thick] (11*\leglength,0) to (16*\leglength,0); 
        \SingleTrRight{(16*\leglength,0)} 
        \draw [very thick] (-4*\leglength,0) to (0,0); 
        \GHZTensor{0,0}{\leglength}{\squarehalf}{0} 
        \GHZTensor{2*\halfunit,0}{\leglength}{\squarehalf}{0} 
        \SingleDots{4*\halfunit, 0}{2*\leglength} 
        \GHZTensor{6*\halfunit,0}{\leglength}{\squarehalf}{0}
        % y > 0 part
        \begin{scope}
          \clip (0,0) rectangle (7,1.5);
          \draw[very thick,color=purification]
            (\halfunit,\halfunit) .. controls (\halfunit,0) and (0,-0.5*\halfunit) .. (0,-1.5*\halfunit);
          \draw[very thick,color=purification]
            (3*\halfunit,\halfunit) .. controls (3*\halfunit,0) and (2*\halfunit,-0.5*\halfunit) .. (2*\halfunit,-1.5*\halfunit);
          \draw[very thick,color=purification]
            (7*\halfunit,\halfunit) .. controls (7*\halfunit,0) and (6*\halfunit,-0.5*\halfunit) .. (6*\halfunit,-1.5*\halfunit);
        \end{scope}
        % y < 0 part
        \begin{scope}
          \clip (-0.2,-1.5) rectangle (7,0);
          \draw[very thick]
            (\halfunit,\halfunit) .. controls (\halfunit,0) and (0,-0.5*\halfunit) .. (0,-1.5*\halfunit);
          \draw[very thick]
            (3*\halfunit,\halfunit) .. controls (3*\halfunit,0) and (2*\halfunit,-0.5*\halfunit) .. (2*\halfunit,-1.5*\halfunit);
          \draw[very thick]
            (7*\halfunit,\halfunit) .. controls (7*\halfunit,0) and (6*\halfunit,-0.5*\halfunit) .. (6*\halfunit,-1.5*\halfunit);
        \end{scope}
    \end{tikzpicture}
    \;,
\end{align*}
where
$
    \begin{tikzpicture}[scale=0.6,baseline={([yshift=-0.65ex] current bounding box.center)}] 
        \pgfmathsetmacro{\squarehalf}{0.4}
        \pgfmathsetmacro{\labels}{\squarehalf+\leglength/2}
        \GHZTensor{0,0}{\leglength}{\squarehalf}{0} 
        \draw (\labels,0.2) node {\scriptsize 0};
        \draw (-\labels,0.2) node {\scriptsize 0};
        \draw (0.2,\labels) node {\scriptsize 0};
    \end{tikzpicture}
    =
    \begin{tikzpicture}[scale=0.6,baseline={([yshift=-0.65ex] current bounding box.center)}] 
        \pgfmathsetmacro{\squarehalf}{0.4}
        \pgfmathsetmacro{\labels}{\squarehalf+\leglength/2}
        \GHZTensor{0,0}{\leglength}{\squarehalf}{0}
        \draw (\labels,0.2) node {\scriptsize 1};
        \draw (-\labels,0.2) node {\scriptsize 1};
        \draw (0.2,\labels) node {\scriptsize 1};
    \end{tikzpicture}
    =1
$
    or vanishing otherwise is the local tensor of an unnormalized $\mathrm{GHZ}$ state. The resulting channel generates long-range entanglement due to the correlations present in the $\mathrm{GHZ}$ state.
\end{ex}

This appears to be in tension with \cref{depth2}, which states that hMPIs cannot generate long-range entanglement. The discrepancy arises because this example does not fit within the scope of the theorem, but merely because of the $1/\sqrt{2}$ constant: the homogeneous MPI constructed from the tensor
$\mathcal{V} = \begin{tikzpicture}[scale=0.4,baseline={([yshift=-0.65ex] current bounding box.center)}] 
        \pgfmathsetmacro{\squarehalf}{0.4}
        \pgfmathsetmacro{\halfunit}{\squarehalf+\leglength} 
        \GHZTensor{2*\halfunit,0}{\leglength}{\squarehalf}{0} 
        % y > 0 part
        \begin{scope}
          \clip (1,0) rectangle (3,1.5);
          \draw[very thick,color=purification]
            (\halfunit,\halfunit) .. controls (\halfunit,0) and (0,-0.5*\halfunit) .. (0,-1.5*\halfunit);
          \draw[very thick,color=purification]
            (3*\halfunit,\halfunit) .. controls (3*\halfunit,0) and (2*\halfunit,-0.5*\halfunit) .. (2*\halfunit,-1.5*\halfunit);
          \draw[very thick,color=purification]
            (7*\halfunit,\halfunit) .. controls (7*\halfunit,0) and (6*\halfunit,-0.5*\halfunit) .. (6*\halfunit,-1.5*\halfunit);
        \end{scope}
        % y < 0 part
        \begin{scope}
          \clip (1,-1.5) rectangle (3,0);
          \draw[very thick]
            (\halfunit,\halfunit) .. controls (\halfunit,0) and (0,-0.5*\halfunit) .. (0,-1.5*\halfunit);
          \draw[very thick]
            (3*\halfunit,\halfunit) .. controls (3*\halfunit,0) and (2*\halfunit,-0.5*\halfunit) .. (2*\halfunit,-1.5*\halfunit);
          \draw[very thick]
            (7*\halfunit,\halfunit) .. controls (7*\halfunit,0) and (6*\halfunit,-0.5*\halfunit) .. (6*\halfunit,-1.5*\halfunit);
        \end{scope}
    \end{tikzpicture}$
is an isometry only after normalizing with an $N$-independent constant, given that $V_{N}^\dagger V_{N}^{} = 2\,\Id_{d_{\rm in}}^{\otimes N}$. Since this constant can have important physical implications (short-range versus long-range entanglement, as the example shows), this motivates relaxing the strict isometry requirement $V_{N}^\dagger V_{N}^{} =\,\Id_{d_{\rm in}}^{\otimes N}$ of {\sffamily{hMPI}} to
\begin{align}
V_{N}^\dagger V_{N}^{} = c\,\Id_{d_{\rm in}}^{\otimes N} \quad \forall N\ge 1,
\end{align}
for some $N$-independent constant \(c>0\). Notice that the constant cannot be incorporated into the generating tensor $\mathcal V$ of $V_N$ without breaking homogeneity. We call the resulting class satisfying the relaxed condition \emph{scaled homogeneous MPI} or {\sffamily sMPI}, which we study as the minimal possible extension exhibiting features absent in the homogeneous framework.

We emphasize that the resulting isometries are valid physical objects; the only nuance is that they admit a homogeneous matrix-product representation only after an appropriate normalization. However, since \cref{depth2} no longer applies, their quantum circuit complexity and the structure of the entanglement they generate are not immediately clear. We address both questions by explicitly characterizing their structure.
%, and the corresponding channels gLPs.

{\em Physical Implementation of sMPIs.---}
We now address how to efficiently implement isometries in {\sffamily {sMPI}}. Recall that, although the tensor network representation of sMPIs gives a compact mathematical description, the individual tensors do not directly relate to physical operations.
We thus now consider the problem of how to turn these tensor-network isometries into quantum circuits composed of \emph{geometrically local} unitary gates. In this setting, due to the existence of long-range entanglement as in \cref{longrangeisometry}, it follows that necessarily at least an $\Omega(N)$-depth circuit is necessary to prepare these isometries. The same holds when allowing local ancillas initialized in a product state, which are eventually discarded.

Despite these fundamental limitations, we now show how to deterministically implement all isometries in {\sffamily sMPI}, and thus also the resulting MPQCs, in constant (i.e., $N$-independent) depth, by allowing, along with the unitary gates, local product-basis measurements and nonlocal classical communication (feedforward), in the spirit of Refs.~\cite{Raussendorf_2005, Broadbent_2009, Watts_2019,Piroli_2021, Lu_2022, smith2023deterministic,baumer2024efficient,buhrman2024state,stephen2024preparing,zhang2024characterizing,Tantivasadakarn_2024, sahay2025classifying, Lootens_2025}. That is, we allow conditioning subsequent unitary gates on classical measurement outcomes, but forbid any post-selection. Our result is:

\begin{thm}
\label{approach1}
    Any sMPI can be deterministically implemented by a constant-depth circuit of local unitary gates, assisted by two rounds of local measurements, nonlocal classical communication and local unitary corrections, as well as $O(N)$ local ancilla qubits.
\end{thm}

The result provides an explicit protocol for implementation, given the single-site tensor that generates the sMPI. The steps, explained in detail in~\cite{sm}, are:
\begin{enumerate}[(i)]
    \item Prepare a state of the form $\ket{\mathrm{GHZ}} = \frac{1}{\sqrt{c}} \sum_{j=0}^{g-1} m_j \ket{j}_A^{\otimes N}$ on an ancillary system consisting of $N$ $g$-dimensional qudits distributed across the system ($g \le D$). This can be achieved deterministically with local operations in constant depth with a single round of measurements and feedforward~\cite{Piroli_2024}, 
    \item Feed this to the control space of a depth-2 control-brickwork circuit of the form
    \begin{equation}
\label{controlunitary_main}
    CU = 
\pgfmathsetmacro{\spacing}{1.5cm}
\begin{tikzpicture}[scale=0.5,x=\spacing, baseline={([yshift=-0.65ex] current bounding box.center)}]
    \definecolor{cream}{RGB}{255,255,221}
    \pgfmathsetmacro{\height}{1.5}
    \pgfmathsetmacro{\shift}{0.7}
    \pgfmathsetmacro{\labelpositions}{\height / 2 + \shift / 2 + 0.2}
    \pgfmathsetmacro{\labelcloseness}{0.15}
    \pgfmathsetmacro{\extra}{.15}

    \clip (0.9,-1.5) rectangle (8.1,1.5);
    
    \foreach \i in {0,1,3,4,6,7,9} {
        \draw[color=black,very thick] (\i,\shift) -- (\i,\height);
    }
    \foreach \i in {0,3,6,9} {
        \draw[color=purification,very thick] (\i,\shift) -- (\i,\height);
    }
    \foreach \i in {1.5, 4.5, 7.5} {
        \draw[color=black,very thick] (\i,\height) -- (\i,-\height);
        \fill (\i,\shift) circle (3pt);
        \fill (\i,-\shift) circle (3pt);
        
    }
    \foreach \i in {0,2,3,5,6,8,9} {
        \draw[color=black,very thick] (\i,-\shift) -- (\i,-\height);
        %\draw (\i - \labelcloseness, -\labelpositions) node {\scriptsize $d$};
    }
    \foreach \i in {0,3,6,9} {
        \draw[color=black,very thick] (\i,-\shift) -- (\i,\shift);
        \draw[color=black,very thick] (\i + \extra,\shift) -- (\i + \extra,\height);
        %\draw (\i - \labelcloseness, 0) node {\scriptsize $(d\chi)^2$};
        %\draw (\i - \labelcloseness, -\labelpositions) node {\scriptsize $d$};
        %\draw (\i - \labelcloseness, \labelpositions) node {\scriptsize $d$};
        %\draw (\i + \labelcloseness + \extra, +\labelpositions) node {\scriptsize $\chi$};
    }
    \foreach \i in {1,4,7} {
        %\draw[color=black,very thick] (\i,-\shift) -- (\i,\shift);
        \draw[color=purification,very thick] (\i - \extra,\shift) -- (\i - \extra,\height);
        %\draw (\i + \labelcloseness, -\labelpositions) node {\scriptsize $d$};
        %\draw (\i + \labelcloseness, +\labelpositions) node {\scriptsize $d$};
        %\draw (\i - \labelcloseness - \extra, +\labelpositions) node {\scriptsize $\chi$};
    }
    \foreach \i in {1,4,7} {
            \draw[color=black,very thick, rounded corners=8pt]
    (\i, \shift) -- (\i, 0) -- (\i+1, 0) -- (\i+1, -\shift);
        }
    \foreach \i in {0,3,6} {
        \draw[color=black,very thick] (\i,\shift) -- (\i+1.5,\shift);
        \draw[thick, fill=cream, rounded corners=2pt]
            (\i - \extra, -0.3 + \shift) rectangle (\i+1+\extra, 0.3 + \shift);
        \draw (\i+0.5,\shift) node {\scriptsize $T$};
    }
    \foreach \i in {-1,2,5,8} {
        \draw[color=black,very thick] (\i,-\shift) -- (\i+2.5,-\shift);
        \draw[thick, fill=cream, rounded corners=2pt]
            (\i - \extra, -0.3 - \shift) rectangle (\i+1 + \extra, 0.3 - \shift);
            \draw (\i+0.5,-\shift) node {\scriptsize $B$};     
    }
\end{tikzpicture},
\end{equation}
where $B$ and $T$ are unitary gates for all values of the control.
    \item Measure all ancilla qudits in parallel and record the outcomes. Discard the ancillas, which factorize after the measurements.
    \item Apply a phase unitary at a single site, which depends on all classical outcomes of the previous measurement step.
\end{enumerate}
%
%For the special case of a hMPI, the circuit of $(ii)$ can be simplified to that of \cref{depth2}, with a trivial control space, thus without the need for measurements and feedforward. 
The key technical result we obtain, which allows us to design the protocol --- specifically the circuit of $(ii)$ --- and guarantee its validity and overcome the need for post-selection, is a complete characterization of the structure of all sMPIs. More precisely, we show that they all decompose as
\begin{align}
    V_N = \frac{1}{\sqrt{c}}\sum_{j=0}^{g-1} m_j V_{j,N}
\end{align}
where all $V_j$ are hMPIs which are \emph{orthogonal}, i.e., $V^\dagger_{i,N} V_{j,N} = \delta_{ij} \Id_{d_{\rm in}}^{\otimes N}$, $g \le D$ and the integer constants $m_j$ satisfy $\sum_j m_j^2 = c$ (see \cite{sm} for details). 

{\em sMPIs from MPUs.---}
Any hMPI can be viewed as a hMPU acting on a fixed input MPS (a product state): the isometric gates in \cref{depth2} can be extended to unitaries, whose action on a fixed product-state input reproduces the original isometries. Likewise, \cref{longrangeisometry} can be viewed as a product MPU --- namely, the product of swaps --- of which half of the input legs are contracted with a $\mathrm{GHZ}$ state. 

This naturally raises the question of whether every sMPI admits a representation in terms of simpler components, that is, as a hMPU acting on an input MPS. We show that the answer is indeed positive, with the corresponding MPU having homogeneous bulk but also a boundary, graphically:
\begin{equation}
    U = \begin{tikzpicture}[scale=\defaultscaling,baseline={([yshift=-0.65ex] current bounding box.center)}]
            \pgfmathsetmacro{\leglength}{0.5}
            \pgfmathsetmacro{\squarehalfside}{0.4}
            \SingleTrLeft{(-4*\leglength,0)}
            \draw [very thick] (-4*\leglength,0) to  (0,0);
            \MPOTensor{0,0}{\leglength}{\squarehalfside}{${\cal U}$}{0}
            \MPOTensor{2*\leglength + 2*\squarehalfside,0}{\leglength}{\squarehalfside}{${\cal U}$}{0}
            \SingleDots{4*\leglength + 4*\squarehalfside, 0}{2*\leglength}
            \MPOTensor{6*\leglength + 6*\squarehalfside,0}{\leglength}{\squarehalfside}{${\cal U}$}{0}
            %\draw [very thick] (11*\leglength,0) to  (14*\leglength,0);
            \SingleTrRight{(8*\leglength + 7.5*\squarehalfside,0)}
            \filldraw[color=black, fill=whitetensorcolor, thick] (-2.5*\leglength,0) circle (0.5);
            \node at (-2.5*\leglength,0) {\scriptsize $\lambda$};
    \end{tikzpicture}.
\end{equation}
The need for a boundary can be expected from the case of sMPUs that create long-range correlations, e.g., $U = \frac{1}{\sqrt{2}} (\Id^{\otimes N} + i X^{\otimes N})$~\cite{styliaris2025matrix}.

\begin{thm}

Any normalized sMPI $V$ on an {$N$\nobreakdash-qudit} system can be represented as an MPU with homogeneous bulk and a boundary, acting over a $\ket{\mathrm{GHZ}}=\frac{1}{\sqrt{c}}\sum_{j=0}^{g-1} m_j \ket{j}^{\otimes N}$ input state.
\end{thm}

The construction follows \cref{approach1} up to step (iii), where it diverges due to the requirement that the overall operation be unitary, as we seek an MPU and thus cannot directly employ measurements. The idea is to instead use an \textit{amplitude amplification} method, which deterministically yields the target isometry by iteratively suppressing undesired components using unitary operations alone~\cite{styliaris2025quantumcircuitcomplexitymatrixproduct,gilyen2019quantum} (see \cite{sm} for details).

{\em Outlook.---}
We have developed the theory of MPQCs as a tensor-network representation of one-dimensional quantum channels and the physically motivated LP subclass, where the environment admits a locality structure. For homogeneous channels, we proved a circuit decomposition into depth-two isometries after blocking, introduced a classification, and established the existence of a unique equivalence class. Finally, we identified a translation-invariant class with long-range entanglement and provided an efficient constant-depth implementation protocol using measurements and feedforward. We have also established how this class connects to elementary components, namely, MPUs and MPS.

The present work suggests several future directions. First, the focus was on translationally invariant channels, but non-translationally invariant MPQCs could be relevant as an ansatz for describing correlated noise beyond the low-weight regime \cite{torlai2020quantumprocesstomographyunsupervised,filippov2022matrixproductchannelvariationally,liu2026efficientlylearningglobalquantum}. \cref{depth2} shows that tomography schemes based on locally purified matrix-product ansätze and translational invariance will fail to reconstruct long-range correlated noise, motivating the development of more expressive channel classes such as the {\sffamily{sLP}} introduced here. From the many-body perspective, one could study MPDO symmetries arising from homogeneous MPQC, systematically extending the results of Ref.~\cite{de_Groot_2022}, similar to how MPO symmetries function for MPS~\cite{wang2015bosonic,roose2019anomalous,Garre_Rubio_2023}. \gs{hLP} may also be relevant for describing edge dynamics in dissipative 2D Floquet systems. This parallels how causality-preserving unitaries, which are equivalent to 1D MPUs, determine edge behavior in the closed setting~\cite{Potter2016,Po_2016}. Moreover, local quantum channels naturally arise in methods for classifying mixed-state phases~\cite{Hastings_2011,de2022symmetry,ma2023average,ruiz2024matrix,sun2025anomalous,lessa2025mixed,liu2025trading,liu2026parent}. In that context, it would be relevant to extend the structure result to MPIs allowing for normalization constants dependent on the system size. Finally, it would be interesting to investigate if the strict light cone of hMPIs, and thus also the classification of the resulting channels, extends to higher spatial dimensions.

{\em Acknowledgments.---} We thank Yuhan Liu for helpful discussions. I.C. acknowledges the project THEQUCO within the Munich Quantum Valley (MQV), which is supported by the Bavarian state government with funds from the Hightech Agenda Bayern Plus. R.T. acknowledges funding from the European Union’s Horizon Europe research and innovation program under grant agreement number 101221560 (ToNQS). I.C., R.T., and G.S. acknowledge funding from the Deutsche Forschungsgemeinschaft (DFG, German Research Foundation) under Germany’s Excellence Strategy EXC2111-390814868.

\bibliography{references}

%apsrev4-2.bst 2019-01-14 (MD) hand-edited version of apsrev4-1.bst
%Control: key (0)
%Control: author (8) initials jnrlst
%Control: editor formatted (1) identically to author
%Control: production of article title (0) allowed
%Control: page (0) single
%Control: year (1) truncated
%Control: production of eprint (0) enabled
\begin{thebibliography}{76}%
\makeatletter
\providecommand \@ifxundefined [1]{%
 \@ifx{#1\undefined}
}%
\providecommand \@ifnum [1]{%
 \ifnum #1\expandafter \@firstoftwo
 \else \expandafter \@secondoftwo
 \fi
}%
\providecommand \@ifx [1]{%
 \ifx #1\expandafter \@firstoftwo
 \else \expandafter \@secondoftwo
 \fi
}%
\providecommand \natexlab [1]{#1}%
\providecommand \enquote  [1]{``#1''}%
\providecommand \bibnamefont  [1]{#1}%
\providecommand \bibfnamefont [1]{#1}%
\providecommand \citenamefont [1]{#1}%
\providecommand \href@noop [0]{\@secondoftwo}%
\providecommand \href [0]{\begingroup \@sanitize@url \@href}%
\providecommand \@href[1]{\@@startlink{#1}\@@href}%
\providecommand \@@href[1]{\endgroup#1\@@endlink}%
\providecommand \@sanitize@url [0]{\catcode `\\12\catcode `\$12\catcode
  `\&12\catcode `\#12\catcode `\^12\catcode `\_12\catcode `\%12\relax}%
\providecommand \@@startlink[1]{}%
\providecommand \@@endlink[0]{}%
\providecommand \url  [0]{\begingroup\@sanitize@url \@url }%
\providecommand \@url [1]{\endgroup\@href {#1}{\urlprefix }}%
\providecommand \urlprefix  [0]{URL }%
\providecommand \Eprint [0]{\href }%
\providecommand \doibase [0]{https://doi.org/}%
\providecommand \selectlanguage [0]{\@gobble}%
\providecommand \bibinfo  [0]{\@secondoftwo}%
\providecommand \bibfield  [0]{\@secondoftwo}%
\providecommand \translation [1]{[#1]}%
\providecommand \BibitemOpen [0]{}%
\providecommand \bibitemStop [0]{}%
\providecommand \bibitemNoStop [0]{.\EOS\space}%
\providecommand \EOS [0]{\spacefactor3000\relax}%
\providecommand \BibitemShut  [1]{\csname bibitem#1\endcsname}%
\let\auto@bib@innerbib\@empty
%</preamble>
\bibitem [{\citenamefont {Hastings}(2007)}]{hastings:arealaw}%
  \BibitemOpen
  \bibfield  {author} {\bibinfo {author} {\bibfnamefont {M.~B.}\ \bibnamefont
  {Hastings}},\ }\bibfield  {title} {\bibinfo {title} {An area law for
  one-dimensional quantum systems},\ }\href
  {https://doi.org/10.1088/1742-5468/2007/08/P08024} {\bibfield  {journal}
  {\bibinfo  {journal} {Journal of Statistical Mechanics: Theory and
  Experiment}\ }\textbf {\bibinfo {volume} {2007}},\ \bibinfo {pages} {P08024}
  (\bibinfo {year} {2007})},\ \Eprint {https://arxiv.org/abs/0705.2024}
  {arXiv:0705.2024} \BibitemShut {NoStop}%
\bibitem [{\citenamefont {Arad}\ \emph {et~al.}(2013)\citenamefont {Arad},
  \citenamefont {Kitaev}, \citenamefont {Landau},\ and\ \citenamefont
  {Vazirani}}]{arad2013area}%
  \BibitemOpen
  \bibfield  {author} {\bibinfo {author} {\bibfnamefont {I.}~\bibnamefont
  {Arad}}, \bibinfo {author} {\bibfnamefont {A.}~\bibnamefont {Kitaev}},
  \bibinfo {author} {\bibfnamefont {Z.}~\bibnamefont {Landau}},\ and\ \bibinfo
  {author} {\bibfnamefont {U.}~\bibnamefont {Vazirani}},\ }\href@noop {}
  {\bibinfo {title} {An area law and sub-exponential algorithm for {1D}
  systems}} (\bibinfo {year} {2013}),\ \Eprint
  {https://arxiv.org/abs/1301.1162} {arXiv:1301.1162 [quant-ph]} \BibitemShut
  {NoStop}%
\bibitem [{\citenamefont {Chen}\ \emph {et~al.}(2011)\citenamefont {Chen},
  \citenamefont {Gu},\ and\ \citenamefont {Wen}}]{chen2011complete}%
  \BibitemOpen
  \bibfield  {author} {\bibinfo {author} {\bibfnamefont {X.}~\bibnamefont
  {Chen}}, \bibinfo {author} {\bibfnamefont {Z.-C.}\ \bibnamefont {Gu}},\ and\
  \bibinfo {author} {\bibfnamefont {X.-G.}\ \bibnamefont {Wen}},\ }\bibfield
  {title} {\bibinfo {title} {Complete classification of one-dimensional gapped
  quantum phases in interacting spin systems},\ }\href
  {https://doi.org/10.1103/PhysRevB.84.235128} {\bibfield  {journal} {\bibinfo
  {journal} {Physical Review B}\ }\textbf {\bibinfo {volume} {84}},\ \bibinfo
  {pages} {235128} (\bibinfo {year} {2011})}\BibitemShut {NoStop}%
\bibitem [{\citenamefont {Schuch}\ \emph {et~al.}(2011)\citenamefont {Schuch},
  \citenamefont {P\'{e}rez-Garc\'{i}a},\ and\ \citenamefont
  {Cirac}}]{schuch2011classifying}%
  \BibitemOpen
  \bibfield  {author} {\bibinfo {author} {\bibfnamefont {N.}~\bibnamefont
  {Schuch}}, \bibinfo {author} {\bibfnamefont {D.}~\bibnamefont
  {P\'{e}rez-Garc\'{i}a}},\ and\ \bibinfo {author} {\bibfnamefont
  {I.}~\bibnamefont {Cirac}},\ }\bibfield  {title} {\bibinfo {title}
  {Classifying quantum phases using matrix product states and projected
  entangled pair states},\ }\href {https://doi.org/10.1103/PhysRevB.84.165139}
  {\bibfield  {journal} {\bibinfo  {journal} {Physical Review B}\ }\textbf
  {\bibinfo {volume} {84}},\ \bibinfo {pages} {165139} (\bibinfo {year}
  {2011})}\BibitemShut {NoStop}%
\bibitem [{\citenamefont {Hastings}(2006)}]{hastings2006solving}%
  \BibitemOpen
  \bibfield  {author} {\bibinfo {author} {\bibfnamefont {M.~B.}\ \bibnamefont
  {Hastings}},\ }\bibfield  {title} {\bibinfo {title} {Solving gapped
  {Hamiltonians} locally},\ }\href {https://doi.org/10.1103/PhysRevB.73.085115}
  {\bibfield  {journal} {\bibinfo  {journal} {Physical Review B}\ }\textbf
  {\bibinfo {volume} {73}},\ \bibinfo {pages} {085115} (\bibinfo {year}
  {2006})}\BibitemShut {NoStop}%
\bibitem [{\citenamefont {Molnar}\ \emph {et~al.}(2015)\citenamefont {Molnar},
  \citenamefont {Schuch}, \citenamefont {Verstraete},\ and\ \citenamefont
  {Cirac}}]{molnar2014approximating}%
  \BibitemOpen
  \bibfield  {author} {\bibinfo {author} {\bibfnamefont {A.}~\bibnamefont
  {Molnar}}, \bibinfo {author} {\bibfnamefont {N.}~\bibnamefont {Schuch}},
  \bibinfo {author} {\bibfnamefont {F.}~\bibnamefont {Verstraete}},\ and\
  \bibinfo {author} {\bibfnamefont {J.~I.}\ \bibnamefont {Cirac}},\ }\bibfield
  {title} {\bibinfo {title} {Approximating {Gibbs} states of local
  {Hamiltonians} efficiently with projected entangled pair states},\ }\href
  {https://doi.org/10.1103/PhysRevB.91.045138} {\bibfield  {journal} {\bibinfo
  {journal} {Physical Review B}\ }\textbf {\bibinfo {volume} {91}},\ \bibinfo
  {pages} {045138} (\bibinfo {year} {2015})}\BibitemShut {NoStop}%
\bibitem [{\citenamefont {Cirac}\ \emph {et~al.}(2011)\citenamefont {Cirac},
  \citenamefont {Poilblanc}, \citenamefont {Schuch},\ and\ \citenamefont
  {Verstraete}}]{cirac2011entanglement}%
  \BibitemOpen
  \bibfield  {author} {\bibinfo {author} {\bibfnamefont {J.~I.}\ \bibnamefont
  {Cirac}}, \bibinfo {author} {\bibfnamefont {D.}~\bibnamefont {Poilblanc}},
  \bibinfo {author} {\bibfnamefont {N.}~\bibnamefont {Schuch}},\ and\ \bibinfo
  {author} {\bibfnamefont {F.}~\bibnamefont {Verstraete}},\ }\bibfield  {title}
  {\bibinfo {title} {Entanglement spectrum and boundary theories with projected
  entangled-pair states},\ }\href {https://doi.org/10.1103/PhysRevB.83.245134}
  {\bibfield  {journal} {\bibinfo  {journal} {Physical Review B}\ }\textbf
  {\bibinfo {volume} {83}},\ \bibinfo {pages} {245134} (\bibinfo {year}
  {2011})}\BibitemShut {NoStop}%
\bibitem [{\citenamefont {Cirac}\ \emph
  {et~al.}(2017{\natexlab{a}})\citenamefont {Cirac}, \citenamefont
  {P\'{e}rez-Garc\'{i}a}, \citenamefont {Schuch},\ and\ \citenamefont
  {Verstraete}}]{cirac_matrix_2017}%
  \BibitemOpen
  \bibfield  {author} {\bibinfo {author} {\bibfnamefont {J.~I.}\ \bibnamefont
  {Cirac}}, \bibinfo {author} {\bibfnamefont {D.}~\bibnamefont
  {P\'{e}rez-Garc\'{i}a}}, \bibinfo {author} {\bibfnamefont {N.}~\bibnamefont
  {Schuch}},\ and\ \bibinfo {author} {\bibfnamefont {F.}~\bibnamefont
  {Verstraete}},\ }\bibfield  {title} {\bibinfo {title} {Matrix product density
  operators: {Renormalization} fixed points and boundary theories},\ }\href
  {https://doi.org/10.1016/j.aop.2016.12.030} {\bibfield  {journal} {\bibinfo
  {journal} {Annals of Physics}\ }\textbf {\bibinfo {volume} {378}},\ \bibinfo
  {pages} {100} (\bibinfo {year} {2017}{\natexlab{a}})}\BibitemShut {NoStop}%
\bibitem [{\citenamefont {Hastings}(2011)}]{Hastings_2011}%
  \BibitemOpen
  \bibfield  {author} {\bibinfo {author} {\bibfnamefont {M.~B.}\ \bibnamefont
  {Hastings}},\ }\bibfield  {title} {\bibinfo {title} {Topological order at
  nonzero temperature},\ }\href
  {https://doi.org/10.1103/PhysRevLett.107.210501} {\bibfield  {journal}
  {\bibinfo  {journal} {Physical Review Letters}\ }\textbf {\bibinfo {volume}
  {107}},\ \bibinfo {pages} {210501} (\bibinfo {year} {2011})}\BibitemShut
  {NoStop}%
\bibitem [{\citenamefont {de~Groot}\ \emph
  {et~al.}(2022{\natexlab{a}})\citenamefont {de~Groot}, \citenamefont
  {Turzillo},\ and\ \citenamefont {Schuch}}]{de2022symmetry}%
  \BibitemOpen
  \bibfield  {author} {\bibinfo {author} {\bibfnamefont {C.}~\bibnamefont
  {de~Groot}}, \bibinfo {author} {\bibfnamefont {A.}~\bibnamefont {Turzillo}},\
  and\ \bibinfo {author} {\bibfnamefont {N.}~\bibnamefont {Schuch}},\
  }\bibfield  {title} {\bibinfo {title} {Symmetry protected topological order
  in open quantum systems},\ }\href {https://doi.org/10.22331/q-2022-11-10-856}
  {\bibfield  {journal} {\bibinfo  {journal} {Quantum}\ }\textbf {\bibinfo
  {volume} {6}},\ \bibinfo {pages} {856} (\bibinfo {year}
  {2022}{\natexlab{a}})}\BibitemShut {NoStop}%
\bibitem [{\citenamefont {Ma}\ and\ \citenamefont
  {Wang}(2023)}]{ma2023average}%
  \BibitemOpen
  \bibfield  {author} {\bibinfo {author} {\bibfnamefont {R.}~\bibnamefont
  {Ma}}\ and\ \bibinfo {author} {\bibfnamefont {C.}~\bibnamefont {Wang}},\
  }\bibfield  {title} {\bibinfo {title} {Average symmetry-protected topological
  phases},\ }\href {https://doi.org/10.1103/PhysRevX.13.031016} {\bibfield
  {journal} {\bibinfo  {journal} {Physical Review X}\ }\textbf {\bibinfo
  {volume} {13}},\ \bibinfo {pages} {031016} (\bibinfo {year}
  {2023})}\BibitemShut {NoStop}%
\bibitem [{\citenamefont {Ruiz-de Alarc\'{o}n}\ \emph
  {et~al.}(2024)\citenamefont {Ruiz-de Alarc\'{o}n}, \citenamefont
  {Garre-Rubio}, \citenamefont {Moln\'{a}r},\ and\ \citenamefont
  {P\'{e}rez-Garc\'{i}a}}]{ruiz2024matrix}%
  \BibitemOpen
  \bibfield  {author} {\bibinfo {author} {\bibfnamefont {A.}~\bibnamefont
  {Ruiz-de Alarc\'{o}n}}, \bibinfo {author} {\bibfnamefont {J.}~\bibnamefont
  {Garre-Rubio}}, \bibinfo {author} {\bibfnamefont {A.}~\bibnamefont
  {Moln\'{a}r}},\ and\ \bibinfo {author} {\bibfnamefont {D.}~\bibnamefont
  {P\'{e}rez-Garc\'{i}a}},\ }\bibfield  {title} {\bibinfo {title} {Matrix
  product operator algebras {II}: {Phases} of matter for {1D} mixed states},\
  }\href {https://doi.org/10.1007/s11005-024-01778-z} {\bibfield  {journal}
  {\bibinfo  {journal} {Letters in Mathematical Physics}\ }\textbf {\bibinfo
  {volume} {114}},\ \bibinfo {pages} {43} (\bibinfo {year} {2024})}\BibitemShut
  {NoStop}%
\bibitem [{\citenamefont {Sun}(2025)}]{sun2025anomalous}%
  \BibitemOpen
  \bibfield  {author} {\bibinfo {author} {\bibfnamefont {X.-Q.}\ \bibnamefont
  {Sun}},\ }\href@noop {} {\bibinfo {title} {Anomalous matrix product operator
  symmetries and {1D} mixed-state phases}} (\bibinfo {year} {2025}),\ \Eprint
  {https://arxiv.org/abs/2504.16985} {arXiv:2504.16985 [quant-ph]} \BibitemShut
  {NoStop}%
\bibitem [{\citenamefont {Lessa}\ \emph {et~al.}(2025)\citenamefont {Lessa},
  \citenamefont {Cheng},\ and\ \citenamefont {Wang}}]{lessa2025mixed}%
  \BibitemOpen
  \bibfield  {author} {\bibinfo {author} {\bibfnamefont {L.~A.}\ \bibnamefont
  {Lessa}}, \bibinfo {author} {\bibfnamefont {M.}~\bibnamefont {Cheng}},\ and\
  \bibinfo {author} {\bibfnamefont {C.}~\bibnamefont {Wang}},\ }\bibfield
  {title} {\bibinfo {title} {Mixed-state quantum anomaly and multipartite
  entanglement},\ }\href {https://doi.org/10.1103/PhysRevX.15.011069}
  {\bibfield  {journal} {\bibinfo  {journal} {Physical Review X}\ }\textbf
  {\bibinfo {volume} {15}},\ \bibinfo {pages} {011069} (\bibinfo {year}
  {2025})}\BibitemShut {NoStop}%
\bibitem [{\citenamefont {Liu}\ \emph {et~al.}(2025)\citenamefont {Liu},
  \citenamefont {Molnar}, \citenamefont {Sun}, \citenamefont {Verstraete},
  \citenamefont {Kato},\ and\ \citenamefont {Lootens}}]{liu2025trading}%
  \BibitemOpen
  \bibfield  {author} {\bibinfo {author} {\bibfnamefont {Y.}~\bibnamefont
  {Liu}}, \bibinfo {author} {\bibfnamefont {A.}~\bibnamefont {Molnar}},
  \bibinfo {author} {\bibfnamefont {X.-Q.}\ \bibnamefont {Sun}}, \bibinfo
  {author} {\bibfnamefont {F.}~\bibnamefont {Verstraete}}, \bibinfo {author}
  {\bibfnamefont {K.}~\bibnamefont {Kato}},\ and\ \bibinfo {author}
  {\bibfnamefont {L.}~\bibnamefont {Lootens}},\ }\href@noop {} {\bibinfo
  {title} {Trading mathematical for physical simplicity: {Bialgebraic}
  structures in matrix product operator symmetries}} (\bibinfo {year} {2025}),\
  \Eprint {https://arxiv.org/abs/2509.03600} {arXiv:2509.03600 [quant-ph]}
  \BibitemShut {NoStop}%
\bibitem [{\citenamefont {Liu}\ \emph {et~al.}(2026)\citenamefont {Liu},
  \citenamefont {Ruiz-de Alarc\'{o}n}, \citenamefont {Styliaris}, \citenamefont
  {Sun}, \citenamefont {P\'{e}rez-Garc\'{i}a},\ and\ \citenamefont
  {Cirac}}]{liu2026parent}%
  \BibitemOpen
  \bibfield  {author} {\bibinfo {author} {\bibfnamefont {Y.}~\bibnamefont
  {Liu}}, \bibinfo {author} {\bibfnamefont {A.}~\bibnamefont {Ruiz-de
  Alarc\'{o}n}}, \bibinfo {author} {\bibfnamefont {G.}~\bibnamefont
  {Styliaris}}, \bibinfo {author} {\bibfnamefont {X.-Q.}\ \bibnamefont {Sun}},
  \bibinfo {author} {\bibfnamefont {D.}~\bibnamefont {P\'{e}rez-Garc\'{i}a}},\
  and\ \bibinfo {author} {\bibfnamefont {J.~I.}\ \bibnamefont {Cirac}},\
  }\bibfield  {title} {\bibinfo {title} {Parent {Lindbladians} for matrix
  product density operators},\ }\href {https://doi.org/10.1103/1qyd-59z7}
  {\bibfield  {journal} {\bibinfo  {journal} {Physical Review Research}\
  }\textbf {\bibinfo {volume} {8}},\ \bibinfo {pages} {013210} (\bibinfo {year}
  {2026})}\BibitemShut {NoStop}%
\bibitem [{\citenamefont {Wang}\ \emph {et~al.}(2015)\citenamefont {Wang},
  \citenamefont {Santos},\ and\ \citenamefont {Wen}}]{wang2015bosonic}%
  \BibitemOpen
  \bibfield  {author} {\bibinfo {author} {\bibfnamefont {J.~C.}\ \bibnamefont
  {Wang}}, \bibinfo {author} {\bibfnamefont {L.~H.}\ \bibnamefont {Santos}},\
  and\ \bibinfo {author} {\bibfnamefont {X.-G.}\ \bibnamefont {Wen}},\
  }\bibfield  {title} {\bibinfo {title} {Bosonic anomalies, induced fractional
  quantum numbers, and degenerate zero modes: {The} anomalous edge physics of
  symmetry-protected topological states},\ }\href
  {https://doi.org/10.1103/PhysRevB.91.195134} {\bibfield  {journal} {\bibinfo
  {journal} {Physical Review B}\ }\textbf {\bibinfo {volume} {91}},\ \bibinfo
  {pages} {195134} (\bibinfo {year} {2015})}\BibitemShut {NoStop}%
\bibitem [{\citenamefont {Roose}\ \emph {et~al.}(2019)\citenamefont {Roose},
  \citenamefont {Vanderstraeten}, \citenamefont {Haegeman},\ and\ \citenamefont
  {Bultinck}}]{roose2019anomalous}%
  \BibitemOpen
  \bibfield  {author} {\bibinfo {author} {\bibfnamefont {G.}~\bibnamefont
  {Roose}}, \bibinfo {author} {\bibfnamefont {L.}~\bibnamefont
  {Vanderstraeten}}, \bibinfo {author} {\bibfnamefont {J.}~\bibnamefont
  {Haegeman}},\ and\ \bibinfo {author} {\bibfnamefont {N.}~\bibnamefont
  {Bultinck}},\ }\bibfield  {title} {\bibinfo {title} {Anomalous domain wall
  condensation in a modified {Ising} chain},\ }\href
  {https://doi.org/10.1103/PhysRevB.99.195132} {\bibfield  {journal} {\bibinfo
  {journal} {Physical Review B}\ }\textbf {\bibinfo {volume} {99}},\ \bibinfo
  {pages} {195132} (\bibinfo {year} {2019})}\BibitemShut {NoStop}%
\bibitem [{\citenamefont {Garre-Rubio}\ \emph {et~al.}(2023)\citenamefont
  {Garre-Rubio}, \citenamefont {Lootens},\ and\ \citenamefont
  {Moln\'{a}r}}]{Garre_Rubio_2023}%
  \BibitemOpen
  \bibfield  {author} {\bibinfo {author} {\bibfnamefont {J.}~\bibnamefont
  {Garre-Rubio}}, \bibinfo {author} {\bibfnamefont {L.}~\bibnamefont
  {Lootens}},\ and\ \bibinfo {author} {\bibfnamefont {A.}~\bibnamefont
  {Moln\'{a}r}},\ }\bibfield  {title} {\bibinfo {title} {Classifying phases
  protected by matrix product operator symmetries using matrix product
  states},\ }\href {https://doi.org/10.22331/q-2023-02-21-927} {\bibfield
  {journal} {\bibinfo  {journal} {Quantum}\ }\textbf {\bibinfo {volume} {7}},\
  \bibinfo {pages} {927} (\bibinfo {year} {2023})}\BibitemShut {NoStop}%
\bibitem [{\citenamefont {Torlai}\ \emph {et~al.}(2020)\citenamefont {Torlai},
  \citenamefont {Wood}, \citenamefont {Acharya}, \citenamefont {Carleo},
  \citenamefont {Carrasquilla},\ and\ \citenamefont
  {Aolita}}]{torlai2020quantumprocesstomographyunsupervised}%
  \BibitemOpen
  \bibfield  {author} {\bibinfo {author} {\bibfnamefont {G.}~\bibnamefont
  {Torlai}}, \bibinfo {author} {\bibfnamefont {C.~J.}\ \bibnamefont {Wood}},
  \bibinfo {author} {\bibfnamefont {A.}~\bibnamefont {Acharya}}, \bibinfo
  {author} {\bibfnamefont {G.}~\bibnamefont {Carleo}}, \bibinfo {author}
  {\bibfnamefont {J.}~\bibnamefont {Carrasquilla}},\ and\ \bibinfo {author}
  {\bibfnamefont {L.}~\bibnamefont {Aolita}},\ }\href
  {https://arxiv.org/abs/2006.02424} {\bibinfo {title} {Quantum process
  tomography with unsupervised learning and tensor networks}} (\bibinfo {year}
  {2020}),\ \Eprint {https://arxiv.org/abs/2006.02424} {arXiv:2006.02424
  [quant-ph]} \BibitemShut {NoStop}%
\bibitem [{\citenamefont {Filippov}\ \emph {et~al.}(2022)\citenamefont
  {Filippov}, \citenamefont {Sokolov}, \citenamefont {Rossi}, \citenamefont
  {Malmi}, \citenamefont {Borrelli}, \citenamefont {Cavalcanti}, \citenamefont
  {Maniscalco},\ and\ \citenamefont
  {Garc\'{i}a-P\'{e}rez}}]{filippov2022matrixproductchannelvariationally}%
  \BibitemOpen
  \bibfield  {author} {\bibinfo {author} {\bibfnamefont {S.}~\bibnamefont
  {Filippov}}, \bibinfo {author} {\bibfnamefont {B.}~\bibnamefont {Sokolov}},
  \bibinfo {author} {\bibfnamefont {M.~A.~C.}\ \bibnamefont {Rossi}}, \bibinfo
  {author} {\bibfnamefont {J.}~\bibnamefont {Malmi}}, \bibinfo {author}
  {\bibfnamefont {E.-M.}\ \bibnamefont {Borrelli}}, \bibinfo {author}
  {\bibfnamefont {D.}~\bibnamefont {Cavalcanti}}, \bibinfo {author}
  {\bibfnamefont {S.}~\bibnamefont {Maniscalco}},\ and\ \bibinfo {author}
  {\bibfnamefont {G.}~\bibnamefont {Garc\'{i}a-P\'{e}rez}},\ }\href
  {https://arxiv.org/abs/2212.10225} {\bibinfo {title} {Matrix product channel:
  {Variationally} optimized quantum tensor network to mitigate noise and reduce
  errors for the variational quantum eigensolver}} (\bibinfo {year} {2022}),\
  \Eprint {https://arxiv.org/abs/2212.10225} {arXiv:2212.10225 [quant-ph]}
  \BibitemShut {NoStop}%
\bibitem [{\citenamefont {Liu}\ and\ \citenamefont
  {Wild}(2026)}]{liu2026efficientlylearningglobalquantum}%
  \BibitemOpen
  \bibfield  {author} {\bibinfo {author} {\bibfnamefont {Z.}~\bibnamefont
  {Liu}}\ and\ \bibinfo {author} {\bibfnamefont {D.~S.}\ \bibnamefont {Wild}},\
  }\href {https://arxiv.org/abs/2603.07037} {\bibinfo {title} {Efficiently
  learning global quantum channels with local tomography}} (\bibinfo {year}
  {2026}),\ \Eprint {https://arxiv.org/abs/2603.07037} {arXiv:2603.07037
  [quant-ph]} \BibitemShut {NoStop}%
\bibitem [{\citenamefont {Cirac}\ \emph
  {et~al.}(2017{\natexlab{b}})\citenamefont {Cirac}, \citenamefont
  {Perez-Garcia}, \citenamefont {Schuch},\ and\ \citenamefont
  {Verstraete}}]{cirac_matrix_2017-1}%
  \BibitemOpen
  \bibfield  {author} {\bibinfo {author} {\bibfnamefont {J.~I.}\ \bibnamefont
  {Cirac}}, \bibinfo {author} {\bibfnamefont {D.}~\bibnamefont {Perez-Garcia}},
  \bibinfo {author} {\bibfnamefont {N.}~\bibnamefont {Schuch}},\ and\ \bibinfo
  {author} {\bibfnamefont {F.}~\bibnamefont {Verstraete}},\ }\bibfield  {title}
  {\bibinfo {title} {Matrix product unitaries: {Structure}, symmetries, and
  topological invariants},\ }\href {https://doi.org/10.1088/1742-5468/aa7e55}
  {\bibfield  {journal} {\bibinfo  {journal} {Journal of Statistical Mechanics:
  Theory and Experiment}\ }\textbf {\bibinfo {volume} {2017}},\ \bibinfo
  {pages} {083105} (\bibinfo {year} {2017}{\natexlab{b}})}\BibitemShut
  {NoStop}%
\bibitem [{\citenamefont {\c{S}ahino\v{g}lu}\ \emph {et~al.}(2018)\citenamefont
  {\c{S}ahino\v{g}lu}, \citenamefont {Shukla}, \citenamefont {Bi},\ and\
  \citenamefont {Chen}}]{sahinoglu2018matrix}%
  \BibitemOpen
  \bibfield  {author} {\bibinfo {author} {\bibfnamefont {M.~B.}\ \bibnamefont
  {\c{S}ahino\v{g}lu}}, \bibinfo {author} {\bibfnamefont {S.~K.}\ \bibnamefont
  {Shukla}}, \bibinfo {author} {\bibfnamefont {F.}~\bibnamefont {Bi}},\ and\
  \bibinfo {author} {\bibfnamefont {X.}~\bibnamefont {Chen}},\ }\bibfield
  {title} {\bibinfo {title} {Matrix product representation of locality
  preserving unitaries},\ }\href {https://doi.org/10.1103/PhysRevB.98.245122}
  {\bibfield  {journal} {\bibinfo  {journal} {Physical Review B}\ }\textbf
  {\bibinfo {volume} {98}},\ \bibinfo {pages} {245122} (\bibinfo {year}
  {2018})}\BibitemShut {NoStop}%
\bibitem [{\citenamefont {Piroli}\ \emph
  {et~al.}(2021{\natexlab{a}})\citenamefont {Piroli}, \citenamefont {Turzillo},
  \citenamefont {Shukla},\ and\ \citenamefont {Cirac}}]{piroli2021fermionic}%
  \BibitemOpen
  \bibfield  {author} {\bibinfo {author} {\bibfnamefont {L.}~\bibnamefont
  {Piroli}}, \bibinfo {author} {\bibfnamefont {A.}~\bibnamefont {Turzillo}},
  \bibinfo {author} {\bibfnamefont {S.~K.}\ \bibnamefont {Shukla}},\ and\
  \bibinfo {author} {\bibfnamefont {J.~I.}\ \bibnamefont {Cirac}},\ }\bibfield
  {title} {\bibinfo {title} {Fermionic quantum cellular automata and
  generalized matrix-product unitaries},\ }\href
  {https://doi.org/10.1088/1742-5468/abd30f} {\bibfield  {journal} {\bibinfo
  {journal} {Journal of Statistical Mechanics: Theory and Experiment}\ }\textbf
  {\bibinfo {volume} {2021}},\ \bibinfo {pages} {013107} (\bibinfo {year}
  {2021}{\natexlab{a}})}\BibitemShut {NoStop}%
\bibitem [{\citenamefont {Shukla}(2025)}]{shukla2025simple}%
  \BibitemOpen
  \bibfield  {author} {\bibinfo {author} {\bibfnamefont {S.~K.}\ \bibnamefont
  {Shukla}},\ }\bibfield  {title} {\bibinfo {title} {A simple and general
  equation for matrix product unitary generation},\ }\bibfield  {journal}
  {\bibinfo  {journal} {Journal of Mathematical Physics}\ }\textbf {\bibinfo
  {volume} {66}},\ \href {https://doi.org/10.1088/1742-5468/abd30f}
  {10.1088/1742-5468/abd30f} (\bibinfo {year} {2025})\BibitemShut {NoStop}%
\bibitem [{\citenamefont {Styliaris}\ \emph
  {et~al.}(2025{\natexlab{a}})\citenamefont {Styliaris}, \citenamefont
  {Trivedi}, \citenamefont {P\'{e}rez-Garc\'{i}a},\ and\ \citenamefont
  {Cirac}}]{styliaris2025matrix}%
  \BibitemOpen
  \bibfield  {author} {\bibinfo {author} {\bibfnamefont {G.}~\bibnamefont
  {Styliaris}}, \bibinfo {author} {\bibfnamefont {R.}~\bibnamefont {Trivedi}},
  \bibinfo {author} {\bibfnamefont {D.}~\bibnamefont {P\'{e}rez-Garc\'{i}a}},\
  and\ \bibinfo {author} {\bibfnamefont {J.~I.}\ \bibnamefont {Cirac}},\
  }\bibfield  {title} {\bibinfo {title} {Matrix-product unitaries: {Beyond}
  quantum cellular automata},\ }\href
  {https://doi.org/10.22331/q-2025-02-25-1645} {\bibfield  {journal} {\bibinfo
  {journal} {Quantum}\ }\textbf {\bibinfo {volume} {9}},\ \bibinfo {pages}
  {1645} (\bibinfo {year} {2025}{\natexlab{a}})}\BibitemShut {NoStop}%
\bibitem [{\citenamefont {Franco-Rubio}\ \emph {et~al.}(2025)\citenamefont
  {Franco-Rubio}, \citenamefont {Bochniak},\ and\ \citenamefont
  {Cirac}}]{francorubio2025symmetrydefectsgaugingquantum}%
  \BibitemOpen
  \bibfield  {author} {\bibinfo {author} {\bibfnamefont {A.}~\bibnamefont
  {Franco-Rubio}}, \bibinfo {author} {\bibfnamefont {A.}~\bibnamefont
  {Bochniak}},\ and\ \bibinfo {author} {\bibfnamefont {J.~I.}\ \bibnamefont
  {Cirac}},\ }\href {https://arxiv.org/abs/2502.20257} {\bibinfo {title}
  {Symmetry defects and gauging for quantum states with matrix product unitary
  symmetries}} (\bibinfo {year} {2025}),\ \Eprint
  {https://arxiv.org/abs/2502.20257} {arXiv:2502.20257 [quant-ph]} \BibitemShut
  {NoStop}%
\bibitem [{\citenamefont {Lootens}\ \emph {et~al.}(2025)\citenamefont
  {Lootens}, \citenamefont {Delcamp}, \citenamefont {Williamson},\ and\
  \citenamefont {Verstraete}}]{Lootens_2025}%
  \BibitemOpen
  \bibfield  {author} {\bibinfo {author} {\bibfnamefont {L.}~\bibnamefont
  {Lootens}}, \bibinfo {author} {\bibfnamefont {C.}~\bibnamefont {Delcamp}},
  \bibinfo {author} {\bibfnamefont {D.}~\bibnamefont {Williamson}},\ and\
  \bibinfo {author} {\bibfnamefont {F.}~\bibnamefont {Verstraete}},\ }\bibfield
   {title} {\bibinfo {title} {Low-depth unitary quantum circuits for dualities
  in one-dimensional quantum lattice models},\ }\href
  {https://doi.org/10.1103/PhysRevLett.134.130403} {\bibfield  {journal}
  {\bibinfo  {journal} {Physical Review Letters}\ }\textbf {\bibinfo {volume}
  {134}},\ \bibinfo {pages} {130403} (\bibinfo {year} {2025})}\BibitemShut
  {NoStop}%
\bibitem [{\citenamefont {Arrighi}(2019)}]{arrighi_overview_2019}%
  \BibitemOpen
  \bibfield  {author} {\bibinfo {author} {\bibfnamefont {P.}~\bibnamefont
  {Arrighi}},\ }\bibfield  {title} {\bibinfo {title} {An overview of quantum
  cellular automata},\ }\href {https://doi.org/10.1007/s11047-019-09762-6}
  {\bibfield  {journal} {\bibinfo  {journal} {Natural Computing}\ }\textbf
  {\bibinfo {volume} {18}},\ \bibinfo {pages} {885} (\bibinfo {year}
  {2019})}\BibitemShut {NoStop}%
\bibitem [{\citenamefont {Gross}\ \emph {et~al.}(2012)\citenamefont {Gross},
  \citenamefont {Nesme}, \citenamefont {Vogts},\ and\ \citenamefont
  {Werner}}]{gross_index_2012}%
  \BibitemOpen
  \bibfield  {author} {\bibinfo {author} {\bibfnamefont {D.}~\bibnamefont
  {Gross}}, \bibinfo {author} {\bibfnamefont {V.}~\bibnamefont {Nesme}},
  \bibinfo {author} {\bibfnamefont {H.}~\bibnamefont {Vogts}},\ and\ \bibinfo
  {author} {\bibfnamefont {R.~F.}\ \bibnamefont {Werner}},\ }\bibfield  {title}
  {\bibinfo {title} {Index theory of one dimensional quantum walks and cellular
  automata},\ }\href {https://doi.org/10.1007/s00220-012-1423-1} {\bibfield
  {journal} {\bibinfo  {journal} {Communications in Mathematical Physics}\
  }\textbf {\bibinfo {volume} {310}},\ \bibinfo {pages} {419} (\bibinfo {year}
  {2012})}\BibitemShut {NoStop}%
\bibitem [{\citenamefont {Farrelly}\ and\ \citenamefont
  {Short}(2014)}]{farrelly_causal_2014}%
  \BibitemOpen
  \bibfield  {author} {\bibinfo {author} {\bibfnamefont {T.~C.}\ \bibnamefont
  {Farrelly}}\ and\ \bibinfo {author} {\bibfnamefont {A.~J.}\ \bibnamefont
  {Short}},\ }\bibfield  {title} {\bibinfo {title} {Causal fermions in discrete
  space-time},\ }\href {https://doi.org/10.1103/PhysRevA.89.012302} {\bibfield
  {journal} {\bibinfo  {journal} {Physical Review A}\ }\textbf {\bibinfo
  {volume} {89}},\ \bibinfo {pages} {012302} (\bibinfo {year}
  {2014})}\BibitemShut {NoStop}%
\bibitem [{\citenamefont {Styliaris}\ \emph
  {et~al.}(2025{\natexlab{b}})\citenamefont {Styliaris}, \citenamefont
  {Trivedi},\ and\ \citenamefont
  {Cirac}}]{styliaris2025quantumcircuitcomplexitymatrixproduct}%
  \BibitemOpen
  \bibfield  {author} {\bibinfo {author} {\bibfnamefont {G.}~\bibnamefont
  {Styliaris}}, \bibinfo {author} {\bibfnamefont {R.}~\bibnamefont {Trivedi}},\
  and\ \bibinfo {author} {\bibfnamefont {J.~I.}\ \bibnamefont {Cirac}},\
  }\bibfield  {title} {\bibinfo {title} {Quantum circuits for matrix-product
  unitaries},\ }\bibfield  {journal} {\bibinfo  {journal} {Physical Review
  Letters}\ }\textbf {\bibinfo {volume} {135}},\ \href
  {https://doi.org/10.1103/yshb-hmml} {10.1103/yshb-hmml} (\bibinfo {year}
  {2025}{\natexlab{b}})\BibitemShut {NoStop}%
\bibitem [{\citenamefont {Vidal}(2003)}]{vidal2003efficient}%
  \BibitemOpen
  \bibfield  {author} {\bibinfo {author} {\bibfnamefont {G.}~\bibnamefont
  {Vidal}},\ }\bibfield  {title} {\bibinfo {title} {Efficient classical
  simulation of slightly entangled quantum computations},\ }\href
  {https://doi.org/10.1103/PhysRevLett.91.147902} {\bibfield  {journal}
  {\bibinfo  {journal} {Physical Review Letters}\ }\textbf {\bibinfo {volume}
  {91}},\ \bibinfo {pages} {147902} (\bibinfo {year} {2003})}\BibitemShut
  {NoStop}%
\bibitem [{\citenamefont {Sch\"{o}n}\ \emph {et~al.}(2005)\citenamefont
  {Sch\"{o}n}, \citenamefont {Solano}, \citenamefont {Verstraete},
  \citenamefont {Cirac},\ and\ \citenamefont {Wolf}}]{Sch_n_2005}%
  \BibitemOpen
  \bibfield  {author} {\bibinfo {author} {\bibfnamefont {C.}~\bibnamefont
  {Sch\"{o}n}}, \bibinfo {author} {\bibfnamefont {E.}~\bibnamefont {Solano}},
  \bibinfo {author} {\bibfnamefont {F.}~\bibnamefont {Verstraete}}, \bibinfo
  {author} {\bibfnamefont {J.~I.}\ \bibnamefont {Cirac}},\ and\ \bibinfo
  {author} {\bibfnamefont {M.~M.}\ \bibnamefont {Wolf}},\ }\bibfield  {title}
  {\bibinfo {title} {Sequential generation of entangled multiqubit states},\
  }\href {https://doi.org/10.1103/PhysRevLett.95.110503} {\bibfield  {journal}
  {\bibinfo  {journal} {Physical Review Letters}\ }\textbf {\bibinfo {volume}
  {95}},\ \bibinfo {pages} {110503} (\bibinfo {year} {2005})}\BibitemShut
  {NoStop}%
\bibitem [{sm()}]{sm}%
  \BibitemOpen
  \href@noop {} {}\bibinfo {note} {See Supplemental Material for further
  details.}\BibitemShut {Stop}%
\bibitem [{\citenamefont {Cuevas}\ \emph {et~al.}(2016)\citenamefont {Cuevas},
  \citenamefont {Cubitt}, \citenamefont {Cirac}, \citenamefont {Wolf},\ and\
  \citenamefont {P\'{e}rez-Garc\'{i}a}}]{cuevas_fundamental_2016}%
  \BibitemOpen
  \bibfield  {author} {\bibinfo {author} {\bibfnamefont {G.~D.~l.}\
  \bibnamefont {Cuevas}}, \bibinfo {author} {\bibfnamefont {T.~S.}\
  \bibnamefont {Cubitt}}, \bibinfo {author} {\bibfnamefont {J.~I.}\
  \bibnamefont {Cirac}}, \bibinfo {author} {\bibfnamefont {M.~M.}\ \bibnamefont
  {Wolf}},\ and\ \bibinfo {author} {\bibfnamefont {D.}~\bibnamefont
  {P\'{e}rez-Garc\'{i}a}},\ }\bibfield  {title} {\bibinfo {title} {Fundamental
  limitations in the purifications of tensor networks},\ }\href
  {https://doi.org/10.1063/1.4954983} {\bibfield  {journal} {\bibinfo
  {journal} {Journal of Mathematical Physics}\ }\textbf {\bibinfo {volume}
  {57}},\ \bibinfo {pages} {071902} (\bibinfo {year} {2016})}\BibitemShut
  {NoStop}%
\bibitem [{\citenamefont {Vanrietvelde}\ \emph {et~al.}(2025)\citenamefont
  {Vanrietvelde}, \citenamefont {Mestoudjian},\ and\ \citenamefont
  {Arrighi}}]{vanrietvelde2025causal}%
  \BibitemOpen
  \bibfield  {author} {\bibinfo {author} {\bibfnamefont {A.}~\bibnamefont
  {Vanrietvelde}}, \bibinfo {author} {\bibfnamefont {O.}~\bibnamefont
  {Mestoudjian}},\ and\ \bibinfo {author} {\bibfnamefont {P.}~\bibnamefont
  {Arrighi}},\ }\href@noop {} {\bibinfo {title} {Causal decompositions of {1D}
  quantum cellular automata}} (\bibinfo {year} {2025}),\ \Eprint
  {https://arxiv.org/abs/2506.22219} {arXiv:2506.22219 [quant-ph]} \BibitemShut
  {NoStop}%
\bibitem [{\citenamefont {Cirac}\ \emph {et~al.}(2021)\citenamefont {Cirac},
  \citenamefont {P\'{e}rez-Garc\'{i}a}, \citenamefont {Schuch},\ and\
  \citenamefont {Verstraete}}]{cirac_matrix_2021}%
  \BibitemOpen
  \bibfield  {author} {\bibinfo {author} {\bibfnamefont {J.~I.}\ \bibnamefont
  {Cirac}}, \bibinfo {author} {\bibfnamefont {D.}~\bibnamefont
  {P\'{e}rez-Garc\'{i}a}}, \bibinfo {author} {\bibfnamefont {N.}~\bibnamefont
  {Schuch}},\ and\ \bibinfo {author} {\bibfnamefont {F.}~\bibnamefont
  {Verstraete}},\ }\bibfield  {title} {\bibinfo {title} {Matrix product states
  and projected entangled pair states: {Concepts}, symmetries, theorems},\
  }\href {https://doi.org/10.1103/RevModPhys.93.045003} {\bibfield  {journal}
  {\bibinfo  {journal} {Reviews of Modern Physics}\ }\textbf {\bibinfo {volume}
  {93}},\ \bibinfo {pages} {045003} (\bibinfo {year} {2021})}\BibitemShut
  {NoStop}%
\bibitem [{\citenamefont {de~Groot}\ \emph
  {et~al.}(2022{\natexlab{b}})\citenamefont {de~Groot}, \citenamefont
  {Turzillo},\ and\ \citenamefont {Schuch}}]{de_Groot_2022}%
  \BibitemOpen
  \bibfield  {author} {\bibinfo {author} {\bibfnamefont {C.}~\bibnamefont
  {de~Groot}}, \bibinfo {author} {\bibfnamefont {A.}~\bibnamefont {Turzillo}},\
  and\ \bibinfo {author} {\bibfnamefont {N.}~\bibnamefont {Schuch}},\
  }\bibfield  {title} {\bibinfo {title} {Symmetry protected topological order
  in open quantum systems},\ }\href {https://doi.org/10.22331/q-2022-11-10-856}
  {\bibfield  {journal} {\bibinfo  {journal} {Quantum}\ }\textbf {\bibinfo
  {volume} {6}},\ \bibinfo {pages} {856} (\bibinfo {year}
  {2022}{\natexlab{b}})}\BibitemShut {NoStop}%
\bibitem [{\citenamefont {Schumacher}\ and\ \citenamefont
  {Werner}(2004)}]{schumacher_reversible_2004}%
  \BibitemOpen
  \bibfield  {author} {\bibinfo {author} {\bibfnamefont {B.}~\bibnamefont
  {Schumacher}}\ and\ \bibinfo {author} {\bibfnamefont {R.~F.}\ \bibnamefont
  {Werner}},\ }\href@noop {} {\bibinfo {title} {Reversible quantum cellular
  automata}} (\bibinfo {year} {2004}),\ \Eprint
  {https://arxiv.org/abs/quant-ph/0405174} {arXiv:quant-ph/0405174}
  \BibitemShut {NoStop}%
\bibitem [{\citenamefont {Arrighi}\ \emph {et~al.}(2011)\citenamefont
  {Arrighi}, \citenamefont {Nesme},\ and\ \citenamefont
  {Werner}}]{arrighi_unitarity_2009}%
  \BibitemOpen
  \bibfield  {author} {\bibinfo {author} {\bibfnamefont {P.}~\bibnamefont
  {Arrighi}}, \bibinfo {author} {\bibfnamefont {V.}~\bibnamefont {Nesme}},\
  and\ \bibinfo {author} {\bibfnamefont {R.~F.}\ \bibnamefont {Werner}},\
  }\bibfield  {title} {\bibinfo {title} {Unitarity plus causality implies
  localizability},\ }\href {https://doi.org/10.1016/j.jcss.2010.05.004}
  {\bibfield  {journal} {\bibinfo  {journal} {Journal of Computer and System
  Sciences}\ }\textbf {\bibinfo {volume} {77}},\ \bibinfo {pages} {372}
  (\bibinfo {year} {2011})}\BibitemShut {NoStop}%
\bibitem [{\citenamefont {Raussendorf}\ \emph {et~al.}(2005)\citenamefont
  {Raussendorf}, \citenamefont {Bravyi},\ and\ \citenamefont
  {Harrington}}]{Raussendorf_2005}%
  \BibitemOpen
  \bibfield  {author} {\bibinfo {author} {\bibfnamefont {R.}~\bibnamefont
  {Raussendorf}}, \bibinfo {author} {\bibfnamefont {S.}~\bibnamefont
  {Bravyi}},\ and\ \bibinfo {author} {\bibfnamefont {J.}~\bibnamefont
  {Harrington}},\ }\bibfield  {title} {\bibinfo {title} {Long-range quantum
  entanglement in noisy cluster states},\ }\href
  {https://doi.org/10.1103/PhysRevA.71.062313} {\bibfield  {journal} {\bibinfo
  {journal} {Physical Review A}\ }\textbf {\bibinfo {volume} {71}},\ \bibinfo
  {pages} {062313} (\bibinfo {year} {2005})}\BibitemShut {NoStop}%
\bibitem [{\citenamefont {Broadbent}\ and\ \citenamefont
  {Kashefi}(2009)}]{Broadbent_2009}%
  \BibitemOpen
  \bibfield  {author} {\bibinfo {author} {\bibfnamefont {A.}~\bibnamefont
  {Broadbent}}\ and\ \bibinfo {author} {\bibfnamefont {E.}~\bibnamefont
  {Kashefi}},\ }\bibfield  {title} {\bibinfo {title} {Parallelizing quantum
  circuits},\ }\href {https://doi.org/10.1016/j.tcs.2008.12.046} {\bibfield
  {journal} {\bibinfo  {journal} {Theoretical Computer Science}\ }\textbf
  {\bibinfo {volume} {410}},\ \bibinfo {pages} {2489} (\bibinfo {year}
  {2009})}\BibitemShut {NoStop}%
\bibitem [{\citenamefont {Watts}\ \emph {et~al.}(2019)\citenamefont {Watts},
  \citenamefont {Kothari}, \citenamefont {Schaeffer},\ and\ \citenamefont
  {Tal}}]{Watts_2019}%
  \BibitemOpen
  \bibfield  {author} {\bibinfo {author} {\bibfnamefont {A.~B.}\ \bibnamefont
  {Watts}}, \bibinfo {author} {\bibfnamefont {R.}~\bibnamefont {Kothari}},
  \bibinfo {author} {\bibfnamefont {L.}~\bibnamefont {Schaeffer}},\ and\
  \bibinfo {author} {\bibfnamefont {A.}~\bibnamefont {Tal}},\ }\bibfield
  {title} {\bibinfo {title} {Exponential separation between shallow quantum
  circuits and unbounded fan-in shallow classical circuits},\ }in\ \href
  {https://doi.org/10.1145/3313276.3316404} {\emph {\bibinfo {booktitle}
  {Proceedings of the 51st Annual {ACM} {SIGACT} Symposium on Theory of
  Computing}}},\ \bibinfo {series and number} {STOC '19}\ (\bibinfo
  {publisher} {ACM},\ \bibinfo {year} {2019})\ pp.\ \bibinfo {pages}
  {515--526}\BibitemShut {NoStop}%
\bibitem [{\citenamefont {Piroli}\ \emph
  {et~al.}(2021{\natexlab{b}})\citenamefont {Piroli}, \citenamefont
  {Styliaris},\ and\ \citenamefont {Cirac}}]{Piroli_2021}%
  \BibitemOpen
  \bibfield  {author} {\bibinfo {author} {\bibfnamefont {L.}~\bibnamefont
  {Piroli}}, \bibinfo {author} {\bibfnamefont {G.}~\bibnamefont {Styliaris}},\
  and\ \bibinfo {author} {\bibfnamefont {J.~I.}\ \bibnamefont {Cirac}},\
  }\bibfield  {title} {\bibinfo {title} {Quantum circuits assisted by local
  operations and classical communication: {Transformations} and phases of
  matter},\ }\href {https://doi.org/10.1103/PhysRevLett.127.220503} {\bibfield
  {journal} {\bibinfo  {journal} {Physical Review Letters}\ }\textbf {\bibinfo
  {volume} {127}},\ \bibinfo {pages} {220503} (\bibinfo {year}
  {2021}{\natexlab{b}})}\BibitemShut {NoStop}%
\bibitem [{\citenamefont {Lu}\ \emph {et~al.}(2022)\citenamefont {Lu},
  \citenamefont {Lessa}, \citenamefont {Kim},\ and\ \citenamefont
  {Hsieh}}]{Lu_2022}%
  \BibitemOpen
  \bibfield  {author} {\bibinfo {author} {\bibfnamefont {T.-C.}\ \bibnamefont
  {Lu}}, \bibinfo {author} {\bibfnamefont {L.~A.}\ \bibnamefont {Lessa}},
  \bibinfo {author} {\bibfnamefont {I.~H.}\ \bibnamefont {Kim}},\ and\ \bibinfo
  {author} {\bibfnamefont {T.~H.}\ \bibnamefont {Hsieh}},\ }\bibfield  {title}
  {\bibinfo {title} {Measurement as a shortcut to long-range entangled quantum
  matter},\ }\href {https://doi.org/10.1103/PRXQuantum.3.040337} {\bibfield
  {journal} {\bibinfo  {journal} {PRX Quantum}\ }\textbf {\bibinfo {volume}
  {3}},\ \bibinfo {pages} {040337} (\bibinfo {year} {2022})}\BibitemShut
  {NoStop}%
\bibitem [{\citenamefont {Smith}\ \emph {et~al.}(2023)\citenamefont {Smith},
  \citenamefont {Crane}, \citenamefont {Wiebe},\ and\ \citenamefont
  {Girvin}}]{smith2023deterministic}%
  \BibitemOpen
  \bibfield  {author} {\bibinfo {author} {\bibfnamefont {K.~C.}\ \bibnamefont
  {Smith}}, \bibinfo {author} {\bibfnamefont {E.}~\bibnamefont {Crane}},
  \bibinfo {author} {\bibfnamefont {N.}~\bibnamefont {Wiebe}},\ and\ \bibinfo
  {author} {\bibfnamefont {S.~M.}\ \bibnamefont {Girvin}},\ }\bibfield  {title}
  {\bibinfo {title} {Deterministic constant-depth preparation of the {AKLT}
  state on a quantum processor using fusion measurements},\ }\href
  {https://doi.org/10.1103/PRXQuantum.4.020315} {\bibfield  {journal} {\bibinfo
   {journal} {PRX Quantum}\ }\textbf {\bibinfo {volume} {4}},\ \bibinfo {pages}
  {020315} (\bibinfo {year} {2023})}\BibitemShut {NoStop}%
\bibitem [{\citenamefont {B\"{a}umer}\ \emph {et~al.}(2024)\citenamefont
  {B\"{a}umer}, \citenamefont {Tripathi}, \citenamefont {Wang}, \citenamefont
  {Rall}, \citenamefont {Chen}, \citenamefont {Majumder}, \citenamefont
  {Seif},\ and\ \citenamefont {Minev}}]{baumer2024efficient}%
  \BibitemOpen
  \bibfield  {author} {\bibinfo {author} {\bibfnamefont {E.}~\bibnamefont
  {B\"{a}umer}}, \bibinfo {author} {\bibfnamefont {V.}~\bibnamefont
  {Tripathi}}, \bibinfo {author} {\bibfnamefont {D.~S.}\ \bibnamefont {Wang}},
  \bibinfo {author} {\bibfnamefont {P.}~\bibnamefont {Rall}}, \bibinfo {author}
  {\bibfnamefont {E.~H.}\ \bibnamefont {Chen}}, \bibinfo {author}
  {\bibfnamefont {S.}~\bibnamefont {Majumder}}, \bibinfo {author}
  {\bibfnamefont {A.}~\bibnamefont {Seif}},\ and\ \bibinfo {author}
  {\bibfnamefont {Z.~K.}\ \bibnamefont {Minev}},\ }\bibfield  {title} {\bibinfo
  {title} {Efficient long-range entanglement using dynamic circuits},\ }\href
  {https://doi.org/10.1103/PRXQuantum.5.030339} {\bibfield  {journal} {\bibinfo
   {journal} {PRX Quantum}\ }\textbf {\bibinfo {volume} {5}},\ \bibinfo {pages}
  {030339} (\bibinfo {year} {2024})}\BibitemShut {NoStop}%
\bibitem [{\citenamefont {Buhrman}\ \emph {et~al.}(2024)\citenamefont
  {Buhrman}, \citenamefont {Folkertsma}, \citenamefont {Loff},\ and\
  \citenamefont {Neumann}}]{buhrman2024state}%
  \BibitemOpen
  \bibfield  {author} {\bibinfo {author} {\bibfnamefont {H.}~\bibnamefont
  {Buhrman}}, \bibinfo {author} {\bibfnamefont {M.}~\bibnamefont {Folkertsma}},
  \bibinfo {author} {\bibfnamefont {B.}~\bibnamefont {Loff}},\ and\ \bibinfo
  {author} {\bibfnamefont {N.~M.~P.}\ \bibnamefont {Neumann}},\ }\bibfield
  {title} {\bibinfo {title} {State preparation by shallow circuits using feed
  forward},\ }\href {https://doi.org/10.22331/q-2024-12-09-1552} {\bibfield
  {journal} {\bibinfo  {journal} {Quantum}\ }\textbf {\bibinfo {volume} {8}},\
  \bibinfo {pages} {1552} (\bibinfo {year} {2024})}\BibitemShut {NoStop}%
\bibitem [{\citenamefont {Stephen}\ and\ \citenamefont
  {Hart}(2024)}]{stephen2024preparing}%
  \BibitemOpen
  \bibfield  {author} {\bibinfo {author} {\bibfnamefont {D.~T.}\ \bibnamefont
  {Stephen}}\ and\ \bibinfo {author} {\bibfnamefont {O.}~\bibnamefont {Hart}},\
  }\href@noop {} {\bibinfo {title} {Preparing matrix product states via fusion:
  {Constraints} and extensions}} (\bibinfo {year} {2024}),\ \Eprint
  {https://arxiv.org/abs/2404.16360} {arXiv:2404.16360 [quant-ph]} \BibitemShut
  {NoStop}%
\bibitem [{\citenamefont {Zhang}\ \emph {et~al.}(2024)\citenamefont {Zhang},
  \citenamefont {Gopalakrishnan},\ and\ \citenamefont
  {Styliaris}}]{zhang2024characterizing}%
  \BibitemOpen
  \bibfield  {author} {\bibinfo {author} {\bibfnamefont {Y.}~\bibnamefont
  {Zhang}}, \bibinfo {author} {\bibfnamefont {S.}~\bibnamefont
  {Gopalakrishnan}},\ and\ \bibinfo {author} {\bibfnamefont {G.}~\bibnamefont
  {Styliaris}},\ }\bibfield  {title} {\bibinfo {title} {Characterizing
  matrix-product states and projected entangled-pair states preparable via
  measurement and feedback},\ }\href
  {https://doi.org/10.1103/PRXQuantum.5.040304} {\bibfield  {journal} {\bibinfo
   {journal} {PRX Quantum}\ }\textbf {\bibinfo {volume} {5}},\ \bibinfo {pages}
  {040304} (\bibinfo {year} {2024})}\BibitemShut {NoStop}%
\bibitem [{\citenamefont {Tantivasadakarn}\ \emph {et~al.}(2024)\citenamefont
  {Tantivasadakarn}, \citenamefont {Thorngren}, \citenamefont {Vishwanath},\
  and\ \citenamefont {Verresen}}]{Tantivasadakarn_2024}%
  \BibitemOpen
  \bibfield  {author} {\bibinfo {author} {\bibfnamefont {N.}~\bibnamefont
  {Tantivasadakarn}}, \bibinfo {author} {\bibfnamefont {R.}~\bibnamefont
  {Thorngren}}, \bibinfo {author} {\bibfnamefont {A.}~\bibnamefont
  {Vishwanath}},\ and\ \bibinfo {author} {\bibfnamefont {R.}~\bibnamefont
  {Verresen}},\ }\bibfield  {title} {\bibinfo {title} {Long-range entanglement
  from measuring symmetry-protected topological phases},\ }\href
  {https://doi.org/10.1103/PhysRevX.14.021040} {\bibfield  {journal} {\bibinfo
  {journal} {Physical Review X}\ }\textbf {\bibinfo {volume} {14}},\ \bibinfo
  {pages} {021040} (\bibinfo {year} {2024})}\BibitemShut {NoStop}%
\bibitem [{\citenamefont {Sahay}\ and\ \citenamefont
  {Verresen}(2025)}]{sahay2025classifying}%
  \BibitemOpen
  \bibfield  {author} {\bibinfo {author} {\bibfnamefont {R.}~\bibnamefont
  {Sahay}}\ and\ \bibinfo {author} {\bibfnamefont {R.}~\bibnamefont
  {Verresen}},\ }\bibfield  {title} {\bibinfo {title} {Classifying
  one-dimensional quantum states prepared by a single round of measurements},\
  }\href {https://doi.org/10.1103/PRXQuantum.6.010329} {\bibfield  {journal}
  {\bibinfo  {journal} {PRX Quantum}\ }\textbf {\bibinfo {volume} {6}},\
  \bibinfo {pages} {010329} (\bibinfo {year} {2025})}\BibitemShut {NoStop}%
\bibitem [{\citenamefont {Piroli}\ \emph {et~al.}(2024)\citenamefont {Piroli},
  \citenamefont {Styliaris},\ and\ \citenamefont {Cirac}}]{Piroli_2024}%
  \BibitemOpen
  \bibfield  {author} {\bibinfo {author} {\bibfnamefont {L.}~\bibnamefont
  {Piroli}}, \bibinfo {author} {\bibfnamefont {G.}~\bibnamefont {Styliaris}},\
  and\ \bibinfo {author} {\bibfnamefont {J.~I.}\ \bibnamefont {Cirac}},\
  }\bibfield  {title} {\bibinfo {title} {Approximating many-body quantum states
  with quantum circuits and measurements},\ }\bibfield  {journal} {\bibinfo
  {journal} {Physical Review Letters}\ }\textbf {\bibinfo {volume} {133}},\
  \href {https://doi.org/10.1103/PhysRevLett.133.230401}
  {10.1103/PhysRevLett.133.230401} (\bibinfo {year} {2024})\BibitemShut
  {NoStop}%
\bibitem [{\citenamefont {Gily\'{e}n}\ \emph {et~al.}(2019)\citenamefont
  {Gily\'{e}n}, \citenamefont {Su}, \citenamefont {Low},\ and\ \citenamefont
  {Wiebe}}]{gilyen2019quantum}%
  \BibitemOpen
  \bibfield  {author} {\bibinfo {author} {\bibfnamefont {A.}~\bibnamefont
  {Gily\'{e}n}}, \bibinfo {author} {\bibfnamefont {Y.}~\bibnamefont {Su}},
  \bibinfo {author} {\bibfnamefont {G.~H.}\ \bibnamefont {Low}},\ and\ \bibinfo
  {author} {\bibfnamefont {N.}~\bibnamefont {Wiebe}},\ }\bibfield  {title}
  {\bibinfo {title} {Quantum singular value transformation and beyond:
  {Exponential} improvements for quantum matrix arithmetics},\ }in\ \href
  {https://doi.org/10.1145/3313276.3316366} {\emph {\bibinfo {booktitle}
  {Proceedings of the 51st Annual {ACM} {SIGACT} Symposium on Theory of
  Computing}}}\ (\bibinfo {year} {2019})\ pp.\ \bibinfo {pages}
  {193--204}\BibitemShut {NoStop}%
\bibitem [{\citenamefont {Potter}\ \emph {et~al.}(2016)\citenamefont {Potter},
  \citenamefont {Morimoto},\ and\ \citenamefont {Vishwanath}}]{Potter2016}%
  \BibitemOpen
  \bibfield  {author} {\bibinfo {author} {\bibfnamefont {A.~C.}\ \bibnamefont
  {Potter}}, \bibinfo {author} {\bibfnamefont {T.}~\bibnamefont {Morimoto}},\
  and\ \bibinfo {author} {\bibfnamefont {A.}~\bibnamefont {Vishwanath}},\
  }\bibfield  {title} {\bibinfo {title} {Classification of interacting
  topological {Floquet} phases in one dimension},\ }\href
  {https://doi.org/10.1103/PhysRevX.6.041001} {\bibfield  {journal} {\bibinfo
  {journal} {Physical Review X}\ }\textbf {\bibinfo {volume} {6}},\ \bibinfo
  {pages} {041001} (\bibinfo {year} {2016})}\BibitemShut {NoStop}%
\bibitem [{\citenamefont {Po}\ \emph {et~al.}(2016)\citenamefont {Po},
  \citenamefont {Fidkowski}, \citenamefont {Morimoto}, \citenamefont {Potter},\
  and\ \citenamefont {Vishwanath}}]{Po_2016}%
  \BibitemOpen
  \bibfield  {author} {\bibinfo {author} {\bibfnamefont {H.~C.}\ \bibnamefont
  {Po}}, \bibinfo {author} {\bibfnamefont {L.}~\bibnamefont {Fidkowski}},
  \bibinfo {author} {\bibfnamefont {T.}~\bibnamefont {Morimoto}}, \bibinfo
  {author} {\bibfnamefont {A.~C.}\ \bibnamefont {Potter}},\ and\ \bibinfo
  {author} {\bibfnamefont {A.}~\bibnamefont {Vishwanath}},\ }\bibfield  {title}
  {\bibinfo {title} {Chiral {Floquet} phases of many-body localized bosons},\
  }\href {https://doi.org/10.1103/PhysRevX.6.041070} {\bibfield  {journal}
  {\bibinfo  {journal} {Physical Review X}\ }\textbf {\bibinfo {volume} {6}},\
  \bibinfo {pages} {041070} (\bibinfo {year} {2016})}\BibitemShut {NoStop}%
\bibitem [{\citenamefont {Stinespring}(1955)}]{stinespring_positive_1955}%
  \BibitemOpen
  \bibfield  {author} {\bibinfo {author} {\bibfnamefont {W.~F.}\ \bibnamefont
  {Stinespring}},\ }\bibfield  {title} {\bibinfo {title} {Positive functions on
  {C*}-algebras},\ }\href {https://doi.org/10.2307/2032342} {\bibfield
  {journal} {\bibinfo  {journal} {Proceedings of the American Mathematical
  Society}\ }\textbf {\bibinfo {volume} {6}},\ \bibinfo {pages} {211} (\bibinfo
  {year} {1955})}\BibitemShut {NoStop}%
\bibitem [{\citenamefont {Choi}(1975)}]{choi_completely_1975}%
  \BibitemOpen
  \bibfield  {author} {\bibinfo {author} {\bibfnamefont {M.-D.}\ \bibnamefont
  {Choi}},\ }\bibfield  {title} {\bibinfo {title} {Completely positive linear
  maps on complex matrices},\ }\href
  {https://doi.org/10.1016/0024-3795(75)90075-0} {\bibfield  {journal}
  {\bibinfo  {journal} {Linear Algebra and its Applications}\ }\textbf
  {\bibinfo {volume} {10}},\ \bibinfo {pages} {285} (\bibinfo {year}
  {1975})}\BibitemShut {NoStop}%
\bibitem [{\citenamefont
  {Jami\o{}{\l}kowski}(1972)}]{jamiolkowski_linear_1972}%
  \BibitemOpen
  \bibfield  {author} {\bibinfo {author} {\bibfnamefont {A.}~\bibnamefont
  {Jami\o{}{\l}kowski}},\ }\bibfield  {title} {\bibinfo {title} {Linear
  transformations which preserve trace and positive semidefiniteness of
  operators},\ }\href {https://doi.org/10.1016/0034-4877(72)90011-0} {\bibfield
   {journal} {\bibinfo  {journal} {Reports on Mathematical Physics}\ }\textbf
  {\bibinfo {volume} {3}},\ \bibinfo {pages} {275} (\bibinfo {year}
  {1972})}\BibitemShut {NoStop}%
\bibitem [{\citenamefont
  {Schollw\"{o}ck}(2011)}]{schollwoeck_density-matrix_2011}%
  \BibitemOpen
  \bibfield  {author} {\bibinfo {author} {\bibfnamefont {U.}~\bibnamefont
  {Schollw\"{o}ck}},\ }\bibfield  {title} {\bibinfo {title} {The density-matrix
  renormalization group in the age of matrix product states},\ }\href
  {https://doi.org/10.1016/j.aop.2010.09.012} {\bibfield  {journal} {\bibinfo
  {journal} {Annals of Physics}\ }\textbf {\bibinfo {volume} {326}},\ \bibinfo
  {pages} {96} (\bibinfo {year} {2011})}\BibitemShut {NoStop}%
\bibitem [{\citenamefont {Cuevas}\ \emph {et~al.}(2013)\citenamefont {Cuevas},
  \citenamefont {Schuch}, \citenamefont {P\'{e}rez-Garc\'{i}a},\ and\
  \citenamefont {Cirac}}]{cuevas_purifications_2013}%
  \BibitemOpen
  \bibfield  {author} {\bibinfo {author} {\bibfnamefont {G.~D.~l.}\
  \bibnamefont {Cuevas}}, \bibinfo {author} {\bibfnamefont {N.}~\bibnamefont
  {Schuch}}, \bibinfo {author} {\bibfnamefont {D.}~\bibnamefont
  {P\'{e}rez-Garc\'{i}a}},\ and\ \bibinfo {author} {\bibfnamefont {J.~I.}\
  \bibnamefont {Cirac}},\ }\bibfield  {title} {\bibinfo {title} {Purifications
  of multipartite states: limitations and constructive methods},\ }\href
  {https://doi.org/10.1088/1367-2630/15/12/123021} {\bibfield  {journal}
  {\bibinfo  {journal} {New Journal of Physics}\ }\textbf {\bibinfo {volume}
  {15}},\ \bibinfo {pages} {123021} (\bibinfo {year} {2013})}\BibitemShut
  {NoStop}%
\bibitem [{\citenamefont {Farrelly}(2020)}]{farrelly_review_2020}%
  \BibitemOpen
  \bibfield  {author} {\bibinfo {author} {\bibfnamefont {T.}~\bibnamefont
  {Farrelly}},\ }\bibfield  {title} {\bibinfo {title} {A review of quantum
  cellular automata},\ }\href {https://doi.org/10.22331/q-2020-11-30-368}
  {\bibfield  {journal} {\bibinfo  {journal} {Quantum}\ }\textbf {\bibinfo
  {volume} {4}},\ \bibinfo {pages} {368} (\bibinfo {year} {2020})}\BibitemShut
  {NoStop}%
\bibitem [{\citenamefont {Perez-Garcia}\ \emph {et~al.}(2007)\citenamefont
  {Perez-Garcia}, \citenamefont {Verstraete}, \citenamefont {Wolf},\ and\
  \citenamefont {Cirac}}]{perez-garcia_matrix_2007}%
  \BibitemOpen
  \bibfield  {author} {\bibinfo {author} {\bibfnamefont {D.}~\bibnamefont
  {Perez-Garcia}}, \bibinfo {author} {\bibfnamefont {F.}~\bibnamefont
  {Verstraete}}, \bibinfo {author} {\bibfnamefont {M.~M.}\ \bibnamefont
  {Wolf}},\ and\ \bibinfo {author} {\bibfnamefont {J.~I.}\ \bibnamefont
  {Cirac}},\ }\bibfield  {title} {\bibinfo {title} {Matrix product state
  representations},\ }\href@noop {} {\bibfield  {journal} {\bibinfo  {journal}
  {Quantum Information and Computation}\ }\textbf {\bibinfo {volume} {7}},\
  \bibinfo {pages} {401} (\bibinfo {year} {2007})},\ \Eprint
  {https://arxiv.org/abs/quant-ph/0608197} {arXiv:quant-ph/0608197}
  \BibitemShut {NoStop}%
\bibitem [{\citenamefont {Wolf}\ \emph {et~al.}(2008)\citenamefont {Wolf},
  \citenamefont {Verstraete}, \citenamefont {Hastings},\ and\ \citenamefont
  {Cirac}}]{Wolf_2008}%
  \BibitemOpen
  \bibfield  {author} {\bibinfo {author} {\bibfnamefont {M.~M.}\ \bibnamefont
  {Wolf}}, \bibinfo {author} {\bibfnamefont {F.}~\bibnamefont {Verstraete}},
  \bibinfo {author} {\bibfnamefont {M.~B.}\ \bibnamefont {Hastings}},\ and\
  \bibinfo {author} {\bibfnamefont {J.~I.}\ \bibnamefont {Cirac}},\ }\bibfield
  {title} {\bibinfo {title} {Area laws in quantum systems: Mutual information
  and correlations},\ }\bibfield  {journal} {\bibinfo  {journal} {Physical
  Review Letters}\ }\textbf {\bibinfo {volume} {100}},\ \href
  {https://doi.org/10.1103/physrevlett.100.070502}
  {10.1103/physrevlett.100.070502} (\bibinfo {year} {2008})\BibitemShut
  {NoStop}%
\bibitem [{\citenamefont {Groisman}\ \emph {et~al.}(2005)\citenamefont
  {Groisman}, \citenamefont {Popescu},\ and\ \citenamefont
  {Winter}}]{Groisman_2005}%
  \BibitemOpen
  \bibfield  {author} {\bibinfo {author} {\bibfnamefont {B.}~\bibnamefont
  {Groisman}}, \bibinfo {author} {\bibfnamefont {S.}~\bibnamefont {Popescu}},\
  and\ \bibinfo {author} {\bibfnamefont {A.}~\bibnamefont {Winter}},\
  }\bibfield  {title} {\bibinfo {title} {Quantum, classical, and total amount
  of correlations in a quantum state},\ }\bibfield  {journal} {\bibinfo
  {journal} {Physical Review A}\ }\textbf {\bibinfo {volume} {72}},\ \href
  {https://doi.org/10.1103/physreva.72.032317} {10.1103/physreva.72.032317}
  (\bibinfo {year} {2005})\BibitemShut {NoStop}%
\bibitem [{\citenamefont {Terhal}\ \emph {et~al.}(2002)\citenamefont {Terhal},
  \citenamefont {Horodecki}, \citenamefont {Leung},\ and\ \citenamefont
  {DiVincenzo}}]{Terhal_2002}%
  \BibitemOpen
  \bibfield  {author} {\bibinfo {author} {\bibfnamefont {B.~M.}\ \bibnamefont
  {Terhal}}, \bibinfo {author} {\bibfnamefont {M.}~\bibnamefont {Horodecki}},
  \bibinfo {author} {\bibfnamefont {D.~W.}\ \bibnamefont {Leung}},\ and\
  \bibinfo {author} {\bibfnamefont {D.~P.}\ \bibnamefont {DiVincenzo}},\
  }\bibfield  {title} {\bibinfo {title} {The entanglement of purification},\
  }\href {https://doi.org/10.1063/1.1498001} {\bibfield  {journal} {\bibinfo
  {journal} {Journal of Mathematical Physics}\ }\textbf {\bibinfo {volume}
  {43}},\ \bibinfo {pages} {4286–4298} (\bibinfo {year} {2002})}\BibitemShut
  {NoStop}%
\bibitem [{\citenamefont {Guth~Jarkovsk{\`y}}\ \emph
  {et~al.}(2020)\citenamefont {Guth~Jarkovsk{\`y}}, \citenamefont {Moln{\'a}r},
  \citenamefont {Schuch},\ and\ \citenamefont {Cirac}}]{Guth_Jarkovsky_2020}%
  \BibitemOpen
  \bibfield  {author} {\bibinfo {author} {\bibfnamefont {J.}~\bibnamefont
  {Guth~Jarkovsk{\`y}}}, \bibinfo {author} {\bibfnamefont {A.}~\bibnamefont
  {Moln{\'a}r}}, \bibinfo {author} {\bibfnamefont {N.}~\bibnamefont {Schuch}},\
  and\ \bibinfo {author} {\bibfnamefont {J.~I.}\ \bibnamefont {Cirac}},\
  }\bibfield  {title} {\bibinfo {title} {Efficient description of many-body
  systems with matrix product density operators},\ }\href
  {https://doi.org/10.1103/prxquantum.1.010304} {\bibfield  {journal} {\bibinfo
   {journal} {PRX Quantum}\ }\textbf {\bibinfo {volume} {1}},\ \bibinfo {pages}
  {010304} (\bibinfo {year} {2020})}\BibitemShut {NoStop}%
\bibitem [{\citenamefont {Affleck}\ \emph {et~al.}(1987)\citenamefont
  {Affleck}, \citenamefont {Kennedy}, \citenamefont {Lieb},\ and\ \citenamefont
  {Tasaki}}]{affleck1987rigorous}%
  \BibitemOpen
  \bibfield  {author} {\bibinfo {author} {\bibfnamefont {I.}~\bibnamefont
  {Affleck}}, \bibinfo {author} {\bibfnamefont {T.}~\bibnamefont {Kennedy}},
  \bibinfo {author} {\bibfnamefont {E.~H.}\ \bibnamefont {Lieb}},\ and\
  \bibinfo {author} {\bibfnamefont {H.}~\bibnamefont {Tasaki}},\ }\bibfield
  {title} {\bibinfo {title} {Rigorous results on valence-bond ground states in
  antiferromagnets},\ }\href {https://doi.org/10.1103/PhysRevLett.59.799}
  {\bibfield  {journal} {\bibinfo  {journal} {Physical Review Letters}\
  }\textbf {\bibinfo {volume} {59}},\ \bibinfo {pages} {799} (\bibinfo {year}
  {1987})}\BibitemShut {NoStop}%
\bibitem [{\citenamefont {Cirac}\ \emph
  {et~al.}(2017{\natexlab{c}})\citenamefont {Cirac}, \citenamefont
  {Perez-Garcia}, \citenamefont {Schuch},\ and\ \citenamefont
  {Verstraete}}]{cirac:mpdo-rgfp}%
  \BibitemOpen
  \bibfield  {author} {\bibinfo {author} {\bibfnamefont {J.~I.}\ \bibnamefont
  {Cirac}}, \bibinfo {author} {\bibfnamefont {D.}~\bibnamefont {Perez-Garcia}},
  \bibinfo {author} {\bibfnamefont {N.}~\bibnamefont {Schuch}},\ and\ \bibinfo
  {author} {\bibfnamefont {F.}~\bibnamefont {Verstraete}},\ }\bibfield  {title}
  {\bibinfo {title} {Matrix product density operators: {Renormalization} fixed
  points and boundary theories},\ }\href
  {https://doi.org/10.1016/j.aop.2016.12.030} {\bibfield  {journal} {\bibinfo
  {journal} {Annals of Physics}\ }\textbf {\bibinfo {volume} {378}},\ \bibinfo
  {pages} {100} (\bibinfo {year} {2017}{\natexlab{c}})},\ \Eprint
  {https://arxiv.org/abs/1606.00608} {arXiv:1606.00608} \BibitemShut {NoStop}%
\bibitem [{\citenamefont {Nagata}(1952)}]{Nagata}%
  \BibitemOpen
  \bibfield  {author} {\bibinfo {author} {\bibfnamefont {M.}~\bibnamefont
  {Nagata}},\ }\bibfield  {title} {\bibinfo {title} {On the theory of
  {Henselian} rings},\ }\href@noop {} {\bibfield  {journal} {\bibinfo
  {journal} {Journal of the Mathematical Society of Japan}\ }\textbf {\bibinfo
  {volume} {4}},\ \bibinfo {pages} {296} (\bibinfo {year} {1952})}\BibitemShut
  {NoStop}%
\bibitem [{\citenamefont {Higman}(1956)}]{Higman}%
  \BibitemOpen
  \bibfield  {author} {\bibinfo {author} {\bibfnamefont {G.}~\bibnamefont
  {Higman}},\ }\bibfield  {title} {\bibinfo {title} {On infinite simple
  permutation groups},\ }\href@noop {} {\bibfield  {journal} {\bibinfo
  {journal} {Proceedings of the Cambridge Philosophical Society}\ }\textbf
  {\bibinfo {volume} {52}},\ \bibinfo {pages} {1} (\bibinfo {year}
  {1956})}\BibitemShut {NoStop}%
\bibitem [{\citenamefont {Razmyslov}(1974)}]{Razmyslov}%
  \BibitemOpen
  \bibfield  {author} {\bibinfo {author} {\bibfnamefont {Y.}~\bibnamefont
  {Razmyslov}},\ }\bibfield  {title} {\bibinfo {title} {Trace identities of
  full matrix algebras over a field of characteristic zero},\ }\href@noop {}
  {\bibfield  {journal} {\bibinfo  {journal} {Izvestiya Akademii Nauk SSSR.
  Seriya Matematicheskaya}\ }\textbf {\bibinfo {volume} {38}},\ \bibinfo
  {pages} {723} (\bibinfo {year} {1974})},\ \bibinfo {note} {english
  translation: Mathematics of the USSR-Izvestiya 8 (1974), 727--760
  (1975)}\BibitemShut {NoStop}%
\bibitem [{\citenamefont {Berry}\ \emph {et~al.}(2014)\citenamefont {Berry},
  \citenamefont {Childs}, \citenamefont {Cleve}, \citenamefont {Kothari},\ and\
  \citenamefont {Somma}}]{Berry_2014}%
  \BibitemOpen
  \bibfield  {author} {\bibinfo {author} {\bibfnamefont {D.~W.}\ \bibnamefont
  {Berry}}, \bibinfo {author} {\bibfnamefont {A.~M.}\ \bibnamefont {Childs}},
  \bibinfo {author} {\bibfnamefont {R.}~\bibnamefont {Cleve}}, \bibinfo
  {author} {\bibfnamefont {R.}~\bibnamefont {Kothari}},\ and\ \bibinfo {author}
  {\bibfnamefont {R.~D.}\ \bibnamefont {Somma}},\ }\bibfield  {title} {\bibinfo
  {title} {Exponential improvement in precision for simulating sparse
  {Hamiltonians}},\ }in\ \href {https://doi.org/10.1145/2591796.2591854} {\emph
  {\bibinfo {booktitle} {Proceedings of the Forty-Sixth Annual {ACM} Symposium
  on Theory of Computing}}},\ \bibinfo {series and number} {STOC '14}\
  (\bibinfo  {publisher} {ACM},\ \bibinfo {year} {2014})\ pp.\ \bibinfo {pages}
  {283--292}\BibitemShut {NoStop}%
\bibitem [{\citenamefont {Piroli}\ and\ \citenamefont
  {Cirac}(2020)}]{piroli_quantum_2020}%
  \BibitemOpen
  \bibfield  {author} {\bibinfo {author} {\bibfnamefont {L.}~\bibnamefont
  {Piroli}}\ and\ \bibinfo {author} {\bibfnamefont {J.~I.}\ \bibnamefont
  {Cirac}},\ }\bibfield  {title} {\bibinfo {title} {Quantum cellular automata,
  tensor networks, and area laws},\ }\href
  {https://doi.org/10.1103/PhysRevLett.125.190402} {\bibfield  {journal}
  {\bibinfo  {journal} {Physical Review Letters}\ }\textbf {\bibinfo {volume}
  {125}},\ \bibinfo {pages} {190402} (\bibinfo {year} {2020})}\BibitemShut
  {NoStop}%
\end{thebibliography}%

\clearpage
\onecolumngrid
\appendix
\setcounter{equation}{0}
\setcounter{figure}{0}
\setcounter{table}{0}
\setcounter{thm}{0}
\setcounter{defn}{0}
\setcounter{prop}{0}

\makeatletter
\renewcommand{\thefigure}{S\arabic{figure}}

\begin{center}
    {\LARGE \textbf{Supplemental Material}}
\end{center}

\section{Quantum Channels Basics}
\label{sec:QuantumChannelsBasics}
A quantum channel is a completely positive, trace-preserving (CPTP) linear map
\[
\mathcal{E} : \mathcal{B}(\mathcal{H}_{\mathrm{in}}) \to \mathcal{B}(\mathcal{H}_{\mathrm{out}}),
\]
where $\mathcal{B}(\mathcal{H})$ denotes the algebra of bounded operators on the Hilbert space $\mathcal{H}$. Complete positivity requires that $\mathcal{E} \otimes \mathbb{I}_n$ maps positive operators to positive operators for any $n \in \mathbb{N}$, ensuring physical consistency in the presence of entanglement with an ancillary system. By the Stinespring dilation theorem~\cite{stinespring_positive_1955}, any such map admits an isometric representation on an extended Hilbert space: there exist an environment $\mathcal{H}_E$ and an isometry
\[
V : \mathcal{H}_{\mathrm{in}} \longrightarrow \mathcal{H}_{\mathrm{out}} \otimes \mathcal{H}_E,
\]
such that
\[
\mathcal{E}(\rho) = \operatorname{Tr}_E\!\left[ V \rho V^\dagger \right],
\]
where $\operatorname{Tr}_E$ denotes the partial trace over $\mathcal{H}_E$.

An alternative characterization is provided by the Choi–Jamiolkowski isomorphism~\cite{choi_completely_1975, jamiolkowski_linear_1972}. Let ${|\Phi^+\rangle = \frac{1}{\sqrt{d}} \sum_{i=1}^d |i\rangle \otimes |i\rangle}$ be a maximally entangled state on $\mathcal{H}_{\mathrm{in}} \otimes \mathcal{H}_{\mathrm{in}}$ with $d = \dim \mathcal{H}_{\mathrm{in}}$. The Choi state associated with $\mathcal{E}$ is
$$
C_{\mathcal{E}} := (\mathcal{E} \otimes \mathbb{I}_{\mathrm{in}})\big( |\Phi^+\rangle\!\langle\Phi^+| \big) .
$$
The map $\mathcal{E}$ is completely positive if and only if $C_{\mathcal{E}} \ge 0$, and it is trace-preserving if and only if $\operatorname{Tr}_{\mathrm{out}} C_{\mathcal{E}} = \mathbb{I}_{\mathrm{in}} / d$.

\section{Tensor-Network Formulation of Quantum Channels}
\label{sec:TensorNetworksFormulation}

\subsection{Matrix Product Quantum Channels (MPQC)}

Graphical notation is frequently employed in our work, thus we begin by setting forth some conventions. Our focus will often be on rank-6 single-site tensors:

\def\leglength{0.5}
\def\barthickness{0.3}
\def\horizontallength{1.8}
\def\verticalheight{2}
\def\squarehalfside{0.3}

\begin{align}
    \left(A_k^{ij,op}\right)_{mn} =
    \begin{array}{c}
    \begin{tikzpicture}[scale=\defaultscaling,baseline={([yshift=-0.65ex] current bounding box.center)}]
        \HorseshoeTensor{0,0}{\leglength}{\barthickness}{$A_k$}{0}{\horizontallength}{\verticalheight}{1}
    \end{tikzpicture}
    \end{array}
    \in \mathbb C ,
\end{align}
where $k$ denotes the site index, $i,j = 1,\dots,d_{\mathrm{in}}$ label the input physical space, $o,p = 1,\dots,d_{\mathrm{out}}$ label the output physical space, and $m = 1,\dots,D_{k-1}$, $n = 1,\dots,D_{k}$ label the auxiliary space. The integers $D_{k-1}$ and $D_k$ represent the left and right bond dimensions of $A_k$, respectively. Hence, $A_k^{ij,op} \in \mathbb{M}_{D_k,D_{k-1}}$ can be understood as a matrix.

Given suitably chosen tensors $A_1,\dots,A_N$ with compatible bond dimension, the corresponding \textit{Matrix Product Quantum Channel} (MPQC) is the map $$\mathrm{MPQC}_N^{\Vec{A}} :\; \mathcal{L}(\mathcal{H}_{\mathrm{in}})
\;\longrightarrow\;
\mathcal{L}(\mathcal{H}_{\mathrm{out}}),$$
where
$$
\mathcal{H}_{\mathrm{in}} \;=\; 
 (\mathbb{C}^{d_{\mathrm{in}}})^{\otimes N},
\qquad
\mathcal{H}_{\mathrm{out}} \;=\;
 (\mathbb{C}^{d_{\mathrm{out}}})^{\otimes N},
 $$
given by
\begin{equation}
    \mathrm{MPQC}_N^{\Vec{A}} = \sum_{ijop}  C_{ij}^{op} \bra{i_1 \!\dots\! i_N} (\cdot) \ket{j_1 \!\dots\! j_N} \ket{o_1 \!\dots\! o_N} \bra{p_1 \!\dots\! p_N}, \label{ntichannels}
\end{equation}
where $ C_{ij}^{op} = A_1^{i_1 j_1, o_1 p_1} \cdot A_2^{i_2 j_2, o_2 p_2} \cdot \dotso \cdot A_N^{i_N j_N, o_N p_N}$. Here it is assumed that $A_1$ and $A_N$ are row and column vectors, that is, they have respectively trivial left and right bond dimension $D_0 = D_N = 1$. We also assumed, wlog, that $d_{\mathrm{in}}$ and $d_{\mathrm{out}}$ are the same for all sites.\\
%The input and output Hilbert spaces can also be generalized to allow for different local dimensions ${d_{\mathrm{in},i}}$.

Our choice of graphical notation for MPQC can be understood by considering their action on \textit{Matrix Product Density Operators} (MPDO)~\cite{cirac_matrix_2021}. Given a sequence $\vec{G}$ of rank-4 tensors $G_1,\dots,G_N$ with compatible bond dimensions, we recall that MPDO are the mixed states of the form:
\begin{equation}
\label{MPDO}
    \rho^{\vec{G}}_N=\sum_{\substack{{i_{1} \ldots i_{N}} \\ {j_{1} \ldots j_{N}}}} G^{i_{1},j_{1}}_1 \cdots G^{i_{N},j_{N}}_N\left|i_{1} \dots i_{N}\right\rangle\left\langle j_{1} \dots j_{N}\right|.
\end{equation}
Graphically, the action of MPQC on MPDO at the level of individual tensors is:
\begin{align*}
    \begin{array}{c}
    \begin{tikzpicture}[scale=\defaultscaling,baseline={([yshift=-0.65ex] current bounding box.center)}]
        \GTensorOrange{0,0}{\leglength*1.5}{0.35}{$G_k$}{0}
    \end{tikzpicture}
\end{array}
    \mapsto
        \begin{array}{c}
    \begin{tikzpicture}[scale=\defaultscaling,baseline={([yshift=-0.65ex] current bounding box.center)}]
        \GTensorOrange{0,0}{\leglength*1.5}{0.35}{$G_k'$}{0}
    \end{tikzpicture}
\end{array} 
=
    \begin{array}{c}
    \begin{tikzpicture}[scale=\defaultscaling,baseline={([yshift=-0.65ex] current bounding box.center)}]
        \HorseshoeTensorForChannelAction{0,0}{\leglength}{\barthickness}{$A_k$}{0}{\horizontallength}{\verticalheight}{0}
        %\draw [red, very thick] (\horizontallength - 0.5,-\verticalheight/2 + \barthickness + \leglength) to (\horizontallength - 0.5,\verticalheight/2 - \barthickness - \leglength);
    \end{tikzpicture}
    \,.
\end{array}
\end{align*}
It is easy to see that the action of the MPQC may increase the bond dimension of the MPDO, but at most by a multiplicative factor given by the bond dimension of the MPQC (see \cref{arealaw} for further details).

In general, any quantum channel can be represented as an MPQC, as we now explain. Consider the vectorized version of the channel, such that we get an unnormalized pure quantum state on a one-dimensional lattice of $N$ sites, each with a $d$-dimensional local Hilbert space, where $d = d_{\mathrm{in}}^2 d_{\mathrm{out}}^2$. In the computational basis, such a state can be written as
\[
|\psi\rangle = \sum_{\sigma_1, \dots, \sigma_N} c_{\sigma_1 \dots \sigma_N} \, |\sigma_1, \dots, \sigma_N\rangle ,
\]
where the coefficients $c_{\sigma_1 \dots \sigma_N}$ form a tensor with $d^N$ components. One can rewrite this tensor by performing successive singular value decompositions (SVDs) \cite{vidal2003efficient,Sch_n_2005}. To begin, one reshapes the coefficients into a matrix $c_{\sigma_1, (\sigma_2 \dots \sigma_N)}$ by grouping the first index against the remaining ones, and performs an SVD. The left isometry is then split into a set of row vectors $A^{\sigma_1}$ labeled by the first site, while the remaining part is reshaped into a new matrix with composite indices $(a_1, \sigma_2)$ and $(\sigma_3 \dots \sigma_N)$, where $a_1 = 1,\dots, {\mathrm {rank}} (c_{\sigma_1, (\sigma_2 \dots \sigma_N)})$. Repeating this procedure iteratively across the chain yields, at each step, matrices $A^{\sigma_i}$ associated with the local indices, and bond dimensions determined by the number of non-zero singular values at the corresponding cut. After $N-1$ steps, the original tensor of coefficients is exactly expressed as a product of matrices,
\[
c_{\sigma_1 \dots \sigma_N} = A^{\sigma_1} A^{\sigma_2} \cdots A^{\sigma_N},
\]
so that the quantum state takes the compact form
\[
|\psi\rangle = \sum_{\sigma_1, \dots, \sigma_N} A^{\sigma_1} A^{\sigma_2} \cdots A^{\sigma_N} \, |\sigma_1, \dots, \sigma_N\rangle .
\]

The above shows that any state on a finite lattice can be written exactly in matrix-product form. Nevertheless, for a generic state, the ranks of the intermediate decompositions, and hence the bond dimensions of the matrices, can grow exponentially with system size (up to order $d^{N/2}$). While the representation is always formally possible, such exponential growth makes the exact construction inefficient for practical purposes \cite{schollwoeck_density-matrix_2011}. The argument extends to quantum channels. From now on, we will focus on MPQC that can be represented with a bounded bond dimension, independent of the system size $N$:
\begin{defn}[{\sffamily{MPQC}} class] \label{def:mpqc_class}
    Consider an infinite sequence of quantum channels, of increasing system size $N$, expressed in the matrix-product form $\{ \mathrm{MPQC}_N^{\Vec{A}} \}_{N}$. The sequence is in the {\sffamily{MPQC}} class if all bond dimensions $D_k$ are uniformly bounded by a constant independent of $k$ and $N$.
\end{defn}

For the map $\mathrm{MPQC}_N^{\vec{A}}$ to be a proper quantum channel for all system sizes, it needs to satisfy the complete positivity (CP) and trace preservation (TP) conditions for all system sizes $N \ge 1$. These conditions can be rephrased in terms of the Choi–Jamiołkowski state 
${C_N^{\vec{A}} =(\mathrm{MPQC}_N^{\vec{A}} \otimes \mathbb{I}) \ket{\Phi^+}\bra{\Phi^+}}$:
\begin{equation}
    \label{ChoiCondition}
    (i) \; C_N^{\vec{A}} \succeq 0,  \quad
    (ii) \;\Tr_\mathrm{out} (C_N^{\vec{A}}) = \mathbb{I}_\mathrm{in}  \qquad \forall N.
\end{equation}
However, verifying the first condition, i.e.,  whether a matrix product operator is positive semidefinite for all system sizes, is undecidable in the general case~\cite{cuevas_fundamental_2016}. Furthermore, we want to focus on physically relevant channels, i.e., noise models arising from local interactions in both space and time. Consequently, we restrict our attention to a subclass of MPQC, which we call \textit{locally purified} class.

\subsection{Local Purification (LP) of Quantum Channels}

Now we introduce the \textit{locally purified} class of MPQC, for which CP is automatically satisfied, TP can be efficiently checked, and which is well-suited to describe channels with a locality structure in the environment, a property expected to hold on physical grounds. We say that an MPQC tensor $A_k$ admits a \textit{local purification} (LP) if there exists a tensor ${\cal V}_k$ such that:
\begin{align} \label{LP def}
    \begin{array}{c}
        \begin{tikzpicture}[scale=\defaultscaling,baseline={([yshift=-0.65ex] current bounding box.center)}]
            \HorseshoeTensor{0,0}{\leglength}{\barthickness}{$A_k$}{0}{\horizontallength}{\verticalheight}{0}
        \end{tikzpicture}
    \end{array}
        =
    \begin{array}{c}
        \begin{tikzpicture}[scale=\defaultscaling,baseline={([yshift=-0.65ex] current bounding box.center)}]
                \pgfmathsetmacro{\calculatedDistance}{\verticalheight/2-\barthickness/2}
                \pgfmathsetmacro{\squarehalf}{0.35}
                \LPTensor{0,0}{\leglength*0.75}{\squarehalf}{${\cal V}_k$}{0}{\calculatedDistance}
        \end{tikzpicture}
    \end{array}
    \,.
\end{align}
Notice that we are assuming that the virtual legs of the MPQC tensor factorize. If the LP form is admissible for all tensors of an MPQC, while allowing for a constant bond dimension, we say the channel belongs to the {\sffamily{LP}} class:
\begin{defn}[{\sffamily{LP}} class]
    The {\sffamily{LP}} (Locally Purified) class consists of sequences of channels in {\sffamily{MPQC}} whose tensors can be expressed in the local purification form [\cref{LP def}], with the bond dimensions of all $\mathcal V_k$ uniformly bounded by a constant independent of $k$ and $N$.
\end{defn}

Importantly, in the {\sffamily{LP}} class, CP is guaranteed by construction because the Choi state satisfies $C_N^A=WW^\dagger \succeq 0$ for 
\begin{align} 
W =     \begin{tikzpicture}[scale=0.6,baseline={([yshift=-0.65ex] current bounding box.center)}] 
        \pgfmathsetmacro{\squarehalf}{0.4}
        \IsomTensorPermuted{0,0}{\leglength}{\squarehalf}{${\cal V}_1$}{-1} 
        \IsomTensorPermuted{2*\leglength + 2*\squarehalf,0}{\leglength}{\squarehalf}{${\cal V}_2$}{0} \SingleDots{4*\leglength + 4*\squarehalf, 0}{2*\leglength} 
        \IsomTensorPermuted{6*\leglength + 6*\squarehalf,0}{\leglength}{\squarehalf}{${\cal V}_N$}{1} 
    \end{tikzpicture} .
\end{align}
Here $W$
%: (\mathbb C^{\mathrm {\chi}})^{\otimes N} \to  (\mathbb C^{d_{{\rm in} } d_{\mathrm{out}} })^{\otimes N}$
is understood as an operator with input the purification (gray) legs and output the physical (black) legs.
TP, on the other hand, is equivalent to the requirement that $V_N$, defined as
\begin{align}
V_N = 
    \begin{tikzpicture}[scale=\defaultscaling,baseline={([yshift=-0.65ex] current bounding box.center)}]
        \pgfmathsetmacro{\squarehalf}{0.4}
            \IsomTensor{0,0}{\leglength}{\squarehalf}{${\cal V}_1$}{-1}
            \IsomTensor{2*\leglength + 2*\squarehalf,0}{\leglength}{\squarehalf}{${\cal V}_2$}{0}
            \SingleDots{4*\leglength + 4*\squarehalf, 0}{2*\leglength}
            \IsomTensor{6*\leglength + 6*\squarehalf,0}{\leglength}{\squarehalf}{${\cal V}_N$}{1}
    \end{tikzpicture},
\end{align}
satisfies $$V_N^\dagger V_N^{} = \Tr_\mathrm{out} (C_N^{\vec{A}}) = \mathbb{I}_{d_\mathrm{in}}^{\otimes N} \,.$$ That is, the map 
${V_N: (\mathbb C^{d_{\mathrm {in}}})^{\otimes N} \to  (\mathbb C^{\chi d_{\mathrm{out}} })^{\otimes N}}$ 
is an isometry. We call $V_N$ a \emph{matrix product isometry} (MPI), the ${\cal V}_k$ --- which have a uniformly bounded bond dimension --- MPI tensors, and the corresponding class {\sffamily{MPI}}. We have therefore arrived at an equivalent characterization of the {\sffamily{LP}} class:
\begin{prop}
    The class {\sffamily{LP}} consists of sequences of quantum channels admitting a Stinespring dilation which is an MPI with uniformly bounded bond dimension.
\end{prop}

We refer to the space of dimension $\chi$ as the \emph{purification space}. The channel expressed in this form will be indicated as $$\mathrm{LP}_N^{\vec{\cal V}}(\cdot) = \Tr_{\mathrm{pur}}[V_N^{}(\cdot)V_N^\dagger]\,.$$ Notation-wise, the grey wires represent throughout the purification space, i.e., the legs that are traced out when forming the channel corresponding to the isometry.

\subsubsection{Inequivalence of the {\text{\sffamily{LP}}} and {\text{\sffamily{MPQC}}} Classes}
\label{sec:LocalDilationNTI}

We would now like to investigate whether the $\text{\sffamily{MPQC}} = \text{\sffamily{LP}}$, i.e.\ if all MPQC can be written in LP form. Note that the converse is automatically true. Every MPQC admits a Stinespring dilation. However, the difficulty lies in the fact that the bond dimension of the dilation might fail to be uniformly bounded by a constant, independent of the system size, even if this holds true for the original channel. We first work in the most general setting of general, non translationally invariant (NTI) MPQC, as introduced in \cref{def:mpqc_class}.

\begin{prop}[${\text{\sffamily{MPQC}}} \ne {\text{\sffamily{LP}}}$]
\label{ineqLocNTI}
There exist MPQC with uniformly bounded bond dimension that admit no Stinespring dilation with uniformly bounded bond dimension.
\end{prop}

To show the inequivalence of the two classes, we adapt a construction of Ref.~\cite{  cuevas_purifications_2013}. There, it is shown that there exist explicit NTI MPDOs $\rho_0$ of constant bond dimension that cannot be expressed with a uniformly bounded bond dimension (i.e., independent of the system size) in local purification form, obtained by purifying $\rho_0$ into a pure state and expressing the latter as an MPS.

A first trivial example consists of a channel without input, whose output is a state $\rho_0$ as in Ref.~\cite{cuevas_purifications_2013}. The fact that $\rho_0$ cannot be purified to a constant bond dimension MPS implies that the constructed channel with $d_{{\rm in}} = 1$ is not in \text{\sffamily{LP}}. However, this example might be unsatisfying, as channels are usually thought of as having a non-trivial input. To address this, we consider replacement channels, where any state is taken as input and discarded, and a specific state is given as output. A second example is then obtained by choosing the same state $\rho_0$ as before as output: $$ \mathcal{E}(\cdot) = \rho_0 \operatorname{Tr} (\cdot)\,. $$ Its Choi matrix is given by $\rho_0 \otimes \Id$, which admits a local purification in finite bond dimension if and only if $\rho_0$ does, implying again that the constructed channel is in the \text{\sffamily{MPQC}} class but not in the \text{\sffamily{LP}} one.

\subsubsection{Local Purification and Homogeneity}

We would now like to investigate whether the {\sffamily{MPQC}} and {\sffamily{LP}} classes are inequivalent, even under the stricter assumption of \textit{translational invariance} (TI). Although we cannot make a statement under the full TI classes, we establish an inequivalence between two related subclasses that arise by repeating a fixed tensor.

A natural way to build sequences of MPQC (with an increasing system size $N$) is to repeat a suitably chosen, but fixed, tensor $N$ times with periodic boundary conditions. The same can also be done at the MPI level, resulting in a TI LP channel. This is captured by the concept of homogeneity:

\begin{defn}[Homogeneity ({\sffamily homo})]
    {\sffamily{hMPQC}} is the subclass of {\sffamily{MPQC}} generated by a repeated tensor $A$, for all system sizes $N$, with a periodic boundary:
\def\gap{1.5} % adjust this value as needed
\begin{align*}
\mathrm{MPQC}_N^A = 
\begin{array}{c}
    \begin{tikzpicture}[scale=\defaultscaling,baseline={([yshift=-0.75ex] current bounding box.center)}]
        \SingleTrLeft{(-\gap,0)}
        \HorizontalLine{(-\gap,0)}{\gap}
        \HorseshoeTensor{(0,0)}{\leglength}{\barthickness}{$A$}{1}{\horizontallength}{\verticalheight}{0}
        \HorizontalLine{(\leglength,0)}{\gap}
        \HorseshoeTensor{(\gap + \horizontallength/2,0)}{\leglength}{\barthickness}{$A$}{0}{\horizontallength}{\verticalheight}{0}
        \HorizontalLine{(\gap + \horizontallength/2 + \leglength,0)}{\gap}
        \SingleDots{(\leglength + \gap + \horizontallength + \leglength + \gap/2,0)}{\singledx/2}
        \pgfmathsetmacro{\tmp}{\leglength + \gap + \horizontallength + \leglength + \gap + \singledx/2}
        \HorseshoeTensor{(\tmp,0)}{\leglength}{\barthickness}{$A$}{0}{\horizontallength}{\verticalheight}{0}
        \HorizontalLine{(\tmp + \leglength,0)}{1.5*\gap}
        %\fill[red] (\tmp + \leglength + \gap,0) circle (2pt);
        \SingleTrRight{(\tmp + \leglength + 1.5*\gap + 0.8,0)}
    \end{tikzpicture}
\end{array}
\,.
\end{align*}
Analogously, {\sffamily{hMPI}} is the subclass of {\sffamily{MPI}} generated by a repeated tensor $\mathcal V$
\begin{align}
V_N = 
    \begin{tikzpicture}[scale=\defaultscaling,baseline={([yshift=-0.65ex] current bounding box.center)}]
        \pgfmathsetmacro{\squarehalf}{0.4}
            \SingleTrLeft{(-4*\leglength,0)}
            \draw [very thick] (11*\leglength,0) to  (15*\leglength,0);
            \SingleTrRight{(15*\leglength,0)}
            \draw [very thick] (-4*\leglength,0) to  (0,0);
            \IsomTensor{0,0}{\leglength}{\squarehalf}{${\cal V}$}{0}
            \IsomTensor{2*\leglength + 2*\squarehalf,0}{\leglength}{\squarehalf}{${\cal V}$}{0}
            \SingleDots{4*\leglength + 4*\squarehalf, 0}{2*\leglength}
            \IsomTensor{6*\leglength + 6*\squarehalf,0}{\leglength}{\squarehalf}{${\cal V}$}{0}
    \end{tikzpicture},
\end{align}
and {\sffamily{hLP}} the class of channels generated after tracing out a purification space.
\end{defn}

We emphasize that translational invariance and homogeneity are distinct notions: 
Homogeneity refers to constructing the channel by contracting the same tensor $N$ times with periodic boundary conditions, 
whereas translational invariance means that the channel, regarded as an operator, commutes with the translation operator.
All homogeneous channels are trivially TI, but the reverse implication does not hold. For example, consider CP maps built from homogeneous tensors whose output trace is only proportional to that of the input. Converting them into proper channels (i.e., TP maps) requires introducing a normalization via a boundary tensor, i.e., a special extra tensor different than the homogeneous bulk, which breaks homogeneity while leaving translational invariance intact. In this section, we focus on the class of channels that might necessitate such normalization, and mark them with an overline: ${\overline{\text{\sffamily{hMPQC}}}}$ and ${\overline{\text{\sffamily{hLP}}}}$ for the locally purified class. We analogously use $\overline{\mathrm{MPQC}_N^A}$ and $\overline{\mathrm{LP}_N^B}$ for the individual channels.

In Ref.~\cite{cuevas_fundamental_2016}, it was shown that there exist unnormalized classical --- i.e.\ diagonal in the computational basis --- states which can be represented for all system sizes as (unnormalized) TI MPDO, 
\begin{equation}
    \rho^{A}_{N}
    = \sum_{i_{1}\ldots i_{N}}
      \operatorname{Tr}\!\left(A^{i_{1}}\cdots A^{i_{N}}\right)
      |i_{1}\cdots i_{N}\rangle\langle i_{1}\cdots i_{N}|,
    \label{SMeq:MPDO}
\end{equation}
but for which no TI \textit{local purification} of the form
\begin{align}
    |\Psi(B)\rangle_{N}
    &= \sum_{i,e}
       \operatorname{Tr}\!\left(B^{i_{1},e_{1}}\cdots B^{i_{N},e_{N}}\right)
       |i_{1}e_{1}\cdots i_{N}e_{N}\rangle,
    \label{eq:sigma_B}
\end{align}
exists such that 
\[
    \sigma^{B}_{N} = \operatorname{Tr}_{e_{1},\ldots,e_{N}}
    |\Psi(B)\rangle_{N}\langle\Psi(B)|_{N}
    \quad\text{and}\quad
    \rho^{A}_{N}/\Tr\rho^{A}_{N} = \sigma^{B}_{N}/\Tr\sigma^{B}_{N}
    \;\;\forall N.
\]
Graphically, the previous result states that it is impossible to find for all (possibly unnormalized) MPDO $\rho^{A}_{N}$
\be    
\rho^{A}_{N} =
    \begin{array}{c}
        \begin{tikzpicture}[scale=\defaultscaling,baseline={([yshift=-0.75ex] current bounding box.center)}]
            \foreach \x in {0,...,0}{
                \SingleTrRight{(0,0)}
            }
            \foreach \x in {1,...,1}{
                \MPOTensor{(-\doubledx*\x,0)}{0.5}{.5}{$A$}{0}
            }
            \foreach \x in {2,...,2}{
                \MPOTensor{(-\doubledx*\x,0)}{0.5}{.5}{$A$}{0}
            }
            \foreach \x in {3,...,3}{
                \SingleDots{-\doubledx*\x,0}{\doubledx/2}
            }
            \foreach \x in {4,...,4}{
                \MPOTensor{(-\doubledx*\x,0)}{0.5}{.5}{$A$}{0}
            }
            \foreach \x in {5,...,5}{
                \SingleTrLeft{(-\doubledx*\x,0)}
            }
        \end{tikzpicture}
    \end{array},
\ee
a $\sigma^{B}_{N}$
\be
\label{SMlocalpurification}
    \sigma^{B}_{N} = 
    \begin{array}{c}
        \begin{tikzpicture}[scale=\defaultscaling,baseline={([yshift=-0.75ex] current bounding box.center)}]
            \foreach \x in {1,...,1}{
                \DoubleMPOTensor{(-\doubledx*\x,0)}{0.5}{0.5}{$B$}{0}
                \SingleTrRight{(0,1.0)}
                \begin{scope}[yscale=-1]{
                    \SingleTrRight{(0,1.0)};
                }
                \end{scope}
            }
            \foreach \x in {2,...,2}{
                \DoubleMPOTensor{(-\doubledx*\x,0)}{0.5}{0.5}{$B$}{0}
            }
            \foreach \x in {3,...,3}{
                \DoubleDots{-\doubledx*\x,0}{\doubledx/2}{1.0};
            }
            \foreach \x in {4,...,4}{
                \DoubleMPOTensor{(-\doubledx*\x,0)}{0.5}{0.5}{$B$}{0}
            }
            \foreach \x in {5,...,5}{
                \SingleTrLeft{(-\doubledx*\x,1.0)};
                \begin{scope}[yscale=-1]{
                    \SingleTrLeft{(-\doubledx*\x,1.0)};
                }
                \end{scope}
            }
        \end{tikzpicture}
    \end{array}
    ,
\ee
such that $\rho^{A}_{N}$ and $\sigma^{B}_{N}$ are equal after normalization.

We remark that the tensors $A$ and $B$ are independent of $N$.
The definitions of MPQC and LP allow for general input and output dimensions; they need not be equal. Using this freedom, a trivial example to verify the inequivalence of the MPQC and LP classes can be constructed. The example consists in a family of channels without input outputting the state given in Ref.~\cite{cuevas_fundamental_2016}. This family cannot be brought into the local purification form of \cref{SMlocalpurification}, and hence the channel does not admit an LP representation.

In order to show an example of TI MPQC with non-trivial input not admitting TI LP representation, we again use a replacement channel $\mathcal{E}(\cdot) = \tr(\cdot) \rho$ and the fact that its Choi state is $C_{\mathcal E} = \rho \otimes \Id$; thus, the question of the local purification for $\mathcal E$ reduces to that of $\rho$. Using this, combined with the results of Ref.~\cite{cuevas_fundamental_2016}, we directly obtain:

\begin{prop}[${\overline{\text{\sffamily{hMPQC}}}} \ne {\overline{\text{\sffamily{hLP}}}}$]
\label{TIinequivalence}
There exists a tensor A such that there exists no tensor B such that
\be
\overline{\mathrm{MPQC}_N^A} \;=\; \overline{\mathrm{LP}_N^B}
\quad\text{for all }N.
\ee
\end{prop}

\section{Structure of Homogeneous MPI}
\label{sec:MPIstructure}
Here, we analyze the structure of channels in {\sffamily{hLP}}. By definition, these are TI and admit a dilation which is an MPI, that is, an isometry with a 1D locality, which is expected to hold in many physical models of system-environment interaction. We therefore work directly with the {\sffamily{hMPI}} class, which we characterize. In the literature, a class of locally purified channels has already been studied extensively, namely homogeneous unitary channels in matrix-product form, constructed from homogeneous matrix product unitaries (MPUs)~\cite{cirac_matrix_2017-1,sahinoglu2018matrix}. These correspond to the case where the purification space is trivial, i.e., of dimension 1, and input and output dimensions coincide. In Ref.~\cite{cirac_matrix_2017-1}, a standard form for homogeneous MPUs was derived and in the one-dimensional case, they were shown to coincide with Quantum Cellular Automata (QCAs)~\cite{farrelly_review_2020}, i.e., unitaries whose Heisenberg evolution maps strictly local observables to observables supported within a finite light cone. Crucially, the QCA property is a consequence of the imposition of homogeneity and does not follow solely from translational invariance~\cite{styliaris2025matrix}. Our goal is to now determine to what extent these results can be extended to the present isometric setting, and to identify which aspects do not carry over. We specifically address the question of whether the locally purified form, and specifically the homogeneity of the MPI, imposes additional constraints on the correlations present in the output states of the channel.

As noted earlier, MPIs arise from the condition $V_N^\dagger V^{}_N = \mathbb{I}_d^{\otimes N}$, which ensures trace preservation of channels in {\sffamily{LP}}.
Our main finding, analogous to the MPU case, is that any homogeneous MPI can be represented as a depth-two quantum circuit of isometric gates, after blocking. Blocking refers to the procedure of grouping $q$ neighboring tensors into a single effective tensor (see, e.g.,~\cite{cirac_matrix_2021}). The resulting tensor has physical dimension $d^q$ while keeping the bond dimension unchanged. For a one-site tensor ${\cal V}$, we denote by ${\cal V}_k$ the tensor obtained after blocking $k$ consecutive sites. Our main result is:

\begin{thm}
\label{SMdepth2}
        Any hMPI can be written as a depth-two brick wall quantum circuit of isometric gates $u,v$, satisfying $u^\dagger u =\Id_{d_{\mathrm{in}}^2}$ and $v^\dagger v =\Id_{\ell r}$, after blocking at most $D^4$ times:
\pgfmathsetmacro{\spacing}{1.5cm} % define spacing once
\begin{align*}
\begin{tikzpicture}[scale=\defaultscaling,x=1.3*\spacing]
    \definecolor{cream}{RGB}{255,255,221}
    \pgfmathsetmacro{\height}{1.5}
    \pgfmathsetmacro{\shift}{0.7}
    \pgfmathsetmacro{\labelpositions}{\height / 2 + \shift / 2 + 0.2}
    \pgfmathsetmacro{\labelcloseness}{0.15}
    \pgfmathsetmacro{\extra}{.1}
    \foreach \i in {0,...,6} {
        \draw[color=black,very thick] (\i,\shift) -- (\i,\height);
        \draw[color=black,very thick] (\i,-\shift) -- (\i,-\height);
    }
    \foreach \i in {0,2,4,6} {
        \draw[color=black,dashed, very thick] (\i,-\shift) -- (\i,\shift);
        \draw[color=purification,very thick] (\i,\shift) -- (\i,\height);
        \draw[color=black,very thick] (\i + \extra,\shift) -- (\i + \extra,\height);
        \draw (\i - \labelcloseness, 0) node {\scriptsize $r$};
        \draw (\i + \labelcloseness + 0.05, -\labelpositions) node {\scriptsize $d_{\mathrm{in}}$};
        \draw (\i - \labelcloseness, \labelpositions) node {\scriptsize $\chi$};
        \draw (\i + \labelcloseness + \extra + 0.1, +\labelpositions) node {\scriptsize $d_{\mathrm{out}}$};
    }
    \foreach \i in {1,3,5} {
        \draw[color=black,ultra thick] (\i,-\shift) -- (\i,\shift);
        \draw[color=purification,very thick] (\i - \extra,\shift) -- (\i - \extra,\height);
        \draw (\i + \labelcloseness, 0) node {\scriptsize $\ell$};
        \draw (\i + \labelcloseness + 0.05, -\labelpositions) node {\scriptsize $d_{\mathrm{in}}$};
        \draw (\i + \labelcloseness + 0.1, +\labelpositions) node {\scriptsize $d_{\mathrm{out}}$};
        \draw (\i - \labelcloseness - \extra, +\labelpositions) node {\scriptsize $\chi$};
    }
    \foreach \i in {0,2,4} {
        \draw[thick, fill=cream, rounded corners=2pt]
            (\i - \extra, -0.3 + \shift) rectangle (\i+1+\extra, 0.3 + \shift);
        \draw (\i+0.5,\shift) node {\scriptsize $v$};
    }
    \foreach \i in {1,3,5} {
        \draw[thick, fill=cream, rounded corners=2pt]
            (\i - \extra, -0.3 - \shift) rectangle (\i+1 + \extra, 0.3 - \shift);
            \draw (\i+0.5,-\shift) node {\scriptsize $u$};
    }
    % Add right half of bottom rectangle at x=0
    \begin{scope}
        \clip (-0.5, -\height) rectangle (0.1 + \extra, \height); % Clip to right half
        \draw[thick, fill=cream, rounded corners=2pt]
            (-1 - \extra, -0.3 - \shift) rectangle (0 + \extra, 0.3 - \shift);
    \end{scope}
%
    % Add left half of top rectangle at x=6
    \begin{scope}
        \clip (6 - \extra - 0.1, -\height) rectangle (6 + 0.5, \height); % Clip to left half
        \draw[thick, fill=cream, rounded corners=2pt]
            (6 - \extra, -0.3 + \shift) rectangle (6 + 1 + \extra, 0.3 + \shift);
    \end{scope}
\end{tikzpicture}
\end{align*}
The isometries $u:\mathbb{C}^{d_{\mathrm{in}}^2}\mapsto \mathbb{C}^{\ell r}$ and $v:\mathbb{C}^{\ell r}\mapsto \mathbb{C}^{\chi^2d_{\mathrm{out}}^2}$ satisfy $d_{\mathrm{in}}^2 \leq r\ell \leq d_{\mathrm{out}}^2 \chi^2$.
\end{thm}
 
The origin of $r$ and $\ell$ will be specified later, in \cref{decomposition}. Since any circuit of the above form gives rise to a homogeneous MPI (possibly after blocking), this characterizes the {\sffamily{hMPI}} family.
Physically, the result imposes a bound on the distance that the {\sffamily{hLP}} channels are capable of generating correlations. Indeed, consider the connected two-point correlation function 
\[
C(i,j) = \langle O_i^A O_j^B \rangle - \langle O_i^A \rangle \langle O_j^B \rangle ,
\]
where the operators $O_i^A$ and $O_j^B$ act on sites $i$ and $j$, respectively. In the case of a brick-wall circuit of depth~2 acting on a product initial state, one finds that $C(i,j) = 0$ whenever $|j-i|> 4$. Since we performed blocking only a finite number of times that depends on the bond dimension $D$ but not on the system size $N$, the resulting light cone is strictly finite and independent of $N$.

To prove the theorem, a few definitions and intermediate lemmas are needed.
Firstly, we need the concept of \emph{normality} and \emph{canonical form}~\cite{perez-garcia_matrix_2007, cirac_matrix_2021}. We provide the definitions for a tensor $A$ generating a (homogeneous) matrix product state (MPS):
\begin{equation}
    \ket{\Psi({A})}_{N}=\sum_{i_{1} \ldots i_{N}} \operatorname{Tr}\left(A^{i_{1}} \cdots A^{i_{N}}\right)\left|i_{1} \cdots i_{N}\right\rangle.
\end{equation}
This applies to our case simply by vectorizing the physical input and output legs of the MPI tensor together into a single index $i$. We define the transfer matrix
\be
E =  \frac{1}{d_\mathrm{in}} \; \begin{tikzpicture}[scale=\defaultscaling,baseline={([yshift=-0.65ex] current bounding box.center)}]
    \TransferMatrixTensor{0,0}{\leglength}{0.4}{${\cal V}$}{0}
    \end{tikzpicture}
    \;.
\ee

\begin{defn}[Normal tensor]
A tensor, $A$, is a normal tensor (NT) if: (i) there exists no non-trivial projector $P$ such that $A^i P=P A^i P$; (ii) its associated transfer matrix has a unique eigenvalue of magnitude (and value) equal to its spectral radius, which is equal to one. 
\end{defn}

\begin{defn}[Canonical Form]
A tensor, $A$, is in a canonical form (CF), if
$$
A^i=\bigoplus_{k=1}^r \mu_k A_k^i,
$$
where $\mu_k \in \mathbb{C}$ and the tensors $A_k$ are NT.
\end{defn}

Let us denote by $\Phi$ and $\rho$ the left and right eigenvectors of $E$ corresponding to the eigenvalue 1 of a normal tensor.

\begin{defn}[Canonical Form II]
    A normal tensor generating MPS is in Canonical Form II (CFII) if
$$
\begin{aligned}
 \bra{\Phi} &=\sum_{n=1}^D  \bra{n, n} , \\
\ket{\rho}&=\sum_{n=1}^D \rho_n \ket{n, n}.
\end{aligned}
$$
where $\rho_n>0$ and $\bra{\Phi} \rho\rangle=1$.
\end{defn}

In the homogeneous setting, every normal tensor can be brought to CF or CFII without loss of generality \cite{perez-garcia_matrix_2007}. The relevance of the CFII is that the transfer matrix of such a tensor can be interpreted as a channel, i.e., a CPTP map ruling the evolution of states in the bond space.
Finally, we need the concept of \textit{simplicity}.
\begin{defn}[Simplicity \cite{cirac_matrix_2017-1}]
    A tensor ${\cal V}$ is \textit{simple} if there exist two tensors $a$ and $b$ such that
\begin{subequations}
\begin{equation}
\label{simplicity1}
\begin{array}{c}
\begin{tikzpicture}[scale=\defaultscaling,baseline={([yshift=-0.8ex] current bounding box.center)}]
        \SingleSiteDouble{0,0}{\leglength}{0.4}{${\cal V}$}{0};
        \RightEigen{\leglength,0}{\leglength}{0.4}{b};
        \LeftEigen{-\leglength,0}{\leglength}{0.4}{a};
    \end{tikzpicture}
\end{array}
=
\begin{array}{c}
    \begin{tikzpicture}[scale=\defaultscaling,baseline={([yshift=-0.8ex] current bounding box.center)}]
        \pgfmathsetmacro{\identityheight}{4*\leglength + 0.4*2*2}
        \IdentityLine{0,0}{\identityheight}
    \end{tikzpicture}
\end{array}
\end{equation}
\begin{equation}
\label{simplicity2}
\begin{array}{c}
    \begin{tikzpicture}[scale=\defaultscaling,baseline={([yshift=-0.8ex] current bounding box.center)}]
        \SingleSiteDouble{0,0}{\leglength}{0.4}{${\cal V}$}{0}
        \SingleSiteDouble{2*0.4+2*\leglength,0}{\leglength}{0.4}{${\cal V}$}{0}
    \end{tikzpicture}
\end{array}
=
\begin{array}{c}
    \begin{tikzpicture}[scale=\defaultscaling,baseline={([yshift=-0.8ex] current bounding box.center)}]
        \SingleSiteDouble{0,0}{\leglength}{0.4}{${\cal V}$}{0};
        \RightEigen{\leglength,0}{\leglength}{0.4}{b};
        \LeftEigen{5*\leglength,0}{\leglength}{0.4}{a};
        \SingleSiteDouble{6*\leglength,0}{\leglength}{0.4}{${\cal V}$}{0};
    \end{tikzpicture}
\end{array}
\end{equation}
\end{subequations}
\end{defn}

The proof of \cref{SMdepth2} consists of combining the following lemmas, which generalize the results of Ref.~\cite{cirac_matrix_2017-1} from MPUs to MPIs.

\begin{lem}
\label{Vnormal}
The operator $E$ possesses a single nonzero eigenvalue, equal to one, and the tensor $\frac{1}{\sqrt{d_{\mathrm{in}}}}{\cal V}$ is normal.
\end{lem}

\begin{proof}
This proof follows the one in Ref.~\cite{cirac_matrix_2017-1}. Since $V_{N}$ is an isometry for every $N>1$, it follows that 
\[
\operatorname{tr}\!\left(E^N\right)
= \frac{1}{d_{\mathrm{in}}^N} \operatorname{tr}\!\left(V_N^{\dagger} V_N \right)
= 1, \quad \forall\, N>1.
\]
Hence, the spectrum of $E$ consists of a single eigenvalue equal to one, while all others vanish, as a direct consequence of Lemma A.5 in Ref.~\cite{cirac_matrix_2017}. This rules out the possibility of having more than one block in the canonical form. Moreover, since the transfer matrix cannot exhibit multiple eigenvalues of modulus one, it follows that the tensor is normal.
\end{proof}

Since in the homogeneous setting every normal tensor can be brought to CFII without loss of generality \cite{perez-garcia_matrix_2007}, we can assume that $\frac{1}{\sqrt{d_{\mathrm{in}}}} {\cal V}$ is in CFII.
In particular, we verify the following:

\begin{lem}
\label{simpleafterblocking}
For any tensor ${\cal V}$ generating an MPI, there exists an integer $k \leq D^4$ such that ${\cal V}_k$ is simple. In this case, the tensors $a$ and $b$ correspond to the left and right fixed points of the transfer operator $E$, denoted by $\Phi = \Id$ and $\rho$, respectively.
\end{lem}

\begin{proof}
By \cref{Vnormal}, $\frac{1}{\sqrt{d_{\mathrm{in}}}}{\cal V}$ is a normal tensor. Moreover, $V^{\dagger}_N V_N = \mathbb{I}$ and ${\cal V}$ is in CFII.
These are the only conditions used in the proof of Proposition III (b) of Ref.~\cite{cirac_matrix_2017-1}. Hence, the result follows by the exact same argument.
\end{proof}

\begin{lem}
\label{decomposition}
For any MPI tensor $\mathcal{V}$, there exist two decompositions $X_1 Y_1$, $X_2 Y_2$, obtained by grouping either the left auxiliary leg or the right auxiliary leg with the bottom physical index, while the remaining three indices are grouped together.  
The first decomposition has rank $r$ and the second has rank $\ell$.  
The tensors $X_i$ and $Y_i$ can be chosen such that
\be
\label{identities_lemma}
X_2^{\dagger} X_2=\mathbb{I}, 
\qquad 
X_1^{\dagger}(\mathbb{I}\otimes\rho) X_1=\mathbb{I},
\ee
where $\rho$ is a positive diagonal matrix.
Graphically, these two equivalent decompositions are represented as
\begin{align}
    \begin{array}{c}
        \begin{tikzpicture}[scale=0.9,baseline={([yshift=-0.65ex] current bounding box.center)}]
            \IsomTensor{0,0}{\leglength}{0.3}{$\mathcal{V}$}{0}
        \end{tikzpicture}
    \end{array}
        =
    \begin{array}{c}
        \begin{tikzpicture}[scale=0.5,baseline={([yshift=-0.65ex] current bounding box.center)}]
            \draw (-0.4,0.5) node {$r$};
            \LeftDeco{0,0}
        \end{tikzpicture}
    \end{array}
    =
    \begin{array}{c}
        \begin{tikzpicture}[scale=0.5,baseline={([yshift=-0.65ex] current bounding box.center)}]
            \draw (0.4,0.5) node {$\ell$};
            \RightDeco{0,0}
        \end{tikzpicture}
    \end{array}.
\end{align}
\end{lem}

\begin{proof}
This proof is a straightforward generalization of the procedure in Ref.~\cite{  cirac_matrix_2017-1}.
In contrast to the MPU case, the only modification is the presence of an additional output leg of dimension $\chi$. 

In the first decomposition, the left auxiliary and the bottom physical index are grouped together, just as the remaining three indices. In the second decomposition, the grouping is reversed. Applying a singular value decomposition, we obtain
$$
\mathcal{M}_{1,2}=V_{1,2}^{\dagger} D_{1,2} U_{1,2},
$$
where the $U$’s and $V$’s are isometries, satisfying $V_i V_i^{\dagger}= U_i U_i^{\dagger}=\mathbb{I}$, and $D_{1,2}$ are diagonal positive matrices of dimensions $r$ and $\ell$, respectively. 

We now define matrices $X_i$ and $Y_i$ such that $\mathcal{M}_i=X_i Y_i$, expressed in terms of the above decompositions. Concretely, we set $X_2=V_2^{\dagger}$ and
\[
X_1= V_1^{\dagger}\left[V_1(\mathbb{I} \otimes \rho) V_1^{\dagger}\right]^{-1 / 2},
\]
\[
Y_1= \left[V_1(\mathbb{I} \otimes \rho) V_1^{\dagger}\right]^{1 / 2} D_1 U_1,
\]
\[
Y_2= D_2 U_2,
\]
so that the following relations hold:
\be
\label{identities}
X_2^{\dagger} X_2=\mathbb{I}, \quad X_1^{\dagger}(\mathbb{I} \otimes \rho) X_1=\mathbb{I} .
\ee
Here, $\rho$ is a diagonal matrix with entries $\rho_n$.

We identify some tensors graphically:
\begin{subequations}
    \begin{align}
        X_1 &=\begin{tikzpicture}[scale=\defaultscaling,baseline={([yshift=-0.65ex] current bounding box.center)}]
    \XoneDeco{0,0}
    \end{tikzpicture},&
    X_2 &=\begin{tikzpicture}[scale=\defaultscaling,baseline={([yshift=-0.65ex] current bounding box.center)}]
    \XtwoDeco{0,0}
    \end{tikzpicture},
    \end{align}

    \begin{align}
        Y_1 &=\begin{tikzpicture}[scale=\defaultscaling,baseline={([yshift=-0.65ex] current bounding box.center)}]
        \YoneDeco{0,0}
    \end{tikzpicture},&
    Y_2 &=\begin{tikzpicture}[scale=\defaultscaling,baseline={([yshift=-0.65ex] current bounding box.center)}]
        \YtwoDeco{0,0}
    \end{tikzpicture}.
    \end{align}
\end{subequations}

\end{proof}

With these at disposal, we further define:
\begin{align}
\label{uuvv}
u &= \begin{tikzpicture}[scale=\defaultscaling,baseline={([yshift=-0.65ex] current bounding box.center)}]
    \IsomU{0,0}
    \end{tikzpicture}, &
v &= \begin{tikzpicture}[scale=\defaultscaling,baseline={([yshift=-0.65ex] current bounding box.center)}]
    \IsomV{0,0}
    \end{tikzpicture}.
\end{align}

\begin{lem}
\label{lemuisometry}
For any MPI generating tensor ${\cal V}$, it holds that \( u^\dagger u = \Id_{d_{\mathrm{in}}^2} \), and hence \( r\ell \geq d_{\mathrm{in}}^2 \).
\end{lem}

\begin{proof}
The proof follows directly from the MPU case in Ref.~\cite{  cirac_matrix_2017-1}.
\end{proof}
This lemma implies that $u$ is an isometry. If ${\cal V}$ is simple (or rendered simple via blocking), then $v$ is also isometric.

\begin{lem}
\label{simpleifisometry}
A hMPI tensor ${\cal V}$ is simple if and only if $v$ is isometric, i.e., \( v^\dagger v= \Id_{r\ell} \), implying $r\ell \leq d_{\mathrm{out}}^2 \chi^2$.
\end{lem}

\begin{proof} 
\begin{itemize}
  \item[\((\Rightarrow)\)] 
   \begin{equation}
   \label{simplehenceisom}
       \begin{array}{c}
            \begin{tikzpicture}[scale=0.4,baseline={([yshift=-0.8ex] current bounding box.center)}]
                \IsomV{0,0}
                \YoneDeco{0,1.5}
                \YtwoDeco{2.,1.5}
                \begin{scope}[yshift=4.5cm, yscale=-1] 
                    \IsomV{0,0}
                    \YoneDeco{0,1.5}
                    \YtwoDeco{2.,1.5}
                \end{scope}
            \end{tikzpicture}
\end{array}
=
\begin{array}{c}
    \begin{tikzpicture}[scale=\defaultscaling,baseline={([yshift=-0.8ex] current bounding box.center)}]
        \SingleSiteDouble{0,0}{\leglength}{0.4}{${\cal V}$}{0}
        \RightEigen{\leglength,0}{\leglength}{0.4}{$\rho$}
        \LeftEigenStripped{5*\leglength - 0.5,0}{\leglength}{0.4}
        \SingleSiteDouble{6*\leglength - 0.5,0}{\leglength}{0.4}{${\cal V}$}{0}
    \end{tikzpicture}
\end{array}
=
\begin{array}{c}
    \begin{tikzpicture}[scale=0.4,baseline={([yshift=-0.8ex] current bounding box.center)}]
                \YoneDeco{0,1.5}
                \YtwoDeco{2.,1.5}
                \begin{scope}[yshift=2.5cm, yscale=-1] 
                    \YoneDeco{0,1.5}
                    \YtwoDeco{2.,1.5}
                \end{scope}
            \end{tikzpicture},
\end{array}
\end{equation}
 where in the last step we have used \cref{identities}. Using the invertibility of $Y_1$ and $Y_2$, we get $v^\dagger v= \Id_{r\ell}$, and hence $v$ is an isometry.

  \item[\((\Leftarrow)\)]
  The fact that $v$ is isometric provides us with the equality between the LHS and the RHS of \cref{simplehenceisom}. Plugging \cref{identities} into the RHS, we get the equality with the middle term in \cref{simplehenceisom}.
\end{itemize}
\end{proof}

For \( \chi = 1 \) and $d \coloneq d_{\mathrm{in}} = d_\mathrm{out}$, we recover the equality for the MPU case: $r \ell = d^2$, meaning the maps \(u\) and \(v\) are unitaries instead of isometries.

\begin{proof}[Proof of \cref{SMdepth2}]
    By blocking at most $D^4$ times, we obtain a simple tensor (\cref{simpleafterblocking}). \cref{decomposition} then shows the gate structure, and \cref{lemuisometry} and \cref{simpleifisometry} the isometric nature of the gates.
\end{proof}

    \section{Classification of Locally Purifiable (LP) Channels}
\label{sec:Classification}
We further analyze {\sffamily hLP} channels and ask whether they admit an index classification, in analogy with homogeneous MPUs. Recall that 1D QCA, i.e., unitary maps such that the support of the evolution of any local operator is bounded to a neighborhood of the original support, can be completely classified by an positive-rational-valued invariant known as the \emph{index}~\cite{gross_index_2012,schumacher_reversible_2004}. Intuitively, the index captures the net flow of quantum information across the system. An example of QCA is the one-site \emph{shift}, which translates each degree of freedom by one lattice site. Its index is strictly larger than zero, and it cannot be continuously connected to a finite-depth circuit, which is the paradigmatic example of index-zero QCAs.
The latter are also called \emph{trivial}, meaning that they are in the same \textit{phase} as finite-depth circuits. Two QCAs $\alpha_0$ and $\alpha_1$ belong to the same phase if there exists a continuous family of QCAs $\{\alpha(s)\}_{s \in [0,1]}$ such that $\alpha(0) = \alpha_0, \; \alpha(1) = \alpha_1,$ and for every $s \in [0,1]$, $\alpha(s)$ is a valid QCA.
Since homogeneous MPUs are equivalent to 1D QCA, this classification result directly extends to the homogeneous MPU setting~\cite{cirac_matrix_2017-1}, and provides the natural starting point to investigate the presence of an analogous index for {\sffamily{hLP}} channels.
 
We now show that in our case, exploiting the freedom that the extra output space $\chi$ grants us, all MPIs can be continuously deformed one into the other. In other words, all channels in {\sffamily{hLP}} are in the same phase, i.e., they are associated with the same index.

We hereby clarify what is meant by phases and equivalence. We will consider extensions of channels by the addition of input, output, and purification dimensions. More specifically, we allow for an extra input of dimension $x$ that will be simply traced out, which we denote as $\operatorname{Tr}^{(x)}$ and some extra output of dimensions $y$ in the form $\ket{0}\bra{0}_{y}$. We will denote $\mathcal{E}^{(x,y)} = \mathcal{E} \otimes \operatorname{Tr}^{(x)} \otimes \ket{0}\bra{0}_{y}$. Moreover, we can increase the dimension of the purification space, as this does not change the physical channel. With these extensions, we can account for channels with unequal input and output dimensions.

Generalizing the definitions of Ref.~\cite{cirac_matrix_2017-1}, we define the following.

\begin{defn}[Strict Equivalence]
Let $\mathcal E_0$ and $\mathcal E_1$ be channels in {\sffamily{hLP}} 
with identical input and output dimensions $d_{\mathrm{in}}$ and $d_{\mathrm{out}}$. 
They are \emph{strictly equivalent} if there exists a continuous map of hMPI-generating tensors
\begin{align}
    [0,1]\ni s \mapsto \mathcal{V}(s) =
    \begin{array}{c}
        \begin{tikzpicture}[scale=0.9,baseline={([yshift=-0.65ex] current bounding box.center)}]
            \IsomTensor{0,0}{\leglength}{0.4}{$\mathcal{V}(s)$}{0}
            \draw (-0.55*\leglength,0.65) node {\scriptsize $\chi$};
            \draw (0.85*\leglength,0.65) node {\scriptsize $d_{\mathrm{out}}$};
            \draw (1.3*\leglength,0.15) node {\scriptsize $D$};
            \draw (-1.3*\leglength,0.15) node {\scriptsize $D$};
            \draw (0.25,-0.65) node {\scriptsize $d_{\mathrm{in}}$};
        \end{tikzpicture}
    \end{array}
\quad D,\chi \text{ $\rm const.$} \;,
\end{align}
with associated MPI $V_N(s)$, such that
\begin{align}
    \mathcal{E}_i = \Tr_{\mathrm{pur}}[V_N(i)^{}(\cdot)V_N(i)^\dagger], \quad i \in \{ 0,1 \} \quad \forall N\;.
\end{align}
\end{defn}

\begin{defn}[Equivalence]
     Two channels $\mathcal{E}_1$ and $\mathcal{E}_2$ in {\sffamily{hLP}} are \emph{equivalent} if there exists integers $(p_1, p_2)$, $(p_1', p_2')$ such that $\mathcal{E}_1^{(p_1,p_1')}$ and $\mathcal{E}_2^{(p_2,p_2')}$ are strictly equivalent. 
\end{defn}

Strict equivalence essentially captures that two {\sffamily{hLP}} channels, with the same input-output dimensions, can be continuously connected within the {\sffamily{hLP}} class. Equivalence generalizes this to arbitrary input-output dimensions by allowing for auxiliary qubits and tracing out, but without affecting the resulting channel, except at the level of having a redundant input (which is traced out) and outputting an additional (fixed) product state. Note that our definitions, in the absence of the purification space, reduce to the established QCA classification in the MPU formalism~\cite{gross_index_2012,cirac_matrix_2017-1}. Our main result is:

\begin{thm}
    Any two channels in {\sffamily{hLP}} are equivalent.
\end{thm}

\begin{proof}
For simplicity, let us consider the proof for strict equivalence. We will comment in the end on how to generalize to mismatching input and output dimensions. 
    Given two channels  $\mathcal{E}_1$ and $\mathcal{E}_2$ in {\sffamily{hLP}} of bond dimensions $D_1$ and $D_2$, consider a choice of associated MPI tensors $\mathcal{V}_1$ and $\mathcal{V}_2$, with purification space $\chi_1$ and $\chi_2$ respectively, satisfying $\chi_1 p_1'' = \chi_2 p_2''$, where $p_1''$ and $p_2''$ can be chosen, wlog, to be coprimes. Construct the tensor $\mathcal{V}(s)$ as follows:
\begin{equation}
\label{deformation}
    \mathcal V(s) =
    \begin{tikzpicture}[scale=\defaultscaling,baseline={(0,-0.1)}]
        \definecolor{cream}{RGB}{255,255,221}
        \definecolor{purification}{RGB}{124, 138, 128}
        \pgfmathsetmacro{\height}{3.5}
        \pgfmathsetmacro{\spacing}{0.5}
        \pgfmathsetmacro{\shift}{1.}
        \foreach \i in {2,5} {
            \IdentityLine{\i*\spacing,0}{\height}
            \draw (\i*\spacing,-\height/2 - 0.3) node[rotate=-90] {\scriptsize $\ket{0}$};
        }
        \foreach \i in {0,3} {
            \IdentityLine{\i*\spacing ,-\height/4}{\height/2}
            \ifnum\i>0
                \draw (\i*\spacing,-\height/2 - 0.3) node[rotate=-90] {\scriptsize $\ket{0}$};
            \fi
        }
        \foreach \i in {0,...,5} {
            \IdentityLine{\i*\spacing,\height/4}{\height/2}
        }
        \foreach \i in {1,...,5} {
            \draw[very thick, color=purification] (\i*\spacing,\height/4) -- (\i*\spacing,\height/2);
        }
        \draw[very thick] (-0.5,-0.1) -- (5*\spacing + 0.5,-0.1);
        \draw[very thick] (-0.5,0.1) -- (5*\spacing + 0.5,0.1);
        \draw[ thick, fill=tensorcolor, rounded corners=2pt] (-0.2,-0.3) rectangle (5*\spacing + 0.2,0.3);
        \draw (5*\spacing / 2 ,0) node {\scriptsize ${\cal V}$};
        \draw[ thick, fill=cream, rounded corners=4pt] (-0.2,-0.3 + \shift) rectangle (5*\spacing + 0.2,0.3 +\shift);
        \draw (5*\spacing / 2 ,\shift) node {\scriptsize $U(s)$};
        \draw[ thick, fill=cream, rounded corners=4pt] (-0.2,-0.3 - \shift) rectangle (5*\spacing + 0.2,0.3 -\shift);
        \draw (5*\spacing / 2 ,-\shift) node {\scriptsize $W(s)$};
    \end{tikzpicture}
    =
        \begin{tikzpicture}[scale=\defaultscaling,baseline={(0,-0.1)}]
        \definecolor{cream}{RGB}{255,255,221}
        \definecolor{purification}{RGB}{124, 138, 128}
        \pgfmathsetmacro{\height}{3.5}
        \pgfmathsetmacro{\spacing}{0.5}
        \pgfmathsetmacro{\shift}{1.}
        \foreach \i in {0,1,2,5,6,7} {
            \IdentityLine{\i*\spacing,\height/4}{\height/2}
        }
        \foreach \i in {1,2,5,6,7} {
            \draw[very thick, color=purification] (\i*\spacing,\height/4) -- (\i*\spacing,\height/2);
        }
        \foreach \i in {2,7} {
            \IdentityLine{\i*\spacing,-\height/4}{\height/2}
            \ifnum\i>0
                \draw (\i*\spacing,-\height/2 - 0.3) node[rotate=-90] {\scriptsize $\ket{0}$};
            \fi
        }
        \foreach \i in {0,5} {
            \IdentityLine{\i*\spacing + \spacing/2,-\height/4}{\height/2}
            \ifnum\i>0
                \draw (\i*\spacing+ \spacing/2,-\height/2 - 0.3) node[rotate=-90] {\scriptsize $\ket{0}$};
            \fi
        }
        \draw[very thick] (-0.7,0) -- (\spacing + 0.7,0);
        \draw[very thick] (-0.7 + 5*\spacing,0) -- (6*\spacing + 0.7,0);
        \draw[ thick, fill=tensorcolor, rounded corners=2pt] (-0.2,-0.3) rectangle (\spacing + 0.2,0.3);
        \draw (\spacing / 2 ,0) node {\scriptsize $ \mathcal V_1$};
        \draw[ thick, fill=tensorcolor, rounded corners=2pt] (-0.2 + 5*\spacing,-0.3) rectangle (6*\spacing + 0.2,0.3);
        \draw (5.5*\spacing,0) node {\scriptsize $\mathcal V_2$};
        \draw[ thick, fill=cream, rounded corners=4pt] (-0.2,-0.3 + \shift) rectangle (7*\spacing + 0.2,0.3 +\shift);
        \draw (7*\spacing / 2 ,\shift) node {\scriptsize $U(s)$};
        \draw[ thick, fill=cream, rounded corners=4pt] (-0.2,-0.3 - \shift) rectangle (7*\spacing + 0.2,0.3 -\shift);
        \draw (7*\spacing / 2 ,-\shift) node {\scriptsize $W(s)$};
    \end{tikzpicture}
    ,
   \end{equation}
   with dimensions
   \be
            \begin{tikzpicture}[scale=\defaultscaling,baseline={(0,-0.1)}]
        \definecolor{cream}{RGB}{255,255,221}
        \pgfmathsetmacro{\height}{2.5}
        \pgfmathsetmacro{\spacing}{1.5}
        \pgfmathsetmacro{\shiftlabels}{.3}
        \foreach \i in {0,...,5} {
            \IdentityLine{\i*\spacing,\height / 4}{\height / 2}
        }
        \foreach \i in {0,2,3,5} {
            \IdentityLine{\i*\spacing,0}{\height}
        }
        \draw (0-\shiftlabels - 0.2,\height/2 - 0.3) node {\scriptsize $d_{\mathrm{out}}$};
        \draw (\spacing-\shiftlabels,\height/2 - 0.3) node {\scriptsize $\chi_1$};
        \draw (2*\spacing-\shiftlabels,\height/2 - 0.3) node {\scriptsize $p_1''$};
        \draw (3*\spacing-\shiftlabels - 0.2,\height/2 - 0.3) node {\scriptsize $d_{\mathrm{out}}$};
        \draw (4*\spacing-\shiftlabels,\height/2 - 0.3) node {\scriptsize $\chi_2$};
        \draw (5*\spacing-\shiftlabels,\height/2 - 0.3) node {\scriptsize $p_2''$};
        \draw (0-\shiftlabels - 0.1,-\height/2 + 0.3) node {\scriptsize $d_{\mathrm{in}}$};
        %\draw (\spacing-\shiftlabels,-\height/2 + 0.3) node {\scriptsize $\chi_1$};
        \draw (2*\spacing-\shiftlabels,-\height/2 + 0.3) node {\scriptsize $p_1''$};
        \draw (3*\spacing-\shiftlabels - 0.1,-\height/2 + 0.3) node {\scriptsize $d_{\mathrm{in}}$};
        %\draw (4*\spacing-\shiftlabels,-\height/2 + 0.3) node {\scriptsize $\chi_2$};
        \draw (5*\spacing-\shiftlabels,-\height/2 + 0.3) node {\scriptsize $p_2''$};
        \draw (5*\spacing + 0.5, 0.4) node {\scriptsize $D_1$};
        \draw (5*\spacing + 0.5, -0.4) node {\scriptsize $D_2$};
        \draw (0*\spacing - 0.5, 0.4) node {\scriptsize $D_1$};
        \draw (0*\spacing - 0.5, -0.4) node {\scriptsize $D_2$};
        \draw[very thick] (-1.0,-0.1) -- (5*\spacing + 1.0,-0.1);
        \draw[very thick] (-1.0,0.1) -- (5*\spacing + 1.0,0.1);
        \draw[ thick, fill=tensorcolor, rounded corners=2pt] (-0.2,-0.3) rectangle (5*\spacing + 0.2,0.3);
        \draw (5*\spacing / 2 ,0) node {\scriptsize ${\cal V}$};
    \end{tikzpicture}
   \ee
    As $U(s)$ and $W(s)$ endpoints, now choose identities for $s=0$ and respectively the cyclic permutation $\tau: (1,2,3,4,5,6) \;\mapsto\; (4,5,6,1,2,3)$ and $\tau': (1,2,3,4) \;\mapsto\; (3,4,1,2)$ for $s=1$. The path can be taken to be continuous, as the unitary group is connected, and the tensor generates an MPI for all $s$ (because of the unitarity of $U$ and $W$), and yields the correct endpoints $\mathcal{V}(0)$, $\mathcal{V}(1)$ (i.e.\ they generate $\mathcal{E}_1$, $\mathcal{E}_2$), by grouping and identifying the second to fifth output legs as $\chi$, i.e., the purification space to be traced out. 
    Graphically:
    \begin{align}
    \mathcal{V}(0) =
        \begin{tikzpicture}[scale=\defaultscaling,baseline={(0,-0.1)}]
        \definecolor{cream}{RGB}{255,255,221}
        \definecolor{purification}{RGB}{124, 138, 128}
        \pgfmathsetmacro{\height}{2.5}
        \pgfmathsetmacro{\spacing}{0.5}
        \pgfmathsetmacro{\shift}{1.}
        \foreach \i in {2,7} {
            \IdentityLine{\i*\spacing,0}{\height}
            \ifnum\i>1
                \draw (\i*\spacing,-\height/2 - 0.3) node[rotate=-90] {\scriptsize $\ket{0}$};
                \draw[very thick, color=purification] (\i*\spacing,-\height/2) -- (\i*\spacing,\height/2);
            \fi
        }
        \foreach \i in {0,1, 5, 6} {
            \IdentityLine{\i*\spacing,\height/4}{\height/2}
        }
        \foreach \i in {1, 5, 6} {
            \draw[very thick, color=purification] (\i*\spacing,0) -- (\i*\spacing,\height/2);
        }
        \foreach \i in {0, 5} {
            \IdentityLine{\i*\spacing + \spacing /2,-\height/4}{\height/2}
                        \ifnum\i>1
                \draw (\i*\spacing + \spacing /2,-\height/2 - 0.3) node[rotate=-90] {\scriptsize $\ket{0}$};
            \fi
        }
        \draw[very thick] (-0.7,0) -- (\spacing + 0.7,0);
        \draw[very thick] (-0.7 + 5*\spacing,0) -- (6*\spacing + 0.7,0);
        \draw[ thick, fill=tensorcolor, rounded corners=2pt] (-0.2,-0.3) rectangle (\spacing + 0.2,0.3);
        \draw (\spacing / 2 ,0) node {\scriptsize $\mathcal V_1$};
        \draw[ thick, fill=tensorcolor, rounded corners=2pt] (-0.2 + 5*\spacing,-0.3) rectangle (6*\spacing + 0.2,0.3);
        \draw (5.5*\spacing,0) node {\scriptsize $\mathcal V_2$};
    \end{tikzpicture},
    &&
    \mathcal{V}(1) =
        \begin{tikzpicture}[scale=\defaultscaling,baseline={(0,-0.1)}]
        \definecolor{cream}{RGB}{255,255,221}
        \pgfmathsetmacro{\height}{2.5}
        \pgfmathsetmacro{\spacing}{0.5}
        \pgfmathsetmacro{\shift}{1.}
        \foreach \i in {2,7} {
            \IdentityLine{\i*\spacing,0}{\height}
            \ifnum\i>1
                \draw (\i*\spacing,-\height/2 - 0.3) node[rotate=-90] {\scriptsize $\ket{0}$};
                \draw[very thick, color=purification] (\i*\spacing,-\height/2) -- (\i*\spacing,\height/2);
            \fi
        }
        \foreach \i in {0,1, 5, 6} {
            \IdentityLine{\i*\spacing,\height/4}{\height/2}
        }
        \foreach \i in {1, 5, 6} {
            \draw[very thick, color=purification] (\i*\spacing,0) -- (\i*\spacing,\height/2);
        }
        \foreach \i in {0, 5} {
            \IdentityLine{\i*\spacing + \spacing /2,-\height/4}{\height/2}
                        \ifnum\i>1
                \draw (\i*\spacing + \spacing /2,-\height/2 - 0.3) node[rotate=-90] {\scriptsize $\ket{0}$};
            \fi
        }
        \draw[very thick] (-0.7,0) -- (\spacing + 0.7,0);
        \draw[very thick] (-0.7 + 5*\spacing,0) -- (6*\spacing + 0.7,0);
        \draw[ thick, fill=tensorcolor, rounded corners=2pt] (-0.2,-0.3) rectangle (\spacing + 0.2,0.3);
        \draw (\spacing / 2 ,0) node {\scriptsize $\mathcal V_2$};
        \draw[ thick, fill=tensorcolor, rounded corners=2pt] (-0.2 + 5*\spacing,-0.3) rectangle (6*\spacing + 0.2,0.3);
        \draw (5.5*\spacing,0) node {\scriptsize $\mathcal V_1$};
    \end{tikzpicture}
    .
   \end{align}
    This shows the equivalence.

    Consider now the general case in which the input and output dimensions satisfy
\[
d_{\mathrm{in},1} p_1 = d_{\mathrm{in},2} p_2 \coloneq d_{\mathrm{in}},
\qquad
d_{\mathrm{out},1} p_1' = d_{\mathrm{out},2} p_2' \coloneq d_{\mathrm{out}},
\]
where $(p_1, p_2)$ and $(p_1', p_2')$ are pairs of coprime integers.
It suffices to understand how the additional input and output degrees of freedom manifest at the level of the MPI.
On the {output side}, the extra $p_1'$ and $p_2'$ legs correspond to ancillary systems prepared in the state $\ket{0}$. When forming the channels $\mathcal{E}_1$ and $\mathcal{E}_2$, these ancillary outputs are respectively traced out (or not). Graphically, this is represented by grey (or black) legs, respectively.
On the {input side}, the extra $p_1$ and $p_2$ legs correspond to simple identities. These legs are always traced out in the definition of the channels. Concretely, the former trace over an additional input space of dimension $p_1$, while the latter traces over one of dimension $p_2$. Graphically, these legs are therefore gray. At the endpoints of the path, either the input of the $p_1$ legs or of the $p_2$ legs is fixed to some ket $\ket{0}$ respectively.
    This is shown graphically for the initial point of the path:
    \begin{equation}
    \mathcal{V}(0) =
        \begin{tikzpicture}[scale=\defaultscaling,baseline={(0,-0.1)}]
        \definecolor{cream}{RGB}{255,255,221}
        \definecolor{purification}{RGB}{124, 138, 128}
        \pgfmathsetmacro{\height}{2.5}
        \pgfmathsetmacro{\spacing}{0.5}
        \pgfmathsetmacro{\shift}{1.}
        \foreach \i in {2} {
            \draw (\i*\spacing,-\height/2 - 0.3) node[rotate=-90] {\scriptsize $\ket{0}$};
            \draw[very thick, color=purification] (\i*\spacing,-\height/2) -- (\i*\spacing,\height/2);
        }
        \foreach \i in {4} {
            \draw (\i*\spacing,-\height/2 - 0.3) node[rotate=-90] {\scriptsize $\ket{0}$};
            \draw[very thick, color=black] (\i*\spacing,-\height/2) -- (\i*\spacing,\height/2);
        }
        \foreach \i in {3} {
            \draw[very thick, color=purification] (\i*\spacing,-\height/2) -- (\i*\spacing,\height/2);
        }
        \foreach \i in {0,1} {
            \IdentityLine{\i*\spacing,\height/4}{\height/2}
        }
        \foreach \i in {1} {
            \draw[very thick, color=purification] (\i*\spacing,0) -- (\i*\spacing,\height/2);
        }
        \foreach \i in {0} {
            \IdentityLine{\i*\spacing + \spacing /2,-\height/4}{\height/2}
                        \ifnum\i>1
                \draw (\i*\spacing + \spacing /2,-\height/2 - 0.3) node[rotate=-90] {\scriptsize $\ket{0}$};
            \fi
        }
        \draw[very thick] (-0.7,0) -- (\spacing + 0.7,0);
        \draw[ thick, fill=tensorcolor, rounded corners=2pt] (-0.2,-0.3) rectangle (\spacing + 0.2,0.3);
        \draw (\spacing / 2 ,0) node {\scriptsize $\mathcal V_1$};

        \begin{scope}[xshift= 7*\spacing cm]
            \foreach \i in {2} {
            \draw (\i*\spacing,-\height/2 - 0.3) node[rotate=-90] {\scriptsize $\ket{0}$};
            \draw[very thick, color=purification] (\i*\spacing,-\height/2) -- (\i*\spacing,\height/2);
        }
        \foreach \i in {4} {
            \draw (\i*\spacing,-\height/2 - 0.3) node[rotate=-90] {\scriptsize $\ket{0}$};
            \draw[very thick, color=purification] (\i*\spacing,-\height/2) -- (\i*\spacing,\height/2);
        }
        \foreach \i in {3} {
            \draw[very thick, color=purification] (\i*\spacing,-\height/2) -- (\i*\spacing,\height/2);
            \draw (\i*\spacing,-\height/2 - 0.3) node[rotate=-90] {\scriptsize $\ket{0}$};
        }
        \foreach \i in {0,1} {
            \draw[very thick, color=purification] (\i*\spacing,0) -- (\i*\spacing,\height/2);
        }
        \foreach \i in {1} {
            \draw[very thick, color=purification] (\i*\spacing,0) -- (\i*\spacing,\height/2);
        }
        \foreach \i in {0} {
            \IdentityLine{\i*\spacing + \spacing /2,-\height/4}{\height/2}
            \draw (\i*\spacing + \spacing /2,-\height/2 - 0.3) node[rotate=-90] {\scriptsize $\ket{0}$};
                        \ifnum\i>1
                \draw (\i*\spacing + \spacing /2,-\height/2 - 0.3) node[rotate=-90] {\scriptsize $\ket{0}$};
            \fi
        }
        \draw[very thick] (-0.7,0) -- (\spacing + 0.7,0);
        \draw[ thick, fill=tensorcolor, rounded corners=2pt] (-0.2,-0.3) rectangle (\spacing + 0.2,0.3);
        \draw (\spacing / 2 ,0) node {\scriptsize $\mathcal V_2$};
        \end{scope}
        
    \end{tikzpicture}.
   \end{equation}
\end{proof}

Notice that, e.g.,\ shift and identity, treated as unitary channels, are now connected. In fact, we have proven that any two unitary \emph{channels}, even if they are formed by unitaries with a distinct QCA index, are equivalent. The reason for this apparent contradiction is that, when treated as channels, we allow for a purification space for the unitaries (in the form of an isometry), even if it is redundant for realizing the unitary. However, this purification space, which is eventually traced out, is what makes possible the smooth deformation of one unitary channel into the other, which is not always possible in the space of homogeneous MPUs alone.

It is instructive to recall the role of ancillary degrees of freedom in this context. While not all QCA admit a representation as finite-depth quantum circuits, it has been shown that the inclusion of ancillas allows any QCA to be trivialized~\cite{arrighi_unitarity_2009,farrelly_causal_2014}. Ancillas thus provide additional freedom that can be leveraged to implement otherwise nontrivial dynamics in a local manner. The underlying mechanism can be formalized using the global swap construction: one considers two copies of the system, evolving one via the original QCA dynamics and the other via the inverse, and then decomposes the resulting joint dynamics into local unitaries acting on non-overlapping neighborhoods~\cite{farrelly_review_2020}. In this construction, the dynamics of the joint system can be implemented in parallel as a finite-depth circuit, demonstrating that any QCA is effectively localizable once ancillas are included. The spirit of our construction is similar: the purification space can be leveraged to move information around, and then, crucially, it can be traced out (which corresponds to the ancillas decoupling in the QCA construction). These features enable deformations for channels that are not allowed in the setting of MPUs.

Let us finally stress that the previous definitions and procedure define a proper \textit{equivalence} relation. Indeed, any hLP channel can be "deformed" to itself (reflexivity). Moreover, if $a$ can be mapped to $b$ by a map $U(s)$ with $s \in [0,1]$, then the same map, with the parameter taken in reverse order $s \in [1,0]$, maps $b$ back to $a$ (symmetry). Finally, if $a$ can be mapped to $b$ and $b$ to $c$, one can apply the construction of \cref{deformation} with three tensors instead of two, applying unitaries to the first and second and subsequently to the second and third tensors, thereby obtaining a mapping from $a$ to $c$ (transitivity).

\section{Area-law Preserving Operations}\label{arealaw}

In this section, we characterize area-law-preserving quantum operations in 1D.

\subsection{Unitaries Preserving the Area Law}

For pure states, the area law is typically formulated in terms of the bipartite entanglement entropy. Here we consider the strictest form of the law, quantified by the $\alpha = 0$ R\'enyi entropy of the reduced state, $S_0 (\rho) = \log \rank (\rho)$. MPS satisfy this area law by construction, since their entanglement entropy across any bipartition is upper bounded by the logarithm of the bond dimension. Vice versa, every 1D pure state satisfying the area law is also a constant bond dimension MPS~\cite{cirac_matrix_2021}.

One may therefore ask which closed-system evolutions preserve the area law, namely which unitaries map MPS to MPS. This is equivalent to seeking the class of unitary transformations that do not increase the Schmidt rank across any bipartition beyond a constant multiplicative factor for all system sizes. MPUs have this property by construction since their operator Schmidt rank is constant. However, there are unitaries beyond this class mapping MPS to MPS: consider, for instance, the permutation that reflects the chain of qubits $i = 1,\dots N$ as $i \mapsto N-i+1$. This permutation preserves the MPS structure, but its operator Schmidt rank along the middle bipartition is $d^{N}$ ($N$ even), i.e., the maximal possible. This can be understood as a many-body generalization of the fact that the swap operator does not generate bipartite entanglement, but has maximal operator Schmidt rank.

Nevertheless, MPUs emerge as the \emph{unique} class preserving the 1D area law if one demands that this property be stable under extension by ancillas. That is, we now consider the case where each system qudit is augmented with an ancilla of dimension $d_{\rm aux}$. We suppose the system is initialized in an MPS over the $N$-site chain $(\mathbb C^d \otimes \mathbb C^{d_{\rm aux}})^{\otimes N}$ (with system-ancilla qubits possibly entangled) and demand that the evolution $U_N \otimes \Id_{\rm aux}$ preserves the MPS structure. This scenario is a natural one from the perspective of quantum information, reminiscent, e.g., of the complete positivity setting.

\begin{prop}\label{prop:choimpu}
Let $U_N$ be a unitary over $(\mathbb C^d )^{\otimes N}$. The following are equivalent.
\begin{enumerate}[(i)]
    \item $U_N$ is an MPU of constant bond dimension.
    \item $U_N \otimes \Id_{\rm aux}$ acting on a chain $(\mathbb C^d \otimes \mathbb C^{d_{\rm aux}})^{\otimes N}$ preserves the 1D area law, where $d_{\rm aux} \ge d$ is a fixed constant.
\end{enumerate}
\end{prop}

\begin{proof}

$(i) \Rightarrow (ii)$: Follows directly from the fact that the bond dimension of $(U_N \otimes \Id_{\rm aux} )\ket{\Phi}_N$ is at most $D\cdot D_{\rm MPU}$, where $D_{\rm MPU}$ is the bond dimension of $U_N$ and $D$ of the input MPS $\ket{\Phi}_N$.

$(ii) \Rightarrow (i)$: Consider the case $d_{\rm aux} = d$ and as input the $D = 1$ MPS
\[
|\Phi\rangle_N=\bigotimes_{i=1}^N |\Phi^+\rangle_{i i},
\]
where $|\Phi^+\rangle_{i i} = d^{-1/2} \sum_{j=1}^d \ket{jj}_{ii}$ is a maximally entangled Bell state between the $j^{\rm th}$ system qudit and its corresponding ancilla. The action $C_{U_N} = (U_N \otimes \Id_{\rm aux}) \ket{\Phi}_N$ is nothing but the Choi state of $U_N$. The result follows since for any bipartition of the chain, the Schmidt rank of $C_{U_N}$ across the cut is equal to the operator Schmidt rank of $U$. The case $d_{\rm aux} \ge d$ follows similarly.
\end{proof}

\subsection{Quantum Channels Preserving the Area Law}

We now turn to mixed states and quantum channels. In this setting, the notion of an area law is less straightforward. A natural candidate is the mutual information across bipartitions \cite{Wolf_2008},
\[
I(A : A^C)_\rho = S(\rho_A) + S(\rho_{A^C}) - S(\rho) \;.
\]
This is a quantity that captures both classical and quantum correlations \cite{Groisman_2005}.
However, for a generic MPDO, it is not known whether the mutual information is bounded solely in terms of the bond dimension. Moreover, replacing the von Neumann entropy with the R\'enyi version does not lead to a meaningful measure of correlations: the resulting quantity loses the property of being monotonic under local quantum channels, which is essential for any reasonable quantifier of correlation.

To bypass these difficulties, we directly define the area-law in terms of the local purification of an MPDO. Recall that an MPDO admits a purification in the form of an MPS with --- a priori --- a  high-enough bond dimension (see the discussion before \cref{def:mpqc_class}). 
We define the area-law as follows: an MPDO is said to satisfy an area law if there exists a purifying MPS which satisfies an area law for $S_0$, i.e., if there exists a purification that can be expressed as an MPS with a constant (i.e., $N$-independent) bond dimension.

Notice that, since the purification satisfies an entanglement area law, one also obtains a corresponding bound on the mutual information of the resulting mixed state. Indeed, if
\[
\rho=\operatorname{Tr}_E |\Psi\rangle\!\langle\Psi|
\]
with $|\Psi\rangle$ an MPS of bond dimension $D$ and the partial trace over the purifying environment $E$, then for every bipartition the resulting MPDO satisfies
\[
I(A:A^c)_\rho
\le 2 S(\rho)
\le 2 \log D,
\] where $I(A:A^c)$ is the mutual information between $A$ and its complement. Bounded purification bond dimension hence implies a mutual-information area law. In that sense, our version of the area law is a strong one. Similar bounds follow as a consequence of our area-law definition for other quantifiers of total (i.e., quantum and classical) correlations, such as the R\'enyi entanglement of purification for any $\alpha \ge 0$~\cite{Terhal_2002,Guth_Jarkovsky_2020}.

In analogy to the unitary case, one may ask which quantum channels preserve the area-law property.
The resulting objects are matrix product quantum channels in local purification form, introduced earlier as LP channels.

%\begin{prop}
%Let $\rho$ be an LPDO with bounded purification bond dimension, and let $\mathcal{E}$ be a quantum channel acting on a one-dimensional system. Then $\mathcal{E}$ preserves the area-law structure of all such states (i.e., maps finite-bond-dimension LPDOs to finite-bond-dimension LPDOs) under arbitrary finite ancilla extensions (i.e., under maps of the form $\mathcal{E} \otimes \Id_{A}$ for any finite-dimensional ancillary system $A$) if and only if $\mathcal{E}$ is an LP channel.
%\end{prop}

\begin{prop}
Let $\mathcal E_N$ be a quantum channel over $(\mathbb C^d)^{\otimes N}$. The following are equivalent.
\begin{enumerate}[(i)]
    \item $\mathcal E_N$ admits a local purification with bounded bond dimension, i.e., $\mathcal E_N \in \text{\sffamily LP}$.
    \item $\mathcal E_N \otimes \Id_{\rm aux}$ acting on a mixed state over the chain $(\mathbb C^d \otimes \mathbb C^{d_{\rm aux}})^{\otimes N}$ preserves the 1D area law, where $d_{\rm aux} \ge d$ is a fixed constant.
\end{enumerate}
\end{prop}

\begin{proof}
Both directions follow as in the unitary case, working with purifications.

$(i) \Rightarrow (ii)$: Follows directly from the fact that, if $\rho_N$ is an MPDO with purifying MPS of bond dimension $D$, then $(\mathcal E_N \otimes \Id_{\rm aux} )(\rho_N)$ admits a purifying MPS of bond dimension at most $D\cdot D_{\rm LP}$, where $D_{\rm LP}$ is the bond dimension of the purifying isometry of $\mathcal E_N$.

$(ii) \Rightarrow (i)$: Consider the case $d_{\rm aux} = d$ and as input the $D = 1$ MPDO $\rho_N = \ket{\Phi}_N \bra{\Phi}$, where
\[
|\Phi\rangle_N=\bigotimes_{i=1}^N |\Phi^+\rangle_{i i},
\]
and $|\Phi^+\rangle_{i i} = d^{-1/2} \sum_{j=1}^d \ket{jj}_{ii}$ is a maximally entangled Bell state between the $j^{\rm th}$ system qudit and its corresponding ancilla. The action $C_{\mathcal E_N} = (\mathcal E_N \otimes \Id_{\rm aux}) \ket{\Phi}_N\bra{\Phi}$ is nothing but the Choi state of $\mathcal E_N$. Hence the Choi state of $\mathcal{E}_N$ admits a local purification with finite bond dimension. We conclude that $\mathcal{E}$ is an LP channel. The case $d_{\rm aux} \ge d$ follows similarly.
\end{proof}

\section{Beyond Homogeneous Channels}
\label{sec:sMPI}

As discussed in \cref{sec:MPIstructure}, channels in {\sffamily{hLP}} do not produce long-range entanglement, i.e., any reduced state on disjoint regions will be a product state as long as the regions are well separated and the input state is short-range correlated. However, consider the following example, analyzed in the main text.
\begin{ex}
\label{longrangeisometry}
    Consider the quantum channel $\mathcal{E}(\cdot) = \Tr(\cdot) \, \ket{\mathrm{GHZ}}\bra{\mathrm{GHZ}}$ acting on $N$ qubits. For all $N$, it admits a purifying MPI with $D_{\rm MPI} = \chi = 2$
\begin{align*} 
V_N = \frac{1}{\sqrt{2}} \;
\begin{tikzpicture}[scale=0.45,baseline={([yshift=-0.65ex] current bounding box.center)}] 
        \pgfmathsetmacro{\squarehalf}{0.4}
        \pgfmathsetmacro{\halfunit}{\squarehalf+\leglength}
        \SingleTrLeft{(-4*\leglength,0)} 
        \draw [very thick] (11*\leglength,0) to (16*\leglength,0); 
        \SingleTrRight{(16*\leglength,0)} 
        \draw [very thick] (-4*\leglength,0) to (0,0); 
        \GHZTensor{0,0}{\leglength}{\squarehalf}{0} 
        \GHZTensor{2*\halfunit,0}{\leglength}{\squarehalf}{0} 
        \SingleDots{4*\halfunit, 0}{2*\leglength} 
        \GHZTensor{6*\halfunit,0}{\leglength}{\squarehalf}{0}
        % y > 0 part
        \begin{scope}
          \clip (0,0) rectangle (7,1.5);
          \draw[very thick,color=purification]
            (\halfunit,\halfunit) .. controls (\halfunit,0) and (0,-0.5*\halfunit) .. (0,-1.5*\halfunit);
          \draw[very thick,color=purification]
            (3*\halfunit,\halfunit) .. controls (3*\halfunit,0) and (2*\halfunit,-0.5*\halfunit) .. (2*\halfunit,-1.5*\halfunit);
          \draw[very thick,color=purification]
            (7*\halfunit,\halfunit) .. controls (7*\halfunit,0) and (6*\halfunit,-0.5*\halfunit) .. (6*\halfunit,-1.5*\halfunit);
        \end{scope}
        % y < 0 part
        \begin{scope}
          \clip (-0.2,-1.5) rectangle (7,0);
          \draw[very thick]
            (\halfunit,\halfunit) .. controls (\halfunit,0) and (0,-0.5*\halfunit) .. (0,-1.5*\halfunit);
          \draw[very thick]
            (3*\halfunit,\halfunit) .. controls (3*\halfunit,0) and (2*\halfunit,-0.5*\halfunit) .. (2*\halfunit,-1.5*\halfunit);
          \draw[very thick]
            (7*\halfunit,\halfunit) .. controls (7*\halfunit,0) and (6*\halfunit,-0.5*\halfunit) .. (6*\halfunit,-1.5*\halfunit);
        \end{scope}
    \end{tikzpicture}
    \;,
\end{align*}
where
$
    \begin{tikzpicture}[scale=0.6,baseline={([yshift=-0.65ex] current bounding box.center)}] 
        \pgfmathsetmacro{\squarehalf}{0.4}
        \pgfmathsetmacro{\labels}{\squarehalf+\leglength/2}
        \GHZTensor{0,0}{\leglength}{\squarehalf}{0} 
        \draw (\labels,0.2) node {\scriptsize 0};
        \draw (-\labels,0.2) node {\scriptsize 0};
        \draw (0.2,\labels) node {\scriptsize 0};
    \end{tikzpicture}
    =
    \begin{tikzpicture}[scale=0.6,baseline={([yshift=-0.65ex] current bounding box.center)}] 
        \pgfmathsetmacro{\squarehalf}{0.4}
        \pgfmathsetmacro{\labels}{\squarehalf+\leglength/2}
        \GHZTensor{0,0}{\leglength}{\squarehalf}{0}
        \draw (\labels,0.2) node {\scriptsize 1};
        \draw (-\labels,0.2) node {\scriptsize 1};
        \draw (0.2,\labels) node {\scriptsize 1};
    \end{tikzpicture}
    =1
$
    or vanishing otherwise is the local tensor of an unnormalized $\mathrm{GHZ}$ state. The resulting channel generates long-range entanglement due to the correlations present in the $\mathrm{GHZ}$ state.
\end{ex}
This example seemingly violates the no-long-range-entanglement constraint of \cref{SMdepth2}. In reality, no discrepancy arises because this example does not fit within the scope of the theorem: homogeneity is broken, as an additional constant is required to normalize the state generated by the Kronecker Delta. Specifically, it generates the \emph{unnormalized} $\mathrm{GHZ}$ state, which has norm \(2\), and hence the tensor ${\cal V}$ generates a bona fide isometry only after a \emph{system-size independent normalization}: $\frac{1}{2}V_{N}^\dagger V_{N} = \Id_d^{\otimes N}$. The corresponding $\mathcal V$ tensor, thus, does not generate a homogeneous MPI -- these do not create long-range entanglement, on the basis of \cref{SMdepth2}. The requirement of the normalization constant is precisely what separates this example from the {\sffamily{hLP}} class. This motivates us to introduce and analyze the following generalization of this class, taking into account constants.

\subsection{Structure of scaled Homogeneous Matrix Product Isometries (sMPI)}

Motivated by the last example, we examine what happens when the strict isometry requirement $V_{N}^\dagger V_{N} = \Id_d^{\otimes N}$ is relaxed to allow a normalization factor which does not depend on the system size.
\begin{defn} [Scaled Homogeneity]
    An MPO tensor $\mathcal V$ generates a \emph{scaled homogeneous MPI} (sMPI) if there exists a (system-size independent) constant $c>0$ such that $V_N^\dagger V_N = c \Id$ for all $N$. The resulting class is denoted as {\sffamily sMPI}, and {\sffamily sLP} denotes the class proportional to channels, generated after tracing out a purification space.
\end{defn}
The different operations here should not be interpreted as unphysical, but instead as requiring a ``by-hand'' normalization, not included in the tensor-network description. We will show later in \cref{multiplicityrule} that $c$ is more specifically an integer. We have shown that isometries in {\sffamily sMPI} can create long-range correlations, as opposed to {\sffamily hMPI}, and thus we turn to analyze the structure of this class, which has very different physical properties.

Notice that in our definition we could consider, for instance, $c$ to be system-size dependent, giving rise to ${\overline{\text{\sffamily{hMPI}}}}$ in our earlier notation. An example in the latter class is ${V} = {U} \bigl(\ket{{{\rm AKLT}}}\otimes\Id\bigr)$, where \(U\) is a QCA unitary and AKLT denotes the ground state of the corresponding model, which is a homogeneous MPS with exponentially decaying correlations~\cite{affleck1987rigorous}. This shows that
\begin{align}
    \text{\sffamily hMPI} \subset \text{\sffamily sMPI} \subset {\overline{\text{\sffamily{hMPI}}}}
\end{align}
This distinction in classes can also be understood in close analogy with MPS: if a normalization for all system sizes is required, the only allowed states are fixed-point states~\cite{cirac:mpdo-rgfp}, such as product states. If we allow for a normalization constant independent of the system size, we find long-range entangled states such as the $\mathrm{GHZ}$ state. Then, if we allow for a normalization constant that can depend on the system size, we get normal MPS states like the $\mathrm{AKLT}$.

To analyze the general structure of {\sffamily sMPI}, we need some further definitions. Given a tensor ${\cal V}$ generating a sMPI, we define the transfer matrix $E$ and the operators $S_\alpha$ as
\begin{align}
    \label{EandSdefs}
    E &=  \frac{1}{d_\mathrm{in}}\begin{tikzpicture}[scale=\defaultscaling,baseline={([yshift=-0.65ex] current bounding box.center)}]
    \TransferMatrixTensor{0,0}{\leglength}{0.4}{${\cal V}$}{0}
    \end{tikzpicture}, &
    S_\alpha &= \frac{1}{d_\mathrm{in}}\begin{tikzpicture}[scale=\defaultscaling,baseline={([yshift=-0.65ex] current bounding box.center)}]
    \SMatrixTensorMiddleSigma{0,0}{\leglength}{0.4}{${\cal V}$}{0}{$\sigma_\alpha^T$}{0.5}{\leglength / 2 + 0.5}
    \end{tikzpicture},
\end{align}
where $\sigma_\alpha$, for $\alpha = 1, \dots, d_{\mathrm{in}}^2 - 1$ are orthonormal Hermitian traceless operators that form a basis together with $\sigma_0 = \mathbb{I}$.
With these definitions at hand, for any tensor generating a sMPI the following decomposition holds:
\begin{equation}
\label{sitedoubledecomposition}
    \begin{tikzpicture}[scale=\defaultscaling,baseline={([yshift=-0.65ex] current bounding box.center)}]
    \SingleSiteDouble{0,0}{\leglength}{0.4}{${\cal V}$}{0}
    \end{tikzpicture} = 
    E \otimes \Id + \sum_{\alpha} S^{\alpha} \otimes \sigma_{\alpha}
\end{equation}

Moreover, we use the concept of basis of normal tensors and orthogonality. Some of the definitions refer to MPS generating tensors, i.e., without differentiating the input/output/purification spaces of the MPO tensor.

\begin{defn}[Basis of Normal Tensors \cite{cirac:mpdo-rgfp}]
    A basis of normal tensors (BNT) for $A$ is a set of normal tensors $A_j \; {(j=1, \ldots, g)}$ so that (i) for each $N$, $\ket{\Psi({A})}_{N}$ can be written as a linear combination of $\ket{\Psi({A_j})}_{N}$, and (ii) there exists some $N_0$ such that for all $N>N_0$, the $\ket{\Psi({A_j})}_{N}$ are linearly independent.
\end{defn}

\begin{defn}[Orthogonality]
Two tensors, $A_j, A_{j^{\prime}}$ are \emph{weakly orthogonal} if
\begin{equation}
    \begin{array}{c}
    \begin{tikzpicture}[scale=\defaultscaling,baseline={([yshift=-0.8ex] current bounding box.center)}]
    \pgfmathsetmacro{\labelposition}{\leglength + 0.4}
        \TransferMatrixTensor{0,0}{\leglength}{0.4}{$$}{0}
        \draw (0,\labelposition) node {\scriptsize $\bar{A}_j$};
        \draw (0,-\labelposition) node {\scriptsize  $A_{j^{\prime}}$};
    \end{tikzpicture}
\end{array}
=0 \quad \forall j \ne j'.
\end{equation}
Two tensors, $A_j, A_{j^{\prime}}$ are \emph{strongly orthogonal} if
\begin{equation}
    \begin{array}{c}
    \begin{tikzpicture}[scale=\defaultscaling,baseline={([yshift=-0.8ex] current bounding box.center)}]
    \pgfmathsetmacro{\labelposition}{\leglength + 0.4}
        \SingleSiteDouble{0,0}{\leglength}{0.4}{$$}{0}
        \draw (0,\labelposition) node {\scriptsize $\bar{A}_j$};
        \draw (0,-\labelposition) node {\scriptsize  $A_{j^{\prime}}$};
    \end{tikzpicture}
\end{array}
=0 \quad \forall j \ne j'.
\end{equation}
We will say that two sets of tensors are weakly/strongly orthogonal if they are so pairwise.
\end{defn}

We will need the following simple lemma. The proof is a straightforward application of the Vandermonde matrix determinant formula and can be found in Ref.~\cite{de_Groot_2022}.
\begin{lem}
\label{gemmaresult}
Given $K$ distinct $\lambda_k \in \mathbb C\neq 0$, if
\[
\sum_{k=1}^K c_k \lambda_k^{\,N} = 0 
\qquad \forall\, N = N_0, \ldots, K + N_0 - 1,
\]
for some $N_0 \in \mathbb{N}_{>0}$, then $c_k = 0$ for all $k = 1, \ldots, K$.
\end{lem}

We will also need the following lemma, which was proven by Nagata and Higman \cite{Nagata, Higman}, and successively strengthened by Razmyslov \cite{Razmyslov}.
\begin{lem}
\label{nagata-higman}
Let $\mathcal{A} \subset M_D(\mathbb{C})$ be the (associative) algebra generated by finitely many elements 
$\{S_\alpha\}_{\alpha\in I}$. 
Assume that every element $S \in \mathcal{A}$ is nilpotent. 
Then $\mathcal{A}$ is nilpotent of bounded index: there exists an integer $J' \le D^2$
such that for any sequence of generators $\alpha_1,\dots,\alpha_{J'}$ one has
\[
S_{\alpha_1} S_{\alpha_2} \cdots S_{\alpha_{J'}} = 0 .
\]
In particular, any product of $J'$ elements of $\mathcal{A}$ vanishes.
\end{lem}
We stress here that \cref{nagata-higman} involves an algebra, and not merely a set of nilpotent matrices. This feature is crucial for the Lemma to hold.

To analyze the structure of a tensor $\mathcal V$ generating a sMPI, we first bring, without loss of generality, $\mathcal V$ in canonical form: 
\be \label{eq:bnt_v}
\mathcal{V}^i \;=\;
\bigoplus_{j=1}^g \bigoplus_{q=1}^{m_j} \mu_{j,q}\mathcal{V}_{j}^i \,,
\ee
where $\{ \mathcal V^i_j \}_j$ is a BNT for $\mathcal V$~\cite{cirac_matrix_2017-1} and $\mu_{j,q} \ne 0$.
The transfer operator consequently breaks down as
\be 
\label{Edecomp}
E \;=\;
\bigoplus_{j,j'=1}^g \bigoplus_{q,q'=1}^{m_j,{m_{j'}}} \mu_{j,q}\overline{\mu_{j',q'}} \mathbb{E}_{jj'}
\ee
with
\(\mathbb{E}_{jj'}=\sum_{i}\mathcal{V}_j^i\otimes\overline{\mathcal{V}_{j'}^i}\).

\begin{lem}
\label{multiplicityrule}
    For any tensor $\mathcal V$ generating a sMPI with normalization constant $c$, it holds that:
    \begin{enumerate}[(i)]
        \item $\mathcal V$ can be decomposed as $\mathcal{V}^{i} \;=\;\bigoplus_{j=1}^g (\Id_{m_j} \otimes  \mathcal{V}^{i}_j)$, where \(\mathcal{V}_j\) are the $g$ elements of a BNT which are weakly orthogonal and $m_j$ is an integer denoting the block multiplicity. The index $i$ refers to the grouped physical indices.
        \item $\sum_{j=1}^g m_j^2 = c$, and thus $c$ is an integer.
        \item Blocking a finite (system-size independent) number of times, it holds that $\mathbb E_{jj'} = 0$ for all $j \ne j'$. $\mathbb E_{jj}$ has a simple unit eigenvalue, which is equal to its spectral radius, and $\mathbb E_{jj}^2 = \mathbb E_{jj}$ for all $j$.
    \end{enumerate}
\end{lem}

\begin{proof}
    We assume, without loss of generality, that $\mathcal V$ is in canonical form as described above. Imposing
    \begin{equation}
            \tr(E^N)=c
    \quad\forall N
    \end{equation}
    gives
    \begin{equation} \label{eq:mult_vandermonde}
    \sum_{j,j',q,q'} \mu_{j,q}^N\overline{\mu_{j',q'}}^N
    \;\tr\bigl[\mathbb{E}_{jj'}^N\bigr]
    \;=\;  \sum_{j,j',q,q', k} \bigl(\mu_{j,q}\overline{\mu_{j',q'}}
    \;\nu_{jj',k}\bigr)^N
    \;=\; c \cdot 1^N.
    \end{equation}
After grouping all identical values taken by the quantities 
\(\mu_{j,q}\overline{\mu_{j',q'}}\,\nu_{jj',k}\), let \(\mathcal{A}\) denote the resulting finite set and, for each \(\alpha\in\mathcal{A}\), let \(M(\alpha)\) be its multiplicity.  
The identity above then reads \(\sum_{\alpha\in\mathcal{A}} M(\alpha)\,\alpha^{N}=c\) for all \(N\) with $\alpha \in \mathbb C$.  
Separating the contribution at \(\alpha=1\) yields
\[
\sum_{\alpha\neq 1} M(\alpha)\,\alpha^{N} + (M(1)-c) \cdot 1^N   = \sum_{\alpha\in\mathcal{A}} \widetilde M(\alpha)\,\alpha^{N}=0
\qquad\forall\,N,
\]
where we have defined \(\widetilde M(\alpha)=M(\alpha)\) for \(\alpha\neq 1\) and \(\widetilde M(1)=M(1)-c\).

By \cref{gemmaresult}, such an identity can hold only if \(c \in \mathbb{N}\), since \( M(1)\) counts occurrences and is therefore a non-negative integer, and if every term with $\alpha \ne 1$ vanishes. Notice, moreover, that for any $M(1)$ occurrence, the eigenvalue $\nu_{jj',k}$ and the coefficient $\mu_{j,q}\overline{\mu_{j',q'}}$ must be 1 themselves, as otherwise pairing of the coefficient with a different eigenvalue would lead to occurrences of \( M(\alpha \ne 1)\), which is not allowed.
Now, recall that if diagonal blocks \(\mathbb{E}_{jj}\) have spectral radius one, any ``off-diagonal" blocks have spectral radius \(\rho(\mathbb{E}_{jj'})<1\) for \(j\neq j'\)~\cite{cirac_matrix_2017-1}. This means all off-diagonal blocks \(\mathbb{E}_{jj'}\) with \(j\neq j'\) are nilpotent, i.e., $\nu_{jj',k} = 0$ if $j \ne j'$. As a result, they vanish after blocking.
Hence, only diagonal blocks with unit spectral radius remain. From normality, it follows that there is a unique eigenvalue equal to the spectral radius for each of the diagonal block~\cite{cirac_matrix_2017-1} and the rest are thus vanishing, i.e., $\nu_{jj',k} = \delta_{j,j'} \delta_{k,1}$. From its Jordan normal form, it follows that blocking eliminates any Jordan blocks corresponding to the zero eigenvalues. This establishes (iii). 
Using the established fact that $\mu_{j,q} = \mu_j \coloneqq
e^{i \phi_j}$, independent of $q$, \cref{eq:mult_vandermonde} reduces to $\sum_{j} m_j^2 = c$, which shows (ii). The decomposition in (i) follows by absorbing the phase $e^{i \phi_j}$ into the BNT $\mathcal V_j^i$.
\end{proof}

We now strengthen this result to show that the elements of the BNT, into which any sMPI can be broken down, are also \emph{strongly orthogonal} and \emph{hMPI-generating}. Below, we write $\mathcal V^{ik}$ separating the grouped physical output and input indices, respectively; this distinction is relevant because strong orthogonality separates them.

\begin{prop}
\label{sMPIstructure}
    After blocking a finite number of times, any sMPI can be expressed as
    \[
    \mathcal{V}^{ik} \;=\;\bigoplus_{j=1}^g (\Id_{m_j} \otimes \mathcal{V}^{ik}_j),
    \]
where \(\mathcal{V}_j\) are elements of a BNT of strongly orthogonal and hMPI-generating tensors.
\end{prop}

On the physical level, the Proposition shows that any sMPI can be expressed as $V = \frac{1}{\sqrt{c}}\sum_{j=0}^{g-1} m_j V_j$. Each $V_j$ is a hMPI, and $V^\dagger_j V_k = \Id \delta_{jk}$ -- in fact, by strong local orthogonality, this holds locally as well.

\begin{proof}
    We first show that the elements of the BNT are strongly orthogonal tensors, and then that they are hMPI-generating.
    We first block such that part (iii) of \cref{multiplicityrule} holds. The starting point is the isometry condition: $V_{N}^\dagger V_{N}^{} = c\,\Id_d^{\otimes N}$ for all $N$. From it, we can infer that 
    \be
    \label{trace with traceless}
    \begin{array}{c}
        \begin{tikzpicture}[scale=\defaultscaling,baseline={([yshift=-0.75ex] current bounding box.center)}]
            \foreach \x in {1,...,1}{
                \TransferMatrixTensor{(-\doubledx*\x,0),0}{\leglength}{0.4}{${\cal V}$}{0}
                \SingleTrRight{(0,.9)}
                \begin{scope}[yscale=-1]{
                    \SingleTrRight{(0,.9)};
                }
                \end{scope}
            }
            \foreach \x in {2,...,2}{
                \SMatrixTensorMiddleSigma{(-\doubledx*\x,0)}{\leglength}{0.4}{${\cal V}$}{0}{$\sigma_\alpha^T$}{0.5}{\leglength / 2 + 0.5}
            }
            \foreach \x in {3,...,3}{
                \DoubleDots{-\doubledx*\x,0}{\doubledx/2}{.9};
            }
            \foreach \x in {4,...,4}{
                \TransferMatrixTensor{(-\doubledx*\x,0),0}{\leglength}{0.4}{${\cal V}$}{0}
            }
            \foreach \x in {5,...,5}{
            \pgfmathsetmacro{\xp}{\x+0.5}
            \draw[very thick, color=black](-\doubledx*\x + 0.7,0.9) -- (-\doubledx*\xp,0.9);
            \begin{scope}[yscale=-1]{
                    \draw[very thick, color=black](-\doubledx*\x + 0.7,0.9) -- (-\doubledx*\xp,0.9);
                }
                \end{scope}
            }
            \foreach \x in {5.5,...,5.5}{
                \SingleTrLeft{(-\doubledx*\x,.9)};
                \begin{scope}[yscale=-1]{
                    \SingleTrLeft{(-\doubledx*\x,.9)};
                }
                \end{scope}
            }
        \end{tikzpicture}
    \end{array}
    = 0 ,
\ee
i.e., one gets zero any time one or more of the physical legs are contracted with a traceless matrix $\sigma_\alpha$. It follows in particular that the set $\{S^{\alpha} \}_{\alpha}$ (cf.~\cref{EandSdefs}) generates a nilpotent algebra.
In fact, all of the blocks generate a nilpotent algebra independently, due to the direct sum structure,
\be 
    \label{Sdecomp}
S^\alpha \;=\;
\bigoplus_{j,j'=1}^g \bigoplus_{q,q'=1}^{m_j,{m_{j'}}}  \mathbb{S}^{\alpha}_{jj'},
\ee
where $\mathbb{S}^{\alpha}_{jj'}$ is obtained by contracting $\mathcal{V}_j$ with the conjugate of $\mathcal{V}_{j'}$ in the structure of \cref{EandSdefs}. As a consequence, by \cref{nagata-higman}, there exists a finite, $N$-independent blocking length $L$ such that
\begin{align} \label{eq_nagata}
\mathbb S_{jj'}^{\alpha_1}\cdots \mathbb S^{\alpha_L}_{jj'} = 0\quad \forall\, \alpha,  j, j' \,.
\end{align}
Now, denote $\mathbb {S'}^{\alpha}_{jj'}$ the resulting matrices after contracting with a basis of traceless elements in the physical space, but after blocking $L$ times. Note that these might consist not only of products of $\mathbb S_{jj'}^{\alpha}$ but also $\mathbb E_{jj'}$ (at most $L-1$ times). If $j \ne j'$, the $\mathbb {S'}^{\alpha}_{jj'}$ matrices vanish, since they can only have the form $ \mathbb S_{jj'}^{\alpha_1}\cdots \mathbb S^{\alpha_L}_{jj'}$, as if any $\mathbb E_{jj'}$ appeared in the string, it would make the product zero, because $\mathbb{E}_{jj'}=0$, which follows from \cref{multiplicityrule}. This means that the blocks $\mathcal V^{ik}_{j}$ become strongly orthogonal after blocking.

For the hMPI generating property, we only have to consider the case where $j=j'$. We now show that, for all blocks $j$, all lengths $m \ge 1$, and all indices $\alpha_m$,
\begin{align} \label{eq:tr_ese}
\tr[\mathbb E_{jj} \mathbb S_{jj}^{\alpha_1}\cdots \mathbb S_{jj}^{\alpha_m} ] = \mathbb E_{jj} \mathbb S_{jj}^{\alpha_1}\cdots \mathbb S_{jj}^{\alpha_m} \mathbb E_{jj} = 0 \;.    
\end{align}
Indeed, choose $N=z\ell+2$, where $z,\ell \in \mathbb{N}$. Then,
\begin{align}
    0 = \tr[(E S^{\alpha_1}\cdots S^{\alpha_\ell} E)^z] = \sum_i \tr[(\mathbb E_{jj} \mathbb S_{jj}^{\alpha_1} \cdots \mathbb S_{jj}^{\alpha_\ell} \mathbb E_{jj})^z] \overset{\text{\tiny (1)}}{=} \sum_i [\tr(\mathbb E_{jj} \mathbb S_{jj}^{\alpha_1}\cdots \mathbb S_{jj}^{\alpha_\ell} \mathbb E_{jj})]^z
\end{align}
In (1) we make use of a consequence of \cref{multiplicityrule}: $\mathbb{E}_{jj}$ have rank one, while $\mathbb E_{jj'} = 0$ for $j \ne j'$. Then \cref{eq:tr_ese} follows from \cref{gemmaresult}. Combining \cref{eq_nagata} and \cref{eq:tr_ese}, it follows that any trace of the form $ \tr [ \mathbb X_1 \cdots \mathbb X_L ]$, $\mathbb X_i \in \{ \mathbb E_{jj}\} \cup \{ \mathbb S_{jj}^\alpha \}_\alpha$ vanishes unless $\mathbb X_i = \mathbb E_{jj}$ for all $i = 1,\dots, L$. As a result, we conclude that each block $\mathcal V_j^{ik}$ generates an MPI, up to a constant, which can be fixed to one as in \cref{multiplicityrule}.
\end{proof}

\section{Physical Implementation of scaled Homogeneous Matrix Product Isometries (\lowercase{s}MPI)}

We now address how to physically --- and efficiently --- implement the isometries in the $\text{\sffamily sMPI}$ class. The tensor networks representing sMPIs give a useful formal description. However, their individual tensors do not directly relate to physical operations. The issue we tackle here is how to turn these abstract tensor-network isometries into real physical processes. Specifically, we want \textit{finite-depth} quantum circuits made up of \textit{geometrically local} unitary gates. These may also include local measurements, whose results influence the next gates through \textit{feedforward}. The local correcting unitaries may depend on all measurement outcomes, and thus be based on nonlocal classical communication. Finite-depth means that the depth is independent of the system size $N$, while geometrically local refers to the gates acting on connected and constant-sized sets of qudits. When measurements are permitted, calling a circuit finite-depth means that all required measurements can be arranged into a finite number of parallel rounds.
By \textit{efficiently}, we mean that the number of gates, as well as the number of ancillas and repetitions of the protocol, must be finite. One might think that, due to the long-range correlations created by the sMPIs, it should be impossible to create them with finite-depth circuits. For example, generating the GHZ state with only geometrically local gates would require $\Theta(N)$ operations. Here is exactly where measurements and corrections turn out useful: allowing for them yields a protocol for generating GHZ in $O(1)$ \cite{Piroli_2024, Piroli_2021}. We aim to reproduce such an improvement in our setting.

To achieve this, we present two distinct approaches. The first approach uses local measurements and feedforward. It provides a protocol for implementing any sMPI in constant depth, despite the fact that the isometry, and thus the corresponding channel, can create long-range correlations. The second approach avoids measurements and, although more complex, it can be used to reformulate the sMPI as an MPU with a boundary that acts on an input state in the GHZ form.

\begin{thm}
\label{SMapproach1}
        Any sMPI can be deterministically implemented by a constant-depth circuit of local unitary gates, assisted by two rounds of local measurements and local unitary corrections, as well as $O(N)$ local ancilla qubits.
\end{thm}

\begin{proof}
\cref{sMPIstructure} shows that any sMPI, on the physical level, can be written as $V = \frac{1}{\sqrt{c}}\sum_{j=0}^{g-1} m_j V_j$.
From \cref{SMdepth2}, we know that each isometry $V_j$ can be implemented as a depth-2 circuit of isometric gates.
The implementation strategy for $V$ is as follows. By adding $N$ ancillas, which we use as controls, we construct a locally-controlled circuit. The control circuit applies the local gates implementing $V_j$, whenever the control register is in state $\ket{j}$. To globally ``synchronize'' the control ancillas, we feed a suitable GHZ state to them, so that the resulting global operation is a superposition of $V_j$ for all $j$. The preparation of the GHZ state only requires a single round of measurements and feedforward~\cite{Piroli_2024}. The added control ancillas are, however, now entangled with the system, and must be decoupled before being discarded. This will be accomplished through an additional round of measurements followed by local corrections, which are always (deterministically) possible due to the strong orthogonality property of the tensors.

Special care must be taken to implement the controlled circuit. Indeed, in the setting of \cref{SMdepth2}, the inner contracted dimensions $r_j$ and $\ell_j$ need not be equal (e.g. for $\chi= 1$, $r_j=d^2$ and $\ell_j=1$ represent the right shift, a valid unitary) and hence we must take precaution when defining a circuit that selects different bottom ($B_j$) and top ($T_j$) isometric gates for each site, as they might have different $r$ and $\ell$ dimensions.
The issue is solved by extending the bottom and top isometric gates to unitaries $\widetilde{B}_j, \widetilde{T}_j$ on $\mathbb{C}^{d^4\chi^4}$ and then fixing part of the input and discarding part of the output to recover the correct isometry.
%\definecolor{purification}{RGB}{124, 138, 128}
%
In pictures:

\pgfmathsetmacro{\spacing}{3.5cm}
\be \label{eq:BandBTilde}
    \begin{tikzpicture}[scale=\defaultscaling,x=\spacing, baseline=-3.3ex]
    % --- first top diagram left ---
    \definecolor{cream}{RGB}{255,255,221}
    \pgfmathsetmacro{\height}{1.3}
    \pgfmathsetmacro{\shift}{0.7}
    \pgfmathsetmacro{\labelpositions}{\height / 2 + \shift / 2 + 0.2}
    \pgfmathsetmacro{\labelcloseness}{0.1}
    \pgfmathsetmacro{\extra}{.15}
    \foreach \i in {1} {
        \draw[color=black,very thick] (\i,-\shift) -- (\i,-3*\shift);
        \draw (\i - \labelcloseness, -2*\shift) node {\scriptsize $d$};
    }
    \foreach \i in {2} {
        \draw[color=black,very thick] (\i,-\shift) -- (\i,-3*\shift);
        \draw (\i + \labelcloseness, -2*\shift) node {\scriptsize $d$};
    }
    \foreach \i in {2} {
        \draw[color=black,dashed, very thick] (\i,-\shift) -- (\i,0.5*\shift);
        \draw (\i + \labelcloseness, 0) node {\scriptsize $r_i$};
    }
    \foreach \i in {1} {
        \draw[color=black,ultra thick] (\i,-\shift) -- (\i,0.5*\shift);
        \draw (\i - \labelcloseness, 0) node {\scriptsize $\ell_i$};
    }
    \foreach \i in {1} {
        \draw[thick, fill=cream, rounded corners=2pt]
            (\i - \extra, -0.3 - \shift) rectangle (\i+1 + \extra, 0.3 - \shift);
        \draw (\i+0.5,-\shift) node {\scriptsize $B_i$};
    }
\end{tikzpicture}
\to
\begin{tikzpicture}[scale=\defaultscaling,x=\spacing, baseline=-3.3ex]
    % --- first top diagram right ---
    \definecolor{cream}{RGB}{255,255,221}
    \pgfmathsetmacro{\height}{1.3}
    \pgfmathsetmacro{\shift}{0.7}
    \pgfmathsetmacro{\labelpositions}{\height / 2 + \shift / 2 + 0.2}
    \pgfmathsetmacro{\labelcloseness}{0.1}
    \pgfmathsetmacro{\extra}{.15}
    \foreach \i in {1} {
        \draw[color=black,very thick] (\i,-\shift) -- (\i,-3*\shift);
        \draw (\i - \labelcloseness, -2*\shift) node {\scriptsize $d$};
    }
    \foreach \i in {2} {
        \draw[color=black,very thick] (\i,-\shift) -- (\i,-3*\shift);
        \draw (\i + \labelcloseness, -2*\shift) node {\scriptsize $d$};
    }
    \foreach \i in {2} {
        \draw[color=black, very thick] (\i,-\shift) -- (\i,0.5*\shift);
        \draw (\i + \labelcloseness + 0.1, 0) node {\scriptsize $(d \chi)^2$};
    }
    \foreach \i in {1} {
        \draw[color=black,very thick] (\i,-\shift) -- (\i,0.5*\shift);
        \draw (\i - \labelcloseness - 0.1, 0) node {\scriptsize $(d \chi)^2$};
    }
    \foreach \i in {1} {
        \draw[color=black,very thick] (\i + 0.1,-\shift) -- (\i+ 0.1,-3*\shift);
        \draw (\i + 0.15 + \labelcloseness, -2*\shift) node {\scriptsize $d\chi^2$};
    }
    \foreach \i in {2} {
        \draw[color=black,very thick] (\i - 0.1,-\shift) -- (\i - 0.1,-3*\shift);
        \draw (\i - 0.15 - \labelcloseness, -2*\shift) node {\scriptsize $d\chi^2$};
    }
    \foreach \i in {1} {
        \draw[thick, fill=cream, rounded corners=2pt]
            (\i - \extra, -0.3 - \shift) rectangle (\i+1 + \extra, 0.3 - \shift);
        \draw (\i+0.5,-\shift) node {\scriptsize $\widetilde{B}_i$};
    }
\end{tikzpicture}
\text{such that}
\begin{tikzpicture}[scale=\defaultscaling,x=\spacing, baseline=-3.3ex]
    % --- first top diagram right ---
    \definecolor{cream}{RGB}{255,255,221}
    \pgfmathsetmacro{\height}{1.3}
    \pgfmathsetmacro{\shift}{0.7}
    \pgfmathsetmacro{\labelpositions}{\height / 2 + \shift / 2 + 0.2}
    \pgfmathsetmacro{\labelcloseness}{0.1}
    \pgfmathsetmacro{\extra}{.15}
    \foreach \i in {1} {
        \draw[color=black,very thick] (\i,-\shift) -- (\i,-3*\shift);
        \draw (\i - \labelcloseness, -2*\shift) node {\scriptsize $d$};
    }
    \foreach \i in {2} {
        \draw[color=black,very thick] (\i,-\shift) -- (\i,-3*\shift);
        \draw (\i + \labelcloseness, -2*\shift) node {\scriptsize $d$};
    }
    \foreach \i in {2} {
        \draw[color=black, very thick] (\i,-\shift) -- (\i,0.5*\shift);
        \draw (\i + \labelcloseness + 0.1, 0) node {\scriptsize $(d \chi)^2$};
    }
    \foreach \i in {1} {
        \draw[color=black,very thick] (\i,-\shift) -- (\i,0.5*\shift);
        \draw (\i - \labelcloseness - 0.1, 0) node {\scriptsize $(d \chi)^2$};
    }
    \foreach \i in {1} {
        \draw[color=black,very thick] (\i + 0.1,-\shift) -- (\i+ 0.1,-3*\shift);
        \draw (\i + 0.15 + \labelcloseness, -2*\shift) node {\scriptsize $d\chi^2$};
        \draw (\i + 0.1,-3*\shift - 0.25) node[rotate=-90] {\scriptsize $\vb{\ket{0}}$};
    }
    \foreach \i in {2} {
        \draw[color=black,very thick] (\i - 0.1,-\shift) -- (\i - 0.1,-3*\shift);
        \draw (\i - 0.15 - \labelcloseness, -2*\shift) node {\scriptsize $d\chi^2$};
        \draw (\i - 0.1,-3*\shift - 0.25) node[rotate=-90] {\scriptsize $\vb{\ket{0}}$};
    }
    \foreach \i in {1} {
        \draw[thick, fill=cream, rounded corners=2pt]
            (\i - \extra, -0.3 - \shift) rectangle (\i+1 + \extra, 0.3 - \shift);
        \draw (\i+0.5,-\shift) node {\scriptsize $\widetilde{B}_i$};
    }
\end{tikzpicture}
=
\begin{tikzpicture}[scale=\defaultscaling,x=\spacing, baseline=-3.3ex]
    % --- first top diagram right ---
    \definecolor{cream}{RGB}{255,255,221}
    \pgfmathsetmacro{\height}{1.3}
    \pgfmathsetmacro{\shift}{0.7}
    \pgfmathsetmacro{\labelpositions}{\height / 2 + \shift / 2 + 0.2}
    \pgfmathsetmacro{\labelcloseness}{0.1}
    \pgfmathsetmacro{\extra}{.15}
    \foreach \i in {1} {
        \draw[color=black,very thick] (\i,-\shift) -- (\i,-3*\shift);
        \draw (\i - \labelcloseness, -2*\shift) node {\scriptsize $d$};
    }
    \foreach \i in {2} {
        \draw[color=black,very thick] (\i,-\shift) -- (\i,-3*\shift);
        \draw (\i + \labelcloseness, -2*\shift) node {\scriptsize $d$};
    }
    \foreach \i in {2} {
        \draw[color=black,dashed, very thick] (\i,-\shift) -- (\i,0.5*\shift);
        \draw (\i + \labelcloseness, 0) node {\scriptsize $r_i$};
    }
    \foreach \i in {1} {
        \draw[color=black,ultra thick] (\i,-\shift) -- (\i,0.5*\shift);
        \draw (\i - \labelcloseness, 0) node {\scriptsize $\ell_i$};
    }
    \foreach \i in {1} {
        \draw[color=black,very thick] (\i + 0.2,-\shift + 0.93) -- (\i+ 0.2,\height);
        \draw (\i + 0.2,-\shift + 0.7) node[rotate=-90] {\scriptsize $\vb{\ket{0}}$};
        \draw (\i + 0.27 + \labelcloseness, \height - 0.2) node {\scriptsize $\frac{(d\chi)^2}{\ell_i}$};
    }
    \foreach \i in {2} {
        \draw[color=black,very thick] (\i - 0.2,-\shift + 0.93) -- (\i - 0.2,\height);
        \draw (\i - 0.2,-\shift + 0.7) node[rotate=-90] {\scriptsize $\vb{\ket{0}}$};
        \draw (\i - \labelcloseness +0.1, \height - 0.2) node {\scriptsize $\frac{(d\chi)^2}{r_i}$};
    }
    \foreach \i in {1} {
        \draw[thick, fill=cream, rounded corners=2pt]
            (\i - \extra, -0.3 - \shift) rectangle (\i+1 + \extra, 0.3 - \shift);
        \draw (\i+0.5,-\shift) node {\scriptsize $B_i$};
    }
\end{tikzpicture},
\ee

\pgfmathsetmacro{\spacing}{3.5cm} % define spacing once
\be
\label{unitextension}
\begin{tikzpicture}[scale=\defaultscaling,x=\spacing, baseline=2.3ex]
    \definecolor{cream}{RGB}{255,255,221}
    \pgfmathsetmacro{\height}{2.2}
    \pgfmathsetmacro{\shift}{0.7}
    \pgfmathsetmacro{\labelpositions}{\height / 2 + \shift / 2 + 0.2}
    \pgfmathsetmacro{\labelcloseness}{0.1}
    \pgfmathsetmacro{\extra}{.15}
    \pgfmathsetmacro{\extrachi}{.05}
    \foreach \i in {1} {
        \draw[color=black,very thick] (\i,\shift) -- (\i,\height);
    }
    \foreach \i in {0} {
        \draw[color=purification,very thick] (\i,\shift) -- (\i,\height);
    }
    \foreach \i in {0} {
        \draw[color=black,dashed, very thick] (\i,-\shift) -- (\i,\shift);
        \draw[color=black,very thick] (\i + \extrachi,\shift) -- (\i + \extrachi,\height);
        \draw (\i - \labelcloseness, 0) node {\scriptsize $r_i$};
        \draw (\i - \labelcloseness, \labelpositions) node {\scriptsize $\chi$};
        \draw (\i + \labelcloseness + \extrachi, +\labelpositions) node {\scriptsize $d$};
    }
    \foreach \i in {1} {
        \draw[color=black,ultra thick] (\i,-\shift) -- (\i,\shift);
        \draw[color=purification,very thick] (\i - \extrachi,\shift) -- (\i - \extrachi,\height);
        \draw (\i + \labelcloseness, 0) node {\scriptsize $\ell_i$};
        \draw (\i + \labelcloseness, +\labelpositions) node {\scriptsize $d$};
        \draw (\i - \labelcloseness - \extrachi, +\labelpositions) node {\scriptsize $\chi$};
    }
    \foreach \i in {0} {
        \draw[thick, fill=cream, rounded corners=2pt]
            (\i - \extra, -0.3 + \shift) rectangle (\i+1+\extra, 0.3 + \shift);
        \draw (\i+0.5,\shift) node {\scriptsize $T_i$};
    }
\end{tikzpicture}
\to
\begin{tikzpicture}[scale=\defaultscaling,x=\spacing, baseline=2.3ex]
    \definecolor{cream}{RGB}{255,255,221}
    \pgfmathsetmacro{\height}{2.2}
    \pgfmathsetmacro{\shift}{0.7}
    \pgfmathsetmacro{\labelpositions}{\height / 2 + \shift / 2 + 0.2}
    \pgfmathsetmacro{\labelcloseness}{0.1}
    \pgfmathsetmacro{\extra}{.15}
    \pgfmathsetmacro{\extrachi}{.05}
    \foreach \i in {0} {
        \draw[color=purification,very thick] (\i,\shift) -- (\i,\height);
    }
    \foreach \i in {1} {
        \draw[color=black,very thick] (\i,\shift) -- (\i,\height);
    }
    \foreach \i in {0} {
        \draw[color=black,very thick] (\i + 0.1,-\shift) -- (\i+ 0.1,\shift);
        \draw (\i + 0.17 + \labelcloseness, -\shift + 0.5) node {\scriptsize $\frac{(d\chi)^2}{r_i}$};
    }
    \foreach \i in {1} {
        \draw[color=black,very thick] (\i - 0.1,-\shift) -- (\i - 0.1,\shift);
        \draw (\i - \labelcloseness - 0.17, -\shift + 0.5) node {\scriptsize $\frac{(d\chi)^2}{\ell_i}$};
    }
    \foreach \i in {0} {
        \draw[color=black,dashed, very thick] (\i,-\shift) -- (\i,\shift);
        \draw[color=black,very thick] (\i + \extrachi,\shift) -- (\i + \extrachi,\height);
        \draw (\i - \labelcloseness, 0) node {\scriptsize $r_i$};
        \draw (\i - \labelcloseness, \labelpositions) node {\scriptsize $\chi^2$};
        \draw (\i + \labelcloseness + \extrachi, +\labelpositions) node {\scriptsize $d^2$};
    }
    \foreach \i in {1} {
        \draw[color=black,ultra thick] (\i,-\shift) -- (\i,\shift);
        \draw[color=purification,very thick] (\i - \extrachi,\shift) -- (\i - \extrachi,\height);
        \draw (\i + \labelcloseness, 0) node {\scriptsize $\ell_i$};
        \draw (\i + \labelcloseness, +\labelpositions) node {\scriptsize $d^2$};
        \draw (\i - \labelcloseness - \extrachi, +\labelpositions) node {\scriptsize $\chi^2$};
    }
    \foreach \i in {0} {
        \draw[thick, fill=cream, rounded corners=2pt]
            (\i - \extra, -0.3 + \shift) rectangle (\i+1+\extra, 0.3 + \shift);
        \draw (\i+0.5,\shift) node {\scriptsize $\widetilde{T}_i$};
    }
\end{tikzpicture}
\quad
\text{such that}
\quad
\begin{tikzpicture}[scale=\defaultscaling,x=\spacing, baseline=2.3ex]
    \definecolor{cream}{RGB}{255,255,221}
    \pgfmathsetmacro{\height}{2.2}
    \pgfmathsetmacro{\shift}{0.7}
    \pgfmathsetmacro{\labelpositions}{\height / 2 + \shift / 2 + 0.2}
    \pgfmathsetmacro{\labelcloseness}{0.1}
    \pgfmathsetmacro{\extra}{.15}
    \pgfmathsetmacro{\extrachi}{.05}
    \foreach \i in {0} {
        \draw[color=purification,very thick] (\i,\shift) -- (\i,\height);
    }
    \foreach \i in {1} {
        \draw[color=black,very thick] (\i,\shift) -- (\i,\height);
    }
    \foreach \i in {0} {
        \draw[color=black,very thick] (\i + 0.1,-\shift) -- (\i+ 0.1,\shift);
        \draw (\i + 0.17 + \labelcloseness, -\shift + 0.5) node {\scriptsize $\frac{(d\chi)^2}{r_i}$};
        \draw (\i + 0.1,-\shift - 0.25) node[rotate=-90] {\scriptsize $\vb{\ket{0}}$};
    }
    \foreach \i in {1} {
        \draw[color=black,very thick] (\i - 0.1,-\shift) -- (\i - 0.1,\shift);
        \draw (\i - \labelcloseness - 0.17, -\shift + 0.5) node {\scriptsize $\frac{(d\chi)^2}{\ell_i}$};
        \draw (\i - 0.1,-\shift - 0.25) node[rotate=-90] {\scriptsize $\vb{\ket{0}}$};
    }
    \foreach \i in {0} {
        \draw[color=black,dashed, very thick] (\i,-\shift) -- (\i,\shift);
        \draw[color=black,very thick] (\i + \extrachi,\shift) -- (\i + \extrachi,\height);
        \draw (\i - \labelcloseness, 0) node {\scriptsize $r_i$};
        \draw (\i - \labelcloseness, \labelpositions) node {\scriptsize $\chi^2$};
        \draw (\i + \labelcloseness + \extrachi, +\labelpositions) node {\scriptsize $d^2$};
    }
    \foreach \i in {1} {
        \draw[color=black,ultra thick] (\i,-\shift) -- (\i,\shift);
        \draw[color=purification,very thick] (\i - \extrachi,\shift) -- (\i - \extrachi,\height);
        \draw (\i + \labelcloseness, 0) node {\scriptsize $\ell_i$};
        \draw (\i + \labelcloseness, +\labelpositions) node {\scriptsize $d^2$};
        \draw (\i - \labelcloseness - \extrachi, +\labelpositions) node {\scriptsize $\chi^2$};
    }
    \foreach \i in {0} {
        \draw[thick, fill=cream, rounded corners=2pt]
            (\i - \extra, -0.3 + \shift) rectangle (\i+1+\extra, 0.3 + \shift);
        \draw (\i+0.5,\shift) node {\scriptsize $\widetilde{T}_i$};
    }
\end{tikzpicture}
=
\begin{tikzpicture}[scale=\defaultscaling,x=\spacing, baseline=2.3ex]
    \definecolor{cream}{RGB}{255,255,221}
    \pgfmathsetmacro{\height}{2.2}
    \pgfmathsetmacro{\shift}{0.7}
    \pgfmathsetmacro{\labelpositions}{\height / 2 + \shift / 2 + 0.2}
    \pgfmathsetmacro{\labelcloseness}{0.08}
    \pgfmathsetmacro{\extra}{.15}
    \pgfmathsetmacro{\extrachi}{.05}
    \foreach \i in {0} {
        \draw[color=purification,very thick] (\i,\shift) -- (\i,\height);
    }
    \foreach \i in {1} {
        \draw[color=black,very thick] (\i,\shift) -- (\i,\height);
    }
    \foreach \i in {0} {
        \draw[color=black,dashed, very thick] (\i,-\shift) -- (\i,\shift);
        \draw[color=black,very thick] (\i + \extrachi,\shift) -- (\i + \extrachi,\height);
        \draw (\i - \labelcloseness, 0) node {\scriptsize $r_i$};
        \draw (\i - \labelcloseness, \labelpositions) node {\scriptsize $\chi$};
        \draw (\i + \labelcloseness + \extrachi, +\labelpositions) node {\scriptsize $d$};
    }
    \foreach \i in {1} {
        \draw[color=black,ultra thick] (\i,-\shift) -- (\i,\shift);
        \draw[color=purification,very thick] (\i - \extrachi,\shift) -- (\i - \extrachi,\height);
        \draw (\i + \labelcloseness, 0) node {\scriptsize $\ell_i$};
        \draw (\i + \labelcloseness, +\labelpositions) node {\scriptsize $d$};
        \draw (\i - \labelcloseness - \extrachi, +\labelpositions) node {\scriptsize $\chi$};
    }
    \foreach \i in {0} {
        \draw[thick, fill=cream, rounded corners=2pt]
            (\i - \extra, -0.3 + \shift) rectangle (\i+1+\extra, 0.3 + \shift);
        \draw (\i+0.5,\shift) node {\scriptsize $T_i$};
    }
    \draw[color=purification,very thick] (0.3 - \extrachi / 2,\shift + 0.92) -- (0.3 - \extrachi / 2 ,1.5*\height);
    \draw[color=purification,very thick] (0.3 + \extrachi / 2,\shift + 0.92) -- (0.3 + \extrachi / 2,1.5*\height);
    \draw (0.3 + \extrachi / 2 + \labelcloseness, 1.2*\height) node {\scriptsize $\chi$};
    \draw (0.3 + \extrachi / 2 - \labelcloseness - \extrachi, 1.2*\height) node {\scriptsize $d$};
    \draw[color=purification,very thick] (0.7 - \extrachi / 2,\shift + 0.92) -- (0.7 - \extrachi / 2 ,1.5*\height);
    \draw[color=purification,very thick] (0.7 + \extrachi / 2,\shift + 0.92) -- (0.7 + \extrachi / 2,1.5*\height);
    \draw (0.7 + \extrachi / 2 + \labelcloseness, 1.2*\height) node {\scriptsize $\chi$};
    \draw (0.7 + \extrachi / 2 - \labelcloseness - \extrachi, 1.2*\height) node {\scriptsize $d$};
    \draw (0.3,\shift + 0.7) node[rotate=-90] {\scriptsize $\vb{\ket{0}}$};
    \draw (0.7,\shift + 0.7) node[rotate=-90] {\scriptsize $\vb{\ket{0}}$};
\end{tikzpicture}.
\ee
In extending the isometries, we assumed that $\frac{(d\chi)^2}{\ell_i}$ and $\frac{(d\chi)^2}{r_i}$ are integers for all $i$. If this condition does not hold, one can instead enlarge the dimensions to the least common multiple of these quantities, and the construction proceeds in exactly the same way.

As announced, we use $N$ additional ancillas, each of local dimension $g$, denoted as $A$. We define the control version of the unitaries $B$ and $T$ acting over $\mathbb{C}^{d^4\chi^4g}$ such that if the control qudit is in state $\ket{j}$, the resulting unitary will act as $\widetilde{B}_j$ and $\widetilde{T}_j$ respectively.

We define the global controlled unitary based on this extension as:
\begin{equation}
\label{controlunitary}
    CU = 
\pgfmathsetmacro{\spacing}{1.5cm}
\begin{tikzpicture}[scale=0.7,x=\spacing, baseline={([yshift=-0.65ex] current bounding box.center)}]
    \definecolor{cream}{RGB}{255,255,221}
    \pgfmathsetmacro{\height}{1.5}
    \pgfmathsetmacro{\shift}{0.7}
    \pgfmathsetmacro{\labelpositions}{\height / 2 + \shift / 2 + 0.2}
    \pgfmathsetmacro{\labelcloseness}{0.15}
    \pgfmathsetmacro{\extra}{.15}

    \clip (0.9,-1.5) rectangle (8.1,1.5);
    
    \foreach \i in {0,1,3,4,6,7,9} {
        \draw[color=black,very thick] (\i,\shift) -- (\i,\height);
    }
    \foreach \i in {0,3,6,9} {
        \draw[color=purification,very thick] (\i,\shift) -- (\i,\height);
    }
    \foreach \i in {1.5, 4.5, 7.5} {
        \draw[color=black,very thick] (\i,\height) -- (\i,-\height);
        \fill (\i,\shift) circle (3pt);
        \fill (\i,-\shift) circle (3pt);
        
    }
    \foreach \i in {0,2,3,5,6,8,9} {
        \draw[color=black,very thick] (\i,-\shift) -- (\i,-\height);
        %\draw (\i - \labelcloseness, -\labelpositions) node {\scriptsize $d$};
    }
    \foreach \i in {0,3,6,9} {
        \draw[color=black,very thick] (\i,-\shift) -- (\i,\shift);
        \draw[color=black,very thick] (\i + \extra,\shift) -- (\i + \extra,\height);
    }
    \foreach \i in {1,4,7} {
        %\draw[color=black,very thick] (\i,-\shift) -- (\i,\shift);
        \draw[color=purification,very thick] (\i - \extra,\shift) -- (\i - \extra,\height);
    }
    \foreach \i in {1,4,7} {
            \draw[color=black,very thick, rounded corners=8pt]
    (\i, \shift) -- (\i, 0) -- (\i+1, 0) -- (\i+1, -\shift);
        }
    \foreach \i in {0,3,6} {
        \draw[color=black,very thick] (\i,\shift) -- (\i+1.5,\shift);
        \draw[thick, fill=cream, rounded corners=2pt]
            (\i - \extra, -0.3 + \shift) rectangle (\i+1+\extra, 0.3 + \shift);
        \draw (\i+0.5,\shift) node {\scriptsize $T$};
    }
    \foreach \i in {-1,2,5,8} {
        \draw[color=black,very thick] (\i,-\shift) -- (\i+2.5,-\shift);
        \draw[thick, fill=cream, rounded corners=2pt]
            (\i - \extra, -0.3 - \shift) rectangle (\i+1 + \extra, 0.3 - \shift);
            \draw (\i+0.5,-\shift) node {\scriptsize $B$};     
    }
\end{tikzpicture}
\end{equation}

The protocol begins by preparing a generalized GHZ state of the form $\ket{\mathrm{GHZ}} = \frac{1}{\sqrt{c}} \sum_j m_j \ket{j}_A^{\otimes N}$ on the ancilla space $A$. This can be done in constant depth with measurements, without additional ancillas~\cite{Piroli_2024}. We then feed this state to the control input of $CU$, and also use the input $\ket{0}$ to the two middle legs in \cref{eq:BandBTilde}), for all $\widetilde{B}_i$ gates, such that they act as the isometries $B_i$. Overall, this implements to an arbitrary input state the action of the isometry $$V' = \frac{1}{\sqrt{c}} \sum_j m_j V_j \otimes \ket{j}_A^{\otimes N} \,.$$

Subsequently, we need to decouple the $N$ ancillas, since they are entangled and cannot yet be discarded. For this, we measure each of the ancilla qudits individually in the Fourier basis {$\{\ket{F(x)} \bra{F(x)}\}_{x}$}, defined by
$${\ket{x}\mapsto \ket{F(x)}= \frac{1}{\sqrt{g}} \sum_{k=0}^{g-1} \omega^{-kx}\ket{k}} \,.$$
If we record outcomes $z_\alpha \in \{ 0, \dots,g-1 \}$ from measuring the ancilla qudits, labeled by $\alpha$, the state after the measurement will be $$V'' = \frac{1}{\sqrt{c}} \sum_j m_j \omega^{- j \sum_\alpha z_\alpha} V_j \otimes [ \bigotimes_\alpha \ket{F(z_\alpha)}_A],$$ from which we can extract $V'''=\frac{1}{\sqrt{c}} \sum_j m_j \omega^{- j \sum_\alpha z_\alpha} V_j$, as the state of the ancillas is now disentangled and can be discarded. 

The goal is now to revert the unwanted phases to obtain the required $V = \frac{1}{\sqrt{c}} \sum_j m_j V_j$. To this purpose, consider the individual block tensors as maps $\widetilde{\mathcal{V}}_j$ from $\mathbb{C}^{dD^2}$ to $\mathbb{C}^{d\chi}$:
\begin{equation}
    \widetilde{\mathcal V}_j=
\begin{tikzpicture}[scale=\defaultscaling, baseline={([yshift=-0.65ex] current bounding box.center)}]
    \pgfmathsetmacro{\leglength}{0.6}
    \pgfmathsetmacro{\squarehalf}{0.4}
    \IsomTensorMap{0,0}{\leglength}{\squarehalf}{${\cal V}_j$}
\end{tikzpicture},
\ee
Define the local orthogonal projectors on the image of the individual blocks as $P^j_\text{loc}=\widetilde{\mathcal{V}}_j \widetilde{\mathcal{V}}_j^{-1}$, where $-1$ denotes the (Moore-Penrose) pseudoinverse. With them, we further define the projector operator over a single site $P = \sum_j P^j_\text{loc}$ and  $ {O = (\sum_j \omega^{zj} P^j_\text{loc} + \Id - P)\otimes \Id^{\otimes N-1}}$, where we denote $z = \sum_\alpha z_\alpha$. This operator is clearly local (it acts non-trivially only over a the first site) and unitary, as $O^{\dagger} O = (\Id_{\text{Im} P} + \Id_{\text{Im} (\Id - P)}) \otimes \Id^{\otimes N-1} = \Id^{\otimes N}$, where we used the strong orthogonality of the block tensors (\cref{sMPIstructure}), implying $P^j_\text{loc} P^k_\text{loc} = \delta_{jk} P^j_\text{loc}$. Applying $O$ to the isometry, we obtain $O V'''=V$, which is the isometry we sought for.
\end{proof}

Summarizing, the protocol for implementing a sMPI $V$ on a system of $N$ qudits works as follows:
\begin{itemize}
    \item Prepare, using the procedure of Ref.~\cite{Piroli_2024}, a state $\ket{\mathrm{GHZ}} = \frac{1}{\sqrt{c}} \sum_{j=0}^{g-1} m_j \ket{j}_A^{\otimes N}$ on an ancillary system $A$ consisting of $N$ $g$-dimensional qudits distributed across the system. This can be done deterministically with local operations in constant depth with a single round of measurements and feedforward.
    \item Feed this to the control space of the unitary $CU$, and also use the input $\ket{0}$ to the two middle legs in \cref{eq:BandBTilde}).
    \item Measure the ancilla qudits in the Fourier Basis and record the outcomes. Discard the ancillas, which decouple after the measurements.
    \item Apply the operator O, which is a phase unitary at a single site, but depends on all classical outcomes of the previous measurement step.
\end{itemize}

\section{sMPIs as MPUs with GHZ input}

\cref{SMapproach1} addresses the implementation of a sMPI as a finite-depth circuit assisted by local measurements and corrections. On top of that, however, a more structural perspective might be possible, providing a characterization of sMPIs in terms of simpler objects, namely MPS and MPUs.
To illustrate this idea, consider the following example.

\begin{ex}
    Take a homogeneous MPU whose single-site tensor has bipartite input and output space, and feed an (un-normalized) generalized GHZ state $\ket{\mathrm{GHZ}} = \sum_{j=0}^{g-1} m_j \ket{j}^{\otimes N}$ into one of the input legs:
    \begin{equation}
        \begin{tikzpicture}[scale=\defaultscaling,baseline={(0,-0.1)}]
        \pgfmathsetmacro{\squarehalf}{0.4}
        \MPOTensorDoubleLegsWithInputGHZ{6*\leglength + 6*\squarehalf,0}{\leglength}{\squarehalf}{${\cal U}$}{0}
    \end{tikzpicture}
    \quad
    \to
    \quad
    \begin{tikzpicture}[scale=\defaultscaling,baseline={(0,-0.1)}]
        \pgfmathsetmacro{\squarehalf}{0.4}
        \MPOTensorDoubleLegsWithInputGHZ{6*\leglength + 6*\squarehalf,0}{\leglength}{\squarehalf}{${\cal U}$}{0}
        \draw (6*\leglength + 6*\squarehalf + 0.7, - 1.35*\leglength) node {\scriptsize $j$};
    \end{tikzpicture}
    =
    \begin{tikzpicture}[scale=\defaultscaling,baseline={(0,-0.1)}]
        \pgfmathsetmacro{\squarehalf}{0.4}
        \IsomTensor{6*\leglength + 6*\squarehalf,0}{\leglength}{\squarehalf}{${\cal V}_j$}{0}
    \end{tikzpicture}
    \end{equation}
    The resulting tensor is a sMPI generating tensor, where the branches of the GHZ state define the individual blocks. Notice that the presence of the GHZ inherently provides the strong orthogonality of the blocks, as any $\mathbb E_{jj'}$ or $\mathbb S_{jj'}$, as defined in \cref{Edecomp} and \cref{Sdecomp}, will vanish after blocking for any $j\ne j'$. This can be explicitly verified by using the fact that every homogeneous MPU can be decomposed, after blocking, into a depth-2 quantum circuit~\cite{cirac_matrix_2017-1}.
\end{ex}

The importance of this example is that it suggests a general correspondence might be possible: not only does an MPU combined with a GHZ state yield a sMPI, but conversely, any sMPI could possibly be expressed in the MPU with input GHZ form. The plausibility of the reverse direction being true lies in the fact that the GHZ encodes naturally the block form and the strong orthogonality. In the following, we demonstrate this equivalence through an alternative implementation approach, with one important caveat: the MPU must include a boundary tensor. We say an MPO has homogeneous bulk and a boundary (boundary-hMPO) if all its local tensors are identical, up to a single boundary tensor that may differ from the rest \cite{styliaris2025matrix}. Graphically:
\begin{equation}
    O = \begin{tikzpicture}[scale=\defaultscaling,baseline={([yshift=-0.65ex] current bounding box.center)}]
            \pgfmathsetmacro{\leglength}{0.5}
            \pgfmathsetmacro{\squarehalfside}{0.4}
            \SingleTrLeft{(-4*\leglength,0)}
            \draw [very thick] (-4*\leglength,0) to  (0,0);
            \MPOTensor{0,0}{\leglength}{\squarehalfside}{${\cal O}$}{0}
            \MPOTensor{2*\leglength + 2*\squarehalfside,0}{\leglength}{\squarehalfside}{${\cal O}$}{0}
            \SingleDots{4*\leglength + 4*\squarehalfside, 0}{2*\leglength}
            \MPOTensor{6*\leglength + 6*\squarehalfside,0}{\leglength}{\squarehalfside}{${\cal O}$}{0}
            %\draw [very thick] (11*\leglength,0) to  (14*\leglength,0);
            \SingleTrRight{(8*\leglength + 7.5*\squarehalfside,0)}
            \filldraw[color=black, fill=whitetensorcolor, thick] (-2.5*\leglength,0) circle (0.5);
            \node at (-2.5*\leglength,0) {\scriptsize $\lambda$};
    \end{tikzpicture}.
\end{equation}

Our result is:
\begin{thm}[(Formal version) sMPI = boundary-hMPU + GHZ] \label{SMthm:boundary}
Let $\Sigma$ be an $N$-qudit system and $A_2$ a corresponding set of $N$ ancillas of local dimension $g$. 
Any (normalized) $g$-hMPI $V$ acting on $\Sigma$ can be implemented by applying a suitable MPU with homogeneous bulk and boundary on the composite system $\Sigma \otimes A_2$, where the ancillas are prepared in the generalized GHZ state $\ket{\mathrm{GHZ}}_{A_2}=\frac{1}{\sqrt{c}}\sum_{j=0}^{g-1}m_j \ket{j}_{A_2}^{\otimes N}.$ After the action of the MPU, the ancillas decouple from $\Sigma$, and tracing out $A_2$ yields $V$ on $\Sigma$.
\end{thm}

Our goal is to construct a quantum circuit that implements the isometry $V = \frac{1}{\sqrt{c}}\sum_{j=0}^{g-1} m_j V_j$ on an $N$-qubit system $\Sigma$. In contrast to \cref{SMapproach1}, we now require the entire operation to be unitary, given that we seek for an MPU, meaning that measurements and feedforward are not allowed. This is precisely where \textit{amplitude amplification} proves useful, as it enables us to obtain the desired isometry with certainty, effectively eliminating unwanted components step by step, via unitary operations only (\cref{lem:amp_amp}).

To prove our result, we will need the following lemmas.

\begin{lem}
\label{MPOlincomb}
    Any linear combination $\sum_i \lambda_i O_i$, $\lambda_i \in \mathbb{C}$ of MPOs $O_i$ of bond dimension $D_i$ can be expressed as a single MPO $\widetilde{O}$ in block form and with a boundary $\Lambda$ containing the coefficients of the combination, such that $$ \widetilde{\mathcal O}^{jk} = \bigoplus_i^{} {\mathcal O}^{jk}_i, \qquad \Lambda = \bigoplus_i \lambda_i \Id_{D_i}.$$
\end{lem}
\begin{proof}
    We start from a linear combination $\sum_i \lambda_i O_i$, $\lambda_i \in \mathbb{C}$. For simplicity, assume periodic boundary conditions and homogeneity, although neither assumption is essential for the construction. Define the MPO $\widetilde{O}$ with single-site tensors fulfilling $\widetilde{\mathcal O}^{jk} = \bigoplus_i^{} {\mathcal O}^{jk}_i$, where $\mathcal O_i$ denotes the MPO single site tensor, and the matrix $\Lambda = \bigoplus_i \lambda_i \Id_{D_i}$ as the boundary. This yields the statement, as $\operatorname{Tr}\left(\Lambda \;\widetilde{\mathcal O}^{j_{1},k_{1}} \cdots \widetilde{\mathcal O}^{j_{N},k_{N}}\right) = \sum_i \lambda_i \operatorname{Tr}\left({\mathcal O}^{j_{1},k_{1}}_i \cdots {\mathcal O}^{j_{N},k_{N}}_i\right)$.
\end{proof}

\begin{lem}[Subspace amplitude amplification \cite{gilyen2019quantum,Berry_2014,styliaris2025quantumcircuitcomplexitymatrixproduct}] \label{lem:amp_amp}
    Let $U$ be a unitary on $\mathcal H_S \otimes \mathcal H_{A_2} = (\mathbb{C}^d)^{\otimes N} \otimes (\mathbb{C}^{g})^{\otimes N}$, $M$ an operator acting on $\mathcal H_S$, and $\mathcal S$ a subspace of $\mathcal H_S$. Denote $P_{\mathcal S^\perp}$ the orthogonal projector onto the orthogonal complement of $\mathcal S$. Denote $\ket{\mathrm{GHZ}}_{A_2} = \sum_{i=0}^{d'-1} \ket{i}^{\otimes N}$. Suppose that:
    \begin{enumerate}[(i)]
        \item $\bra{\psi'} M^\dagger M \ket{\psi} = \bra{\psi'} \psi \rangle$ for all $\ket{\psi},\ket{\psi'} \in \mathcal S$.
        \item For all $\ket{\psi} \in \mathcal S$,
        \begin{align}
            U (\ket{\psi}_S\ket{\mathrm{GHZ}}_{A_2}) = \sin\theta \ket{\Phi} + \cos\theta \ket{\Phi^\perp}
        \end{align}
        with $\sin\theta\ne 0$, $\cos\theta \ne 0$. Here $\ket{\Phi} \coloneqq (M\ket{\psi}_S) \ket{0}_{A_2}^{\otimes N}$ satisfying $\bra{\Phi^\perp} \Phi \rangle = 0$ and $_{A_2}^{\otimes N}\bra{0} \Phi^\perp \rangle = 0$.
    \end{enumerate}
    Denoting $\ket{\Psi} \coloneqq \ket{\psi}_S\ket{0}_{A_2}^{\otimes N}$, define its orthogonal vector via
    \begin{align}
        U \ket{\Psi^\perp} \coloneqq \cos\theta \ket{\Phi} - \sin \theta \ket{\Phi^\perp}
    \end{align}
    which depends on $\ket{\psi}$. Then:
    \begin{enumerate}[(a)]
        \item For all $\ket{\psi}, \ket{\psi'} \in \mathcal S$,
        \begin{align} \label{eq:orth_psi_perp}
        _S\bra{\psi'}(_{A_2}^{\otimes N}\bra{0}\Psi^\perp \rangle)  = 0 \;.
        \end{align}
        \item For all $\ket{\psi} \in \mathcal S$, the unitary
    \begin{align}
        R_{\Phi} & \coloneqq \1_S \otimes 2 (\ket{0}_{A_2}\bra{0})^{\otimes N} - \1
    \end{align}
    reflects along $\ket{\Phi}$ on the $\{\ket{\Phi}, \ket{\Phi^\perp}\}$ subspace, while
    \begin{align}
        R_{\Psi} & \coloneqq \left( \1_S \otimes 2(\ket{0}_{A_2} \bra{0})^{\otimes N} - \1 \right) \left( \1 - P_{\mathcal S^\perp} \otimes 2 (\ket{0}_{A_2}\bra{0})^{\otimes N} \right)
    \end{align}
    reflects along $\ket{\Psi}$ on the $\{\ket{\Psi}, \ket{\Psi^\perp}\}$ subspace.

    In other words, the reflections are such that $$\begin{aligned}
R_{\Phi}\ket{\Phi} &= \ket{\Phi},\\
R_{\Phi}\ket{\Phi^\perp} &= -\ket{\Phi^\perp},
\end{aligned}$$ and $$\begin{aligned}
R_{\Psi}\ket{\Psi} &= \ket{\Psi},\\
R_{\Psi}\ket{\Psi^\perp} &= -\ket{\Psi^\perp}.
\end{aligned}$$
    \end{enumerate}
\end{lem}

With these at our disposal, we can now prove \cref{SMthm:boundary}.

\begin{proof}
    As explained, the goal is to find a quantum circuit that acts as $V = \frac{1}{\sqrt{c}}\sum_{j=0}^{g-1} m_j V_j$ on an N-qubit system $\Sigma$. The overall protocol needs to be, unlike in \cref{SMapproach1}, a unitary, so that measurements are not allowed. 
    The first steps resemble the previous protocol. Remember that, in that construction, the isometric tensors were extended to unitary gates (\cref{eq:BandBTilde}), yielding unitary circuits overall. In other words, the $V_j$ are turned into unitaries $U_j$. In particular, the bottom tensors are extended to unitary gates via two ancillas $A_1$ of dimension $d\chi^2$.
    Phrased differently, this means that the operator $M =\frac{1}{\sqrt{c}}\sum_j m_j U_j $ is an isometry on the subspace $\mathcal S$ defined by fixing the first set of ancillas to $\ket{0}_{A_1}$: $\mathcal{S} = \{\ket{\psi} \otimes \ket{0}_{A_1}^{\otimes N}\}_{\ket{\psi} }$. The goal is now to implement $M$, aided by the set of $N$ ancillas $A_2$, of local dimension $g$, i.e., of local dimension equal to the number of independent MPIs in the sMPI decomposition.
    In Ref.~\cite{gilyen2019quantum}, it was first shown how to circuitally implement an isometry such as $M$, by means of a purposefully crafted unitary and an orthogonal projector on the image of the target isometry. However, here we follow a different but equivalent approach, presented in Ref.~\cite{styliaris2025quantumcircuitcomplexitymatrixproduct}. To that aim, we construct the controlled unitary $CU$ as in the first approach (see \cref{controlunitary}), using an $N$-ancilla system $A_2$ for the control qudits.
We then input a generalized GHZ of the form $\ket{\mathrm{GHZ}} = \frac{1}{\sqrt{c}} \sum_{j=0}^{g-1} m_j \ket{j}_{A_2}^{\otimes N}$ on the control qudits, yielding $${CU \ket{\mathrm{GHZ}}_{A_2} = \frac{1}{\sqrt{c}} \sum_{j=0}^{g-1} m_j U_j \otimes \ket{j}_{A_2}^{\otimes N}}.$$
We then act on each of the $N$ $A_2$ ancillas with the inverse Fourier Transform defined by the inverse of \\ ${\ket{x}\mapsto \ket{F(x)}= \frac{1}{\sqrt{d}} \sum_{k=0}^{d-1} \omega^{-kx}\ket{k}}$, where $\omega=e^{2\pi i/d}$. In our case, $F \ket{0}_{A_2} \coloneqq \frac{1}{\sqrt{g}}\sum_{j=0}^{g-1} \ket{j}_{A_2}$. Ultimately we define $U \coloneqq [\Id_{\Sigma} \otimes (F^\dagger)^{\otimes N}](CU)$. Now, for any $\ket{\psi}_{S}$ in the previously defined $\mathcal S$, a direct calculation gives
\begin{align} \label{eq:aa_subspace}
    U \ket{\psi}_{S} \ket{\mathrm{GHZ}}_{A_2} = \frac{1}{\sqrt{cd}} \ket{\Phi} + \sqrt{1 - \frac{1}{cd}} \ket{\Phi^\perp} \;,
\end{align}
where $\ket{\Phi} \coloneqq (M\ket{\psi}_{S}) \ket{0}_{A_2}^{\otimes N}$ such that $_{A_2}^{\otimes N}\langle 0 \ket{\Phi^\perp} = 0$. 
In Ref.~\cite{styliaris2025quantumcircuitcomplexitymatrixproduct}, it is shown how the output state can be worked upon by successive rounds of amplitude amplification, specifically by rotations in $\text{span}\{\ket{\Phi}, \ket{\Phi^\perp}\}$, to eliminate the component orthogonal to $\ket{0}_{A_2}^{\otimes N}$, i.e, $\ket{\Phi^\perp}$, which is associated to an action on the system different than the required $\frac{1}{\sqrt{c}}\sum_j m_j V_j$. In particular, the relevance of Ref.~\cite{styliaris2025quantumcircuitcomplexitymatrixproduct} lies in the correspondence to our setting: the procedure works in a \textit{subspace} $\mathcal S$ of the full Hilbert space, where it can be guaranteed that $M$ is an isometry. Previous work \cite{Berry_2014} dealt only with oblivious amplitude amplification in the full space.

By applying the generalization of oblivious amplitude amplification of \cref{lem:amp_amp}, we now define the following operators:

\begin{itemize}
    \item $R_{\Phi} \coloneqq \1_S \otimes 2 (\ket{0}_{A_2}\bra{0} )^{\otimes N}- \1$
    \item $R_{\Psi} \coloneqq \left( \1_S \otimes 2(\ket{0}_{A_2}\bra{0} )^{\otimes N} - \1 \right) \left( \1 - P_{\mathcal S^\perp} \otimes 2(\ket{0}_{A_2}\bra{0} )^{\otimes N} \right) =  P_S \otimes 2(\ket{0}_{A_2}\bra{0} )^{\otimes N} - \1 $ 
    \item $G \coloneqq - U R_{\Psi} U^\dagger R_{\Phi}$.
\end{itemize}
Denoting $\frac{1}{\sqrt{cd}} $ as $\sin \theta $, we straightforwardly get: 
\begin{equation}
    G^\ell U (\ket{\psi}_S\ket{\mathrm{GHZ}}_{A_2}) = G^\ell\left( \sin\theta \ket{\Phi} + \cos\theta \ket{\Phi^\perp} \right) \nonumber \\
         = \sin[(2\ell +1)\theta] \ket{\Phi} + \cos[(2\ell +1)\theta] \ket{\Phi^\perp} \;,
\end{equation}
which becomes, for $\ell$ satisfying $(2\ell + 1)\theta = \pi/2 $:
\begin{equation}
\label{rotation}
    G^\ell U (\ket{\psi}_S\ket{\mathrm{GHZ}}_{A_2}) = (M\ket{\psi}_{S}) \ket{0}_{A_2}^{\otimes N}.
\end{equation}
This means that after applying $\ell$ amplitude amplification rotations, the ancilla decouples and we are left with the action of $\frac{1}{\sqrt{c}}\sum_j m_j V_j$ on the first system (remember that $U_j \ket{0}_{A_1}^{\otimes N} = V_j$). Notice that if no integer $\ell$ satisfies $(2\ell + 1)\theta = \pi/2 $, we could redefine $\frac{1}{\sqrt{cd}} $ by adding an ancillary qubit and rotating it by $\phi$: $M' = M \otimes (\cos\phi \; \Id + i \sin \phi \;Z)$, leading to a prefactor of $\frac{1}{\sqrt{cd} \; (|\cos\phi| + |\sin \phi|) }$. Adjusting $\phi$ leads to an integer $\ell$.

We constructed a unitary $ U_{\mathrm{total}} = G^\ell U $ acting on system endowed with ancillas $\Sigma \otimes A_2$.
By preparing the registers $A_2$ in a generalized GHZ state and contracting those legs,
the induced map on $\Sigma$ is precisely $V$.
Since the ancilla systems are decoupled at the end of the process,
they can be discarded without affecting the action on $\Sigma$.

What is left to show is that $U_{\text {total}}$ is an MPU with homogeneous bulk and non-trivial boundary. This is indeed the case. $U$ is an MPU (with periodic boundary conditions) essentially because the $V_j$ are by hypothesis; its bulk single-site tensor is given by: 
\be
\pgfmathsetmacro{\spacing}{1.5cm}
\begin{tikzpicture}[scale=\defaultscaling,x=\spacing, y = \spacing, baseline={([yshift=-0.65ex] current bounding box.center)}]
    \definecolor{cream}{RGB}{255,255,221}
    \pgfmathsetmacro{\height}{1.5}
    \pgfmathsetmacro{\shift}{0.7}
    \pgfmathsetmacro{\labelpositions}{\height / 2 + \shift / 2 + 0.2}
    \pgfmathsetmacro{\labelcloseness}{0.15}
    \pgfmathsetmacro{\extra}{.15}

    \clip (2.35,-1.5) rectangle (5.65,1.5);
    
    \foreach \i in {0,1,3,4,6,7,9} {
        \draw[color=black,very thick] (\i,\shift) -- (\i,\height);
    }
    \foreach \i in {3} {
        \draw[color=purification,very thick] (\i,\shift) -- (\i,\height);
    }
    \foreach \i in {1.5, 4.5, 7.5} {
        \draw[color=black,very thick] (\i,\height) -- (\i,-\height);
        \fill (\i,\shift) circle (3pt);
        \fill (\i,-\shift) circle (3pt);
        
    }
    \foreach \i in {0,2,3,5,6,8,9} {
        \draw[color=black,very thick] (\i,-\shift) -- (\i,-\height);
        %\draw (\i - \labelcloseness, -\labelpositions) node {\scriptsize $d$};
    }
    \foreach \i in {0,3,6,9} {
        \draw[color=black,very thick] (\i,-\shift) -- (\i,\shift);
        \draw[color=black,very thick] (\i + \extra,\shift) -- (\i + \extra,\height);
    }
    \foreach \i in {1,4,7} {
        %\draw[color=black,very thick] (\i,-\shift) -- (\i,\shift);
        \draw[color=purification,very thick] (\i - \extra,\shift) -- (\i - \extra,\height);
    }

    \foreach \i in {1,4,7} {
            \draw[color=black,very thick, rounded corners=8pt]
    (\i, \shift) -- (\i, 0) -- (\i+1, 0) -- (\i+1, -\shift);
        }
    \foreach \i in {0,3,6} {
        \draw[color=black,very thick] (\i,\shift) -- (\i+1.5,\shift);
        \draw[thick, fill=cream, rounded corners=2pt]
            (\i - \extra, -0.2 + \shift) rectangle (\i+1+\extra, 0.2 + \shift);
        \draw (\i+0.5,\shift) node {\scriptsize $T$};
    }
    \foreach \i in {-1,2,5,8} {
        \draw[color=black,very thick] (\i,-\shift) -- (\i+2.5,-\shift);
        \draw[thick, fill=cream] (\i, -\shift) circle [radius=0.3];
        \draw[thick, fill=cream] (\i+1, -\shift) circle [radius=0.3];
        \draw (\i,-\shift) node {\scriptsize $B_\ell$};     
        \draw (\i+1,-\shift) node {\scriptsize $B_r$};     
    }
\end{tikzpicture},
\ee
where $B_\ell$ and $B_r$ are obtained via operator Schmidt decomposition of $B$. The extra output of the top gate (i.e., the inner grey legs in \cref{unitextension}) is not shown, as it is decoupled and is hence discarded. 

Each amplitude amplification rotation $G$ also can be written as an MPU, this time with a boundary, as it is the composition of $U$, $U^\dagger$ and the two reflections $R_{\Phi}$ and $R_{\Psi}$, which are two unitary linear combinations of MPOs, and hence can be expressed as an MPU with boundary themselves (\cref{MPOlincomb}). For $R_{\Psi}$, we have to show that is indeed a linear combination of MPOs, and in particular that $P_{\mathcal S^\perp}$ can be expressed as such. This is true, and follows from the actual implementation of $P_{\mathcal S^\perp}$: if the isometry defining the subspace $\mathcal S$ is implemented via a unitary acting on the system together with auxiliary qudits initialized in $\ket{0}_{A_1}^{\otimes N}$, then the action of the projector $P_{\mathcal{S}^\perp}$ can be realized by first applying the inverse of the unitary, projecting onto the subspace where the auxiliary qudits are not in $\ket{0}_{A_1}^{\otimes N}$, and finally reapplying the unitary. In our case, the isometry defining $\mathcal{S}$ is simply: $\Id^{\otimes N} \otimes \ket{0}_{A_1}^{\otimes N}$, which can clearly be realized in a local manner. 

As a consequence, $U_{\text {total}}$ can also be written in this way, given that all of its factors are MPUs. It is easy to see that the bond dimension of the MPUs in question remain independent of the system size, because extending them to additional sites amounts to replicating an identical bulk tensor with size-independent bond dimension.
Then, inputting the $ \ket{\mathrm{GHZ}}_{A_2}$ state on $U_{\text {total}}$ yields the statement of the proposition, as highlighted by \cref{rotation}.
\end{proof}

As a side remark, the MPS, i.e., the $\mathrm{GHZ}$, can always be chosen to be homogeneous, such that the only boundary dependence resides in the MPU. Given the definition of generalized GHZ state as $\ket{\mathrm{GHZ}} = \sum_i c_i \ket{i}^{\otimes N}$, with $c_i = n_i/m \in \mathbb{Q}$, we can indeed represent such a state by introducing a boundary. However, if, instead of introducing the boundary, one repeats each block 
$r_i$ times, the resulting GHZ state has coefficients $c'_i = r_i^2$. Hence, up to a global 
normalization, the coefficients $c_i$ can be approximated arbitrarily well by the $c'_i$. 
In particular, GHZ states with arbitrary real boundary coefficients can be approximated 
arbitrarily well by a system-size independent normalization constant and a homogeneous 
bulk GHZ state of larger bond dimension. The conclusion is that the MPS can always be made homogeneous, and the boundary is only contained in the MPU.

\section{Locality and Causality Preservation under Matrix Product Isometries}
We just saw how the inclusion of a normalization factor might lead to the creation of long-range correlations. In terms of the previously defined classes, we saw how the elements of {\sffamily{sMPI}} might create long-range correlations. Reconsider the previously described \cref{longrangeisometry}, which we now want to analyze under a different light. We had: 
\[
\mathcal{E}(\cdot) = \Tr(\cdot) \, \ket{\mathrm{GHZ}}\bra{\mathrm{GHZ}}.
\] 
As we mentioned, this channel clearly generates long-range entanglement and hence does not fit within the homogeneous framework. Nevertheless, the example remains physically relevant, as it highlights an interesting scenario — namely, channels that are \emph{Causality Preserving} but not \emph{Locality Preserving}. We recall the relevant definitions from Ref.~\cite{piroli_quantum_2020}, where these concepts were first adapted to channels (see also Ref.~\cite{arrighi_overview_2019}).

Generally, Causality Preserving Quantum Channels (CPQCs), with corresponding class {\sffamily{CPQC}}, are those for which the evolution of an observable localized at site $x$ depends only on the initial state in a neighborhood of $x$. In the Heisenberg picture, this is expressed by the adjoint map $\mathcal{E}^\dagger$: for any observable $O_x$ supported on site $x$, there exists an operator $O_{\mathcal{N}(x)}$ supported on a neighborhood $\mathcal{N}(x)$ of $x$ such that
\be
\mathcal{E}^\dagger(O_x) = O_{\mathcal{N}(x)}.
\ee
This ensures that the expectation value of $O_x$ on the evolved state is fully determined by the restriction of the initial state to $\mathcal{N}(x)$.
QCA channels are the unitary version of CPQC. 

A quantum channel $\mathcal{E}$ is called a Locality Preserving Quantum Channel (LPQC), contained in the class {\sffamily{LPQC}}, if, for every region $A$ of the lattice (with $a$ denoting its neighborhood and $\bar{A} = A \cup a$, and with analogous definitions for the two blocks of the complement $\bar{B}$ of $\bar{A}$), and for every positive semidefinite $\rho_{\bar{A}}, \rho_{\bar{B}}$ with support in the regions indicated by the corresponding subscript, the following relation holds:
\be
    \Tr_{a,b}\!\left[ \mathcal{E}({\rho_{\bar{A}} \rho_{\bar{B}}}) \right] = \frac{1}{d^N}\, \Tr_{a,\bar{B}}\!\left[ \mathcal{E}(\rho_{\bar{A}}) \right] 
   \Tr_{\bar{A},b}\!\left[ \mathcal{E}(\rho_{\bar{B}}) \right].
\ee
This condition expresses that when $\mathcal{E}$ acts on a product state, it does not generate correlations between the disjoint regions $A$ and $B$, except possibly in the immediate neighbourhood of their boundary.

\cref{longrangeisometry} is causality preserving, but not locality preserving, because this construction produces long-range correlations, reflecting that the GHZ tensor is not simple. The causal light cone on observables, i.e., the causality, is however still preserved, meaning the tensor is an isometric QCA. This is easily shown by considering the adjoint channel's action on a local observable $O$:
\begin{equation}
\label{lightcone}
\mathcal{E}^\dagger (O) 
    \hspace{0.3cm}
=     
\hspace{0.3cm}
    \begin{tikzpicture}[scale=0.6,baseline={([yshift=-0.65ex] current bounding box.center)}] 
        \definecolor{cream}{RGB}{255,255,221}
        \pgfmathsetmacro{\squarehalf}{0.4}
        \pgfmathsetmacro{\halfunit}{\squarehalf+\leglength}
        \pgfmathsetmacro{\extrafat}{0.1}
        \SingleTrLeft{(-4*\leglength,0)} 
        \draw [very thick] (11*\leglength,0) to (16*\leglength,0); 
        \SingleTrRight{(16*\leglength,0)} 
        \draw [very thick] (-4*\leglength,0) to (0,0); 
        \GHZTensor{0,0}{\leglength}{\squarehalf}{0} 
        \GHZTensor{2*\halfunit,0}{\leglength}{\squarehalf}{0} 
        \SingleDots{4*\halfunit, 0}{2*\leglength} 
        \GHZTensor{6*\halfunit,0}{\leglength}{\squarehalf}{0}
        \draw[very thick]  (\halfunit,\halfunit) .. controls (\halfunit,0) and (0,-0.5*\halfunit) .. (0,-1.5*\halfunit);
        \draw[very thick]  (3*\halfunit,\halfunit) .. controls (3*\halfunit,0) and (2*\halfunit,-0.5*\halfunit) .. (2*\halfunit,-1.5*\halfunit);
        \draw[very thick]  (7*\halfunit,\halfunit) .. controls (7*\halfunit,0) and (6*\halfunit,-0.5*\halfunit) .. (6*\halfunit,-1.5*\halfunit);
          \begin{scope}[yscale=-1, shift={(0,-3*\halfunit)}]
            \SingleTrLeft{(-4*\leglength,0)} 
            \draw [very thick] (11*\leglength,0) to (16*\leglength,0); 
            \SingleTrRight{(16*\leglength,0)} 
            \draw [very thick] (-4*\leglength,0) to (0,0); 
            \GHZTensor{0,0}{\leglength}{\squarehalf}{0} 
            \GHZTensor{2*\halfunit,0}{\leglength}{\squarehalf}{0} 
            \SingleDots{4*\halfunit, 0}{2*\leglength} 
            \GHZTensor{6*\halfunit,0}{\leglength}{\squarehalf}{0}
            \draw[very thick]  (\halfunit,\halfunit) .. controls (\halfunit,0) and (0,-0.5*\halfunit) .. (0,-1.5*\halfunit);
            \draw[very thick]  (3*\halfunit,\halfunit) .. controls (3*\halfunit,0) and (2*\halfunit,-0.5*\halfunit) .. (2*\halfunit,-1.5*\halfunit);
            \draw[very thick]  (7*\halfunit,\halfunit) .. controls (7*\halfunit,0) and (6*\halfunit,-0.5*\halfunit) .. (6*\halfunit,-1.5*\halfunit);
          \end{scope}
        \foreach \i in {0,1,6,7} {
            \IdentityLine{\i*\halfunit,1.5*\halfunit}{\halfunit};
        }
        \draw[ thick, fill=cream, rounded corners=4pt] (2*\halfunit - \extrafat,\halfunit - \extrafat) rectangle (3*\halfunit +\extrafat,2*\halfunit+\extrafat);
        \draw (2.5*\halfunit,1.5*\halfunit) node {\scriptsize $O$};
    \end{tikzpicture}
    \hspace{1cm}
    =
    \hspace{1cm}
\begin{tikzpicture}[scale=0.6,baseline={([yshift=-0.65ex] current bounding box.center)}] 
        \definecolor{cream}{RGB}{255,255,221}
        \pgfmathsetmacro{\squarehalf}{0.4}
        \pgfmathsetmacro{\halfunit}{\squarehalf+\leglength}
        \pgfmathsetmacro{\extrafat}{0.1}
        \IdentityLine{-\halfunit,1.5*\halfunit}{4*\halfunit}
        \GHZTensor{2*\halfunit,0}{\leglength}{\squarehalf}{0} 
        \draw[very thick]  (3*\halfunit,1.5*\halfunit) -- (3*\halfunit,-0.5*\halfunit);
        \SingleDots{5*\halfunit, 0}{2*\leglength};
          \begin{scope}[yscale=-1, shift={(0,-3*\halfunit)}]
                \GHZTensor{2*\halfunit,0}{\leglength}{\squarehalf}{0} 
                \draw[very thick]  (3*\halfunit,1.5*\halfunit) -- (3*\halfunit,-0.5*\halfunit);
                \SingleDots{5*\halfunit, 0}{2*\leglength};
          \end{scope}
        %\foreach \i in {0,1,6,7} {
        %    \IdentityLine{\i*\halfunit,1.5*\halfunit}{\halfunit};
        %}
        \IdentityLine{6.5*\halfunit,1.5*\halfunit}{4*\halfunit}
        \IdentityLine{7.5*\halfunit,1.5*\halfunit}{4*\halfunit}
        \draw[ thick, fill=cream, rounded corners=4pt] (2*\halfunit - \extrafat,\halfunit - \extrafat) rectangle (3*\halfunit +\extrafat,2*\halfunit+\extrafat);
        \draw (2.5*\halfunit,1.5*\halfunit) node {\scriptsize $O$};
        \draw[very thick] (3*\halfunit,0) .. controls (4*\halfunit,0) and (4*\halfunit,3*\halfunit) .. (3*\halfunit,3*\halfunit);
        \draw[very thick] (\halfunit,0) .. controls (0,0) and (0,3*\halfunit) .. (\halfunit,3*\halfunit);
    \end{tikzpicture}
\end{equation} 
Remember that the GHZ tensor ${\cal V}$ generates a bona fide isometry only after a system-size independent normalization, implying that the overall tensor is not homogeneous. In Ref.~\cite{piroli_quantum_2020}, it is shown that, for unitary channels, {\sffamily{CPQC}} (that is, usual QCA expressed as unitary channels) coincide with {\sffamily{LPQC}}.
If we define isometric QCA as {\sffamily{CPQC}} implemented by an isometry (i.e.,\ $\mathcal{E}(\rho) = V\rho V^\dagger$, where ${V}$ is an isometry), this begs the question of whether isometric {\sffamily{QCA}} = isometric {\sffamily{LPQC}}, as in the unitary case. This is not the case, as \cref{longrangeisometry} shows. Hence, while unitary QCAs coincide with homogenous MPUs, not all isometric QCAs can be represented as homogeneous MPIs.

As an addendum, the example also shows that {\sffamily{CPQC}} is not the convex hull of {\sffamily{LPQC}}, as the channel is extremal, and cannot be expressed as convex combination of {\sffamily{LPQC}}, contrary to what hypothesized in Ref.~\cite{piroli_quantum_2020}.

Finally, we denote by {\sffamily{dQC}} the set of CPQC that are obtained by a Stinespring dilation in terms of a QCA \cite{piroli_quantum_2020}
$\mathcal{E}_u: L(\mathcal{H} \otimes \mathcal{H}) \rightarrow L(\mathcal{H} \otimes \mathcal{H})$ as
$$
\mathcal{E}(\rho)=\operatorname{tr}_{V^{\prime}}\left[\mathcal{E}_u\left(\rho \otimes(|0\rangle\langle 0|)^{\otimes N}\right)\right],
$$
where $|0\rangle$ is a state of the ancilla qudits, and by {\sffamily{tnQC}} all causality-preserving channel whose Choi-Jamilkowski state can be written in MPO form.

Having set the stage with the definitions from Ref.~\cite{piroli_quantum_2020}, we show where the newly defined classes fall with respect to those.

\begin{prop}
The quantum channel classes considered here satisfy the inclusion relations illustrated in the diagram below:
\begin{equation}
    	\includegraphics[width=0.5\linewidth]{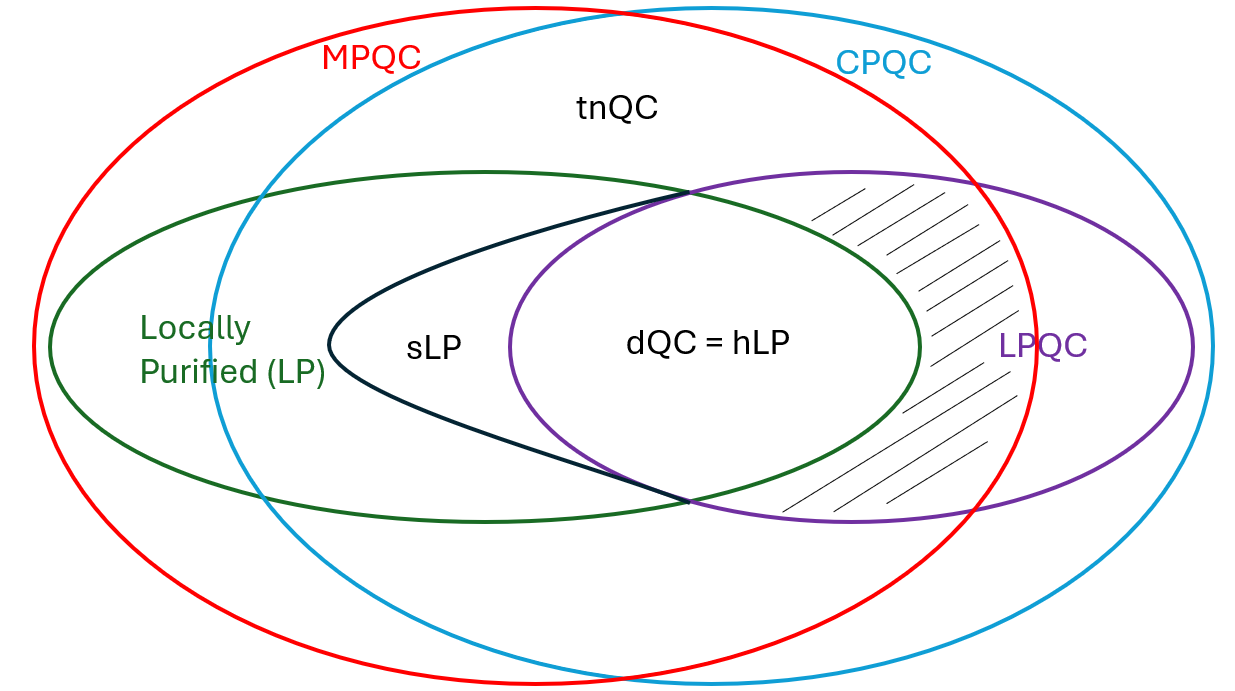}
\end{equation}
\end{prop}

\begin{proof}
We establish the non-trivial inclusions one by one. The right part of the diagram, and relative inclusions, are taken from Ref.~\cite{piroli_quantum_2020}. In particular, the shaded area refers to the open question of whether {\sffamily{dQC}} $=$ {\sffamily{tnQC}} $\cap$ {\sffamily{LPQC}}.
\begin{enumerate}[(i)]

    \item {\sffamily{LP}} $\subsetneq$ {\sffamily{MPQC}}: This is the inequivalence shown in \cref{ineqLocNTI}.

    \item {\sffamily{MPQC}} $\ne$ {\sffamily{tnQC}}: Consider a multicontrol Z gate. This can be written in tensor-network form, but it is not causality-preserving.

    \item {\sffamily{sLP}} $\subseteq$ {\sffamily{CPQC}}: By applying the decomposition from \cref{sMPIstructure}, we compute how the adjoint channel of a sMPI $V$ acts upon a given local operator:
    $$\mathcal{E}^\dagger(O) = V^\dagger O V= (\sum_i m_i V^\dagger_i) O (\sum_j m_j V_j) = \sum_{i,j} m_i m_j V^\dagger_i O V_j \stackrel{(1)}{=}
 \sum_{i} m_i^2 V^\dagger_i O V_i = \sum_{i} m_i^2 \mathcal{E}_i^\dagger(O),$$ where in $(1)$ the strong orthogonality of the $V_i$ was used. The last term is a sum over operators with the same --- bounded --- support, proving that sLPs are causality-preserving.  
    
    \item {\sffamily{hLP}} $=$ {\sffamily{dQC}}: Given that in one dimension homogeneous MPUs coincide with translationally invariant (TI) quantum cellular automata, the Stinespring dilation of any hLP channel is implemented by a TI QCA unitary (from which the hMPI is obtained by fixing the ancilla input). Consequently, {\sffamily{dQC}} $\;\subseteq\;${\sffamily{hLP}}. Conversely, the depth-two quantum circuit decomposition of a hMPI tensor (see \cref{uuvv}) shows that every hMPI can be realized as an MPU with one input fixed. Thus, we have shown the following: {\sffamily{hLP}} $\;\subseteq\;$ {\sffamily{dQC}}, and, consequently, {\sffamily{hLP}} $=$ {\sffamily{dQC}}.
I.e., in 1D, a channel is in {\sffamily{dQC}} if and only if it is in {\sffamily{hLP}}.
\end{enumerate}
\end{proof}

\end{document}